\documentclass[11pt,letterpaper, table]{article}
\usepackage[utf8]{inputenc}
\usepackage{graphicx}
\graphicspath{{figures/}}

\usepackage[letterpaper, left=1in, right=1in, top=0.9in, bottom=0.9in]{geometry}
\usepackage[american]{babel}
\usepackage[normalem]{ulem}
\usepackage{amsmath, amssymb, cases, amsthm}
\usepackage{thmtools}
\usepackage[shortlabels]{enumitem}
\usepackage{mdframed}
\usepackage{bbm}
\usepackage{bm}
\usepackage{microtype}
\usepackage{makecell}
\usepackage{mathtools}
\usepackage{float}
\usepackage{varwidth}
\usepackage{modletter}

\usepackage{tikz}
\usetikzlibrary{positioning, fit}
\usepackage{caption2}
\usepackage{tabularx}
\usepackage{tcolorbox}
\usepackage{multirow}

\newtcolorbox{construction}[2][]
{
        colback = white,
	left*=0mm,
	before skip = 10pt,
	after skip = 10pt,
	title    = \textbf{\space\space #2},
	#1,
}

\usepackage[lined,ruled,linesnumbered,noend]{algorithm2e}
\SetKwInput{KwData}{Input}
\SetKwInput{KwResult}{Output}

\usepackage[bookmarks,colorlinks,breaklinks]{hyperref}
\hypersetup{urlcolor=blue, colorlinks=true, citecolor=green!50!black, linkcolor=blue}
\usepackage[capitalize,nosort]{cleveref}
\usepackage[]{xcolor}
\declaretheorem[numberwithin=section,refname={Theorem,Theorems},Refname={Theorem,Theorems}]{theorem}

\declaretheorem[numberlike=theorem]{lemma}

\declaretheorem[numberlike=theorem]{corollary}
\declaretheorem[numberlike=theorem]{definition}
\declaretheorem[numberlike=theorem]{claim}
\declaretheorem[numberlike=theorem,style=remark]{remark}

\theoremstyle{definition}

\declaretheorem[numberlike=theorem]{assumption}

\newcommand{\tOh}[1]{\widetilde{O}\left(#1\right)}

\newcommand{\MM}{\operatorname{MM}}
\newcommand{\textdet}{\operatorname{det}}

\newcommand{\LL}{L^{\star}}

\newcommand{\congest}{$\mathsf{CONGEST}$\xspace}
\newcommand{\pram}{$\mathsf{PRAM}$\xspace}

\newcommand{\UDVC}{{\sf VC}}

\newcommand{\SSUDVC}{{\sf SSVC}}

\newcommand{\sk}[1]{sk_{#1}}
\newcommand{\ska}[1]{sk_{\approx}{#1}}
\newcommand{\nc}[1]{#1-close}

\newcommand{\cnc}{\textsc{ComNeiClustering}}
\newcommand{\textdeg}{\operatorname{deg}}
\newcommand{\rank}{\operatorname{rank}}
\newcommand{\poly}{\text{poly}}
\newcommand{\nd}{neighborhood difference}
\newcommand{\mnnccn}{{\sf MinNNCC}}
\newcommand{\mn}{{\sf Min-Neighbor}}
\newcommand{\mnnccnalg}{{\textsc MinNNCC}}
\newcommand{\isocut}{\textsc{IsolatingCuts}}
\newcommand{\BMM}{{\sf BMM}}
\newcommand{\BMC}{{\sf BMC}}
\newcommand{\eps}{\epsilon}
\newcommand{\subtri}{subset-tribes}
\DeclarePairedDelimiter{\ceil}{\lceil}{\rceil}

\begin{document}
\sloppy

\title{
Global vs.~s-t Vertex Connectivity Beyond Sequential: \\ Almost-Perfect Reductions \&  Near-Optimal Separations
}
\author{
Joakim Blikstad\thanks{CWI, the Netherlands \& KTH, Sweden, \texttt{joakim@cwi.nl}}
\and
Yonggang Jiang\thanks{MPI-INF \& Saarland University, \texttt{yjiang@mpi-inf.mpg.de}}
\and
Sagnik Mukhopadhyay\thanks{University of Birmingham, United Kingdom, \texttt{s.mukhopadhyay@bham.ac.uk}}
\and
Sorrachai Yingchareonthawornchai\thanks{ETH Zürich, Switzerland \&  the Hebrew University, Israel, \texttt{sorrachai.yingchareonthawornchai@eth-its.ethz.ch}}
}
\date{}
\maketitle \pagenumbering{roman}

\begin{abstract}
A recent breakthrough by [LNPSY STOC'21] showed that solving s-t vertex connectivity is sufficient (up to polylogarithmic factors) to solve (global) vertex connectivity in the sequential model.
This raises a natural question: What is the relationship between s-t and global vertex connectivity in other computational models?
In this paper, we demonstrate that the connection between global and s-t variants behaves very differently across computational models:
\begin{itemize}
    \item
    In parallel and distributed models,
    we obtain almost tight reductions from global to s-t vertex connectivity.
    In PRAM, this leads to a $n^{\omega+o(1)}$-work and $n^{o(1)}$-depth  algorithm for vertex connectivity, improving over the 35-year-old $\tilde O(n^{\omega+1})$-work $O(\log^2n)$-depth algorithm by [LLW FOCS'86], where $\omega$ is the matrix multiplication exponent and $n$ is the number of vertices.
    In CONGEST, the reduction implies the first sublinear-round
    (when the diameter is moderately small)
    vertex connectivity algorithm. This answers an open question in [JM STOC'23].
    \item In contrast, we show that global vertex connectivity is strictly harder than s-t vertex connectivity in the two-party communication setting, requiring $\tT{n^{1.5}}$ bits of communication. The s-t variant was known to be solvable in $\tilde O(n)$ communication~[BvdBEMN FOCS'22].  Our results resolve open problems raised by~[MN STOC'20, BvdBEMN FOCS'22, AS SOSA'23].
\end{itemize}

At the heart of our results is a new graph decomposition framework we call \emph{common-neighborhood clustering}, which can be applied in multiple models.
Finally, we observe that global vertex connectivity cannot be solved without using s-t vertex connectivity by proving an s-t to global reduction in dense graphs in the PRAM and communication models.
\end{abstract}

\clearpage
\setcounter{tocdepth}{2}
\tableofcontents

\clearpage
\pagenumbering{arabic}

\section{Introduction}

The problem of finding (global) vertex connectivity of a simple undirected graph aims to find the minimum number of vertices whose removal disconnects the graph (or the graph becomes a singleton). This problem, along with a very closely related problem of (global) edge connectivity (which asks for the minimum number of edges to be removed to disconnect the graph), are classic and fundamental graph problems that have drawn the attention of researchers for the last fifty years \cite{Kleitman1969methods, Podderyugin1973algorithm, EvenT75, Even75, Galil80, EsfahanianH84, Matula87, BeckerDDHKKMNRW82, LinialLW88, CheriyanT91, NagamochiI92, CheriyanR94, Henzinger97, HenzingerRG00, Gabow06, Censor-HillelGK14, SaranurakY22}.
Both problems can be easily solved by $n^2$ calls to their s-t variants\footnote{Throughout this paper, we use $n$ to denote the number of vertices and $m$ to denote the number of edges of the input graph.}, which is in addition given two vertices $s,t$ and asking the minimum number of vertices (or edges) whose removal disconnects $s$ and $t$. This simple idea is not efficient for two reasons. Firstly, s-t vertex (or edge) connectivity is known to be equivalent to max flow with unit vertex capacities (or edge capacities), which is a hard problem.\footnote{While there now exist almost linear time maximum flow algorithms in the sequential setting \cite{ChenKLPGS22}, these algorithms are far from simple. Morover, in other models of computation it remains largely open how efficiently one can solve (even unit-vertex-capacitated) maximum flow.} Besides, $n^2$ calls add an even larger polynomial factor to their s-t variancts running time. Thus, researchers focus on two questions: is there an algorithm to circumvent their s-t variants? Can we reduce the number of calls to their s-t variants?

For edge connectivity, the answer to the first question is affirmative. In the conventional unit-cost RAM (i.e., sequential) model, a long line of work spanning over many decades in the last century was concluded by a near-linear time algorithm \cite{Karger00}, which does not require s-t edge connectivity computation, but instead uses a technique called `tree-packing'. This technique has been proven to be versatile enough to lead to near-optimal algorithms in various computational models \cite{MukhopadhyayN20, Lopez-MartinezM21, DoryEMN21,0001Z22}.  One can view these developments through the lens of the \textit{cross-paradigm algorithm} design---a relatively new research direction where the general theme is to come up with schematic algorithms that can be implemented in several computational models without much trouble~\cite{MukhopadhyayN20, BlikstadBEMN22, RozhonGHZL22, 0001Z22, Nanongkai24}. The second question was also answered affirmatively recently by the technique `isolating cuts' \cite{LiP20, MukhopadhyayN21-Note}, which reduces the number of calls to polylogarithmic.

For the vertex connectivity problem, in the sequential model, the second question was solved very recently~\cite{LiNPSY21} which reduces the number of calls to polylogarithmic. However, unlike the tree packing framework for edge connectivity which has been successfully implemented in various models, the reduction in~\cite{LiNPSY21} is highly sequential and suffers from several bottlenecks which make it only suitable for the sequential model (we discuss more about these bottlenecks in \cref{sec:overview}). Moreover, most the known algorithms for vertex connectivity~\cite{Even75,
BeckerDDHKKMNRW82, LinialLW88, CheriyanT91, CheriyanR94, HenzingerRG00, Gabow06,LiNPSY21}
This is in contrast with the progress on edge connectivity, which is known to be easier to handle than its s-t variant. Thus, a central question this paper aims to answer is:

\begin{center}
\textit{Is solving s-t vertex connectivity necessary and/or sufficient for solving\\ (global) vertex connectivity in various computational models?}
\end{center}

\noindent
In this paper, we show a trifecta of such connections:
\begin{description}
    \item[Sufficiency:]  We show that s-t vertex connectivity is \textit{sufficient} for vertex connectivity on \textit{parallel and distributed models} by providing an almost-perfect reduction that uses only $n^{o(1)}$ calls. This reduction is built on a new framework using \emph{common-neighborhood clustering}, which is of independent interest.

    \item[Insuffiency:] In contrast, we show that s-t vertex connectivity is \emph{not sufficient} for vertex connectivity in the \textit{communication model}, by showing a near-optimal lower bound separating the complexity of s-t vertex connectivity and vertex connectivity. We also provide a matching upper bound, thus settling the communication complexity of vertex connectivity.

    \item[Necessity:] Unlike the case for edge-connectivity, we observe that s-t vertex connectivity is indeed \textit{necessary} for vertex connectivity, at least in dense graphs. This follows from a straightforward reduction from s-t to global vertex connectivity in dense graphs. The reduction works in the \pram{} and communication models, but when applied in the \congest{} model it adds additional communication edges.
    To the best of our knowledge, this simple reduction is novel.
    This observation implies that any algorithm developed for global vertex connectivity can be adapted to solve s-t vertex connectivity on dense graphs (and hence also unit-vertex-capacitated maximum flow). This explains the lack of vertex connectivity algorithms that work without calling its s-t variant, unlike what is the case for edge connectivity. The formal proof is given in \cref{lem:sttoglobal}.
\end{description}

\subsection{Our Results: Reductions and Separations}

\paragraph{Parallel and Distributed Computing.}

We use the standard \pram{} and \congest{} models for parallel and distributed computing respectively (see~\cref{sec:model-problems} for the detailed definitions). Our first results are almost perfect reduction algorithms in both models.

\begin{theorem}\label{thm:prammain}
    In \pram{} model, if s-t vertex connectivity can be solved in $W(m,n)$ work\footnote{We assume all the complexity functions in this paper are reasonably smooth in the sense that $W(Cm,Cn)=\poly(C)\cdot W(m,n)$. } and $D(m,n)$ depth where $W(m,n)$ is superadditive\footnote{$W(m,n)$ is superadditive on $m$ means $W(m_1+m_2,n)\ge  W(m_1,n)+W(m_2,n)$. For example, $W(m,n) = m \sqrt{n}$ is superadditive, while $W(m,n) = n^{\omega}$ would not be superadditive.} on $m$, then vertex connectivity can be solved in $W(m,n)\cdot n^{o(1)}$ work and $D(m,n)\cdot n^{o(1)}$ depth.
\end{theorem}
\begin{theorem}\label{thm:congestmain}
    In \congest{} model, if S-T vertex connectivity\footnote{S-T vertex connectivity is given two disjoint vertex sets $S,T$ and asking the minimum vertex cut separating $S$ and $T$. In other models we can simply contract $S$ and $T$ into single vertices. However, in \congest{}, this is not allowed as it would change the underlying communication network. This variant does not incur new challenges as commonly both s-t or S-T vertex connectivity are solved by maximum bipartite matching.} can be solved in $R(m,n,D)$ rounds,\footnote{$D$ is the diameter of the input network.} then vertex connectivity can be solved in $R(m,n,D)\cdot n^{o(1)}$ rounds.
\end{theorem}
Our main technical contribution for showing~\cref{thm:prammain,thm:congestmain} is a new algorithmic framework based on graph decomposition for vertex connectivity. We introduce a graph property called \emph{common-neighborhood} and show that if the graph satisfies this property, then the reduction can be handled in \pram{} and \congest{}. For the general graph, we develop an efficient algorithmic framework that decomposes the graph into \emph{common-neighborhood} clusters. The decomposition algorithm, as well as its efficient implementation in both \pram{} and \congest{}\footnote{More precisely, our algorithm works in the recently introduced \emph{minor-aggregation model} \cite{RozhonGHZL22} in $n^{o(1)}$ rounds.}, might be of independent interest.

Our approach is similar to recent advances in graph algorithms, where problems are first solved on graphs with ``nice properties'' and then extended to general graphs through decomposition techniques
(for example, low-diameter decomposition \cite{FakcharoenpholRT04,AbrahamBN08,CohenKMPPRX14,BernsteinNW22}, expander decomposition \cite{SpielmanT04,Sherman13,RackeST14}, length-constraint expander decomposition \cite{HaeuplerR022,HaeuplerH0RS24})
.
Certain graph decompositions are also proven to be highly parallelizable \cite{MillerPX13,ChangS19,ChangS20,AshvinkumarBCGH24}; this is also the case for our common-neighborhood clustering.

\paragraph{Two-Party Communication.} Given the positive results in the sequential, \pram{}, and \congest{} models, one might expect a perfect reduction from vertex connectivity to its s-t variant in other models as well. Perhaps surprisingly, we show that this is not always the case. In particular in the classic \emph{two-party communication} model, we show a lower bound separating the communication complexity of vertex connectivity and s-t vertex connectivity by $\tilde{\Theta}(\sqrt{n})$.\footnote{Throughout this paper, we use $\tilde{O}(\cdot)$ to hide $\poly\log n$ factors, and $\hO{\cdot}$ to hide $n^{o(1)}$ factors.}
We also show a matching upper bound, thus resolving the communication complexity of vertex connectivity up to polylogarithmic factors.

\begin{theorem}\label{thm:ccmain}
    The two-party communication complexity of vertex connectivity is $\tT{n^{1.5}}$.
\end{theorem}

It is known from previous work~\cite{blikstad2022nearly} that the communication complexity of s-t vertex connectivity is $\tT{n}$. Thus, \cref{thm:ccmain} provides a strict separation of these two problems.
Previously, a lower bound of $\Omega(n^2)$ is only known to hold on multigraph~\cite{AssadiS23}.

In contrast, we note that for \emph{edge-connectivity}, seemingly the s-t variant is harder than the global variant: it is known that the global variant admits an $\tilde{O}(n)$ communication protocol \cite{MukhopadhyayN20}, while it remains an an open question if this is also the case for the s-t edge-connectivity.

\subsection{Our Results: Implications}

\paragraph{Implications in \pram{} and \congest{}.}

By folklore reductions, s-t vertex connectivity is known to be equivalent to \emph{maximum bipartite matching} (\BMM{})\footnote{They are equivalent for the value version, i.e., outputting the size of the cut and the size of the maximum matching (see \cref{subsec:reductions}).}, which is an extensively studied problem. For completeness, we summarizes the current state-of-the-art results for \BMM{} in \pram{} and \congest{} in \cref{sec:BMM}.\footnote{In \cref{sec:BMM}, we also show an implementation of a semi-streaming algorithm of \cite{assadi2022semi} in \congest{} and \pram{}, implying new running times for \BMM{} in these models that, to our knowledge, have not previously been observed.} By combining with our reduction Theorems~\ref{thm:prammain} and \ref{thm:congestmain}, we directly get the following results.

\begin{theorem}\label{thm:pramexactrunningtimeVC}
    In \pram{} model, vertex connectivity can be solved in $mn^{0.5+o(1)}$ work and $n^{0.5+o(1)}$ depth.
\end{theorem}

\begin{theorem}\label{thm:congestexactrunningtimeVC}
    In \congest{} model, vertex connectivity can be solved in $n^{o(1)}\cdot (n^{3/8}D^{1/2} + n^{1/2}D + n^{7/10}D^{7/10})$ rounds, which is sublinear as long as $D=n^{3/7-\eps}$ for some constant $\eps>0$.
\end{theorem}

The superadditive requirement in \cref{thm:prammain} might prevent us from reduction to \BMM{} algorithm in matrix multiplication work and depth \cite{Lovasz79,KarpUW86,MulmuleyVV87}. However, although not directly implied by \cref{thm:prammain}, our reduction framework is versatile enough to obtain vertex connectivity in matrix multiplication work and depth as well.

\begin{theorem}\label{thm:pramexactrunningtimeVCmaxtrix}
    In \pram{} model, vertex connectivity can be solved in $n^{\omega+o(1)}$ work and $n^{o(1)}$ depth, where $n^{\omega}$ is the optimal work for matrix multiplication in $n^{o(1)}$ depth.
\end{theorem}
\begin{remark}
    The $n^{\omega+o(1)}$ work and $n^{o(1)}$ depth can be improved to $\tO{n^{\omega}}$ work and $\tO{1}$ depth if we apply a common-neighborhood cover algorithm with $\tO{n^2}$ work and $\tO{1}$ depth instead of $\tO{m}$ work and $n^{o(1)}$ depth. However, for consistency throughout the paper across different models, we do not optimize this here.
\end{remark}

\paragraph{Previous Work.} Parallel vertex connectivity was mostly studied for small $\kappa$ (the vertex connectivity of the input graph) \cite{CheriyanT91,CheriyanKT93,GianinazziH20}, which have at least $\Omega(n)$ depths for large~$\kappa$. Despite the long-known fact that BMM can be solved in matrix multiplication work and depth \cite{Lovasz79,KarpUW86,MulmuleyVV87}, the only known sublinear depth algorithm for vertex connectivity in the general case is by reduction to matrix multiplication \cite{LinialLW88} in $\tO{n^{\omega+1}}$ work and $O(\log^2n)$ depth. Our result \cref{thm:pramexactrunningtimeVCmaxtrix} improves this 35-year-old result in the subpolynomial depth regime to almost matrix multiplication time. Moreover, in the sublinear depth regime, our result \cref{thm:pramexactrunningtimeVC} is faster than matrix multiplication time on moderately sparse graphs.

In the \congest{} model, vertex connectivity was either studied for small $\kappa$ \cite{Thurimella97,PritchardT11,parter2020distributed,parter2022near,JiangM23}, or by $O(\log n)$ approximation \cite{Censor-HillelGK14}. For exact computation and large $\kappa$, all of the mentioned results cost at least $\Omega(n)$ rounds, even for a graph with constant diameter. Our result \cref{thm:congestexactrunningtimeVC} gives the first sublinear round algorithm for exact large vertex connectivity as long as $D$ is moderately small.

\paragraph{Implications for Streaming.} We use the standard graph streaming model (formally defined in \cref{sec:model-problems}). It is a well-known fact that an efficient streaming protocol yields an efficient communication protocol. Thus, our communication lower bound in \Cref{thm:ccmain} immediately implies the following streaming lower bound.

\begin{theorem}[Streaming lower bound]\label{cor:streaminglowerbounds}
    Any $P(n)$-pass randomized streaming algorithm of vertex connectivity needs $\Omega(n^{1.5}/P(n))$ space.
\end{theorem}

Our algorithmic ideas can be used in the streaming model as well and give the following reduction theorem. In the randomized streaming model, we say a reduction from vertex connectivity to maximum bipartite matching is an $(\alpha,\beta)$-reduction if vertex connectivity can be solved in $\alpha S(n)$ space and $\beta P(n)$ passes as long as maximum bipartite matching can be solved in $S(n)$ space and $P(n)$ passes.

\begin{theorem}[Streaming upper bound]\label{thm:streamingupperbounds}
    In the randomized streaming model, there is a $(\tO{n^{0.5}},O(1))$-reduction from vertex connectivity to maximum bipartite matching.
\end{theorem}

A natural question is whether the reduction of~\cref{thm:streamingupperbounds} is optimal or not. Although we cannot show a lower bound, there is a simple observation that any better reduction gives a strong streaming lower bound for maximum bipartite matching.

\begin{corollary}\label{cor:streaminglowerbounds}
    If an $(\alpha,\beta)$-reduction from vertex connectivity to maximum bipartite matching exists for $\alpha\cdot \beta<n^{0.5-\epsilon}$, then in the semi-streaming model, maximum bipartite matching cannot be solved in $o(n^{\epsilon})$ passes.
\end{corollary}

The implied BMM lower bound of~\cref{cor:streaminglowerbounds} would be a breakthrough in semi-streaming lower bound \cite{AhnG11, AssadiR20, ChenKPS0Y21}. Hence, if we are to believe the conjecture that maximum bipartite matching can be solved in $\tO{n}$ space and $n^{o(1)}$ passes, then according to \cref{cor:streaminglowerbounds}, the reduction of~\cref{thm:streamingupperbounds} can not be improved by a polynomial factor. Alternatively, \cref{cor:streaminglowerbounds} can be viewed as a new potential way to obtain polynomial pass semi-streaming lower bound for bipartite matching through developing better reductions for vertex connectivity.

At last, by combining \cref{thm:streamingupperbounds} with the current progress on BMM \cite{assadi2022semi}, we get the following result.

\begin{corollary}\label{cor:exactstreaming}
    There is a randomized streaming algorithm for vertex connectivity with $\tO{n^{1.5}}$ space and $n^{0.75+o(1)}$ passes.
\end{corollary}

\section{Technical Overview} \label{sec:overview}

Given an undirected graph $G=(V,E)$, a partition $(L,S,R)$ of $V$ (where $L,S,R$ are all non-empty) is a \emph{vertex cut} (or simply \emph{cut} in this work) if there are no edges between $L$ and $R$. The size of the cut $(L,S,R)$ is defined as $|S|$, and we say $S$ is a \emph{separator} or more precisely, $(s,t)$-separator for any $s\in L,t\in R$. We say a cut $(L,S,R)$ is a \emph{$t$-sink cut} if $t\in R$.
Throughout this paper, we will focus on solving the following single-sink version of vertex connectivity.

\begin{definition}[\SSUDVC{} (single-sink vertex connectivity)]\label{def:overview-singlesource}
    Given an undirected graph $G=(V,E)$ and a vertex $t\in V$, outputs a minimum $t$-sink cut.
\end{definition}

There is a simple randomized reduction from the general vertex connectivity problem to \SSUDVC{} (see~\cref{lem:globaltosinglesource}), and fixing a sink makes the algorithm easier to state. So we focus on the single-sink version.

Another useful observation is that to solve \SSUDVC{}, it suffices to solve the following \mn{} problem. Given a vertex $u\in V$, define $N_G(u)$ (or $N(u)$ when the context is clear) to be the set of neighboring vertices of $u$, i.e., vertices that share an edge with $u$. Let $N[u]:=N(u)\cup\{u\}$. For a vertex set $V'\subseteq V$, define $N(V')=(\cup_{u\in V'}N(u))-V'$ and $N[V']=N(V')\cup V'$.

\begin{definition}[\mn{}]\label{def:overview-min-neighbor}
    Given an undirected graph $G=(V,E)$ and a vertex set $V'\subseteq V$, find $L\subseteq V'$ such that $|N_G(L)|$ is minimized.
\end{definition}

If we have an algorithm for \mn{}, to solve \SSUDVC{}, we set $V'=V-N[t]$ and call the \mn{} algorithm to get $L\subseteq V'$ such that $|N(L)|$ is minimized. One can verify that $(L,N(L),V-L-N(L))$ is a valid $t$-sink cut for any $L\subseteq V'$. Moreover, every minimum $t$-sink cut $(L,S,R)$ satisfies that $L\subseteq V'$ and $S=N(L)$ (otherwise we can set $S=N(L)$ which decrease the size of the cut).

We can also show the reversed reductions by simple arguments, so \mn{} is, essentially, equivalent to vertex connectivity. Throughout this overview, we focus on solving \mn{}.

\paragraph{Bottlenecks of \cite{LiNPSY21} beyond the sequential model.} We first discuss the bottlenecks of implementing the sequential algorithm of \cite{LiNPSY21} in \pram{}, \congest{}, and two-party communication models.
\begin{enumerate}
    \item The algorithm is inherently sequential in the sense that it uses a breadth-first-search with unbounded depth as a subroutine. This makes it hard to be implemented in situations where parallelism is a requirement, such as \pram{} and \congest{}.
    \item Very informally, the algorithm reduces a vertex connectivity instance on a graph $G$ with $n$ vertices and $m$ edges to a bipartite maximum matching instance on a graph $G'$ with $\tO{m}$ edges and $\tO{n^2}$ vertices. In other words, when $n$ is the parameter, the reduction is not efficient. In many situations, $n$ is the parameter for measuring the complexity, such as \congest{} and communication, or in the matrix multiplication work of \pram{}.
\end{enumerate}
As mentioned in the introduction, we develop a new framework based on common-neighborhood clustering to sidestep these bottlenecks. The organization of this overview is as follows.

\begin{enumerate}
    \item In \cref{subsec:buildingtheframework}, we define the common-neighborhood property of a graph, and show how to reduce \mn{} to s-t vertex connectivity preserving the number of edges (but not vertices) in a \emph{common-neighborhood cluster} (defined in \cref{subsec:buildingtheframework}). The algorithm is described in \pram{} model for the sake of simplicity.
    \item In \cref{subsec:overviewCNC}, we show our decomposition algorithm for common-neighborhood clustering in \pram{}, which can also be implemented efficiently in various models like \congest{} and communication. Common-neighborhood clustering is a way to reduce solving \mn{} on a general graph to solving \mn{} on common-neighborhood clusters. Combined with \cref{subsec:buildingtheframework}, this proves \cref{thm:prammain} (not preserving the number of vertices is fine here because we assumed $W(m,n)$ is superadditive on $m$ in \cref{thm:prammain}).
    \item For the \congest{} model, to prove \cref{thm:congestmain}, we do not have the guarantee that $R(m,n,D)$ is superadditive. Thus, it suffers from the second bottleneck of \cite{LiNPSY21}. We show that it can be resolved nearly perfectly in \cref{subsec:overviewcongest} and prove \cref{thm:congestmain}.
    \item To prove
    that vertex-connectivity can be solved in matrix multiplication work in \pram{} (\cref{thm:pramexactrunningtimeVCmaxtrix}), only preserving the number of edges does not suffice. We show how to leverage our framework to circumvent this bottleneck in \cref{subsec:overviewparallel}.
    \item
    \Cref{subsec:overviewcc} gives an overview of our results in the two-party communication setting (\cref{thm:ccmain}).
    We show that the number of vertices cannot be perfectly preserved by proving a near-optimal lower bound for vertex connectivity separating it from s-t vertex connectivity. The near-optimality of the lower bound is shown by an upper bound.
\end{enumerate}

\subsection{Building the Framework: Common-neighborhood Cluster}\label{subsec:buildingtheframework}

Recall that the input to \mn{} is an undirected graph $G=(V,E)$ and a vertex set $V'\subseteq V$. Throughout this overview, we always use $\LL$ to denote an arbitrary one of the minimizers of $\min_{L'\subseteq V'}|N(L')|$. Wlog., we will assume $G[\LL]$ is connected, otherwise we can take a connected component of $G[\LL]$ denoted as $L'$, we have $N(L')\subseteq N(\LL)$ so we can use $L'$ as the minimizer instead.

Our goal is to find $\LL$, and let us assume that we know its size.\footnote{This assumption can be easily removed by guessing the size of $\LL$ by powers of $2$ which only adds a logarithmic factor to the complexity.} Now we are ready to define common-neighborhood clusters. For two sets $A,B$, define $A\triangle B=(A-B)\cup (B-A)$ as their \emph{symmetric difference}.

\begin{definition}\label{def:common-neighborhood}
    Given an undirected graph $G=(V,E)$ and a vertex set $C\subseteq V$, we say $C$ has \emph{\nd{}} $d$ if for any $u,v\in C$, $|N(u)\triangle N(v)|\le d$.
\end{definition}

We say $C$ is a \emph{common-neighborhood cluster} if $C$ has \nd{} at most $\hO{|\LL|}$.
 In this section, we will assume $V'$ is a common-neighborhood cluster and show a \pram{} algorithm solving \mn{}.
We will use the following property of a common-neighborhood cluster. The formal statement and proofs are in~\cref{lem:twocases}.

\begin{lemma}[Informal version of~\cref{lem:twocases}]
    If the input vertex set $V'$ is a common-neighborhood cluster, then one of the following two cases happens.
    \begin{enumerate}
        \item $|N_G(\LL)\cap V'|=\hO{|\LL|}$.
        \item $G[V']$ is almost a clique. For this overview, assume $V'$ is indeed a clique for convenience.
    \end{enumerate}
\end{lemma}

Now we show the rough idea of how to find $\LL$ in each case. The formal proofs are in~\cref{sec:framework}.

\paragraph{Case 1: $|N(\LL)\cap V'|=\hO{|\LL|}$.} In this case, we can employ the \emph{isolating cuts} technique~\cite{LiNPSY21}, that, given an independent set $T$, finds the minimum vertex cut separating exactly one node in $T$ from the other nodes in $T$. For our case, we let $T$ be an independent sample by including each node $v\in V'$ in $T$ with probability $1/\hO{|\LL|}$. This means that with probability $\hM{1}$, $T$ has exactly one node in $\LL$ and no nodes in $N(\LL)$.

\paragraph{Case 2: $G[V']$ is a clique.} In this case, we are trying to solve \mn{} given that (i) $V'$ is a clique, and (ii) every two nodes in $V'$ has their neighbors in $G$ differ by at most $\hO{|\LL|}$ nodes. We call this problem
MinNeighbor-NearClique-CommonNeighborhood, or
\mnnccn{} for short.

For a vertex $s\in V'$, define the \emph{$s$-source \mn{}} problem as finding $\arg\min_{s\in L'\subseteq V'}N(L')$.
Note that it can be solved in one s-t vertex connectivity call: add a super node $t$ connecting to $N(V')$ and call s-t vertex connectivity. Given this, the algorithm for dealing with Case 2 contains two steps.

\begin{enumerate}
    \item Sample $|V'|/|\LL|$ nodes uniformly at random so that one of them is in $\LL$.
    \item For every sample node, solve $s$-source \mn{} on a \emph{sparsified graph} (which we will define later) with $|V'|\cdot |\LL|$ edges.
\end{enumerate}

The cumulative number of edges for calling s-t vertex connectivity is $|V'|^2$, which is proportional to the number of edges in $G$ (since $V'$ is a clique). Thus, the reduction preserves the total number of edges. Moreover, if the work function $W(m,n)$ for solving s-t connectivity is superadditive on $m$, then the work of solving \mnnccn{} can be written as $O(W(|V'|\cdot |\LL|,n))\cdot |V'|/|\LL|\le O(W(|V'|\cdot |\LL|\cdot (|V'|/|\LL|),n))=O(W(m,n))$.

The remaining question is how to run $s$-source \mn{} on a \emph{sparsified graph with $|V'|\cdot |\LL|$ edges}. This is because we can delete all edges from $V'-\{s\}$ to neighbors of $s$, and deleting them will not change the set $N(L')$ for any set $L'\ni s$.
Recall that $V'$ has \nd{} at most $\hO{|\LL|}$, and hence
$u\in V'$ can have at most  $\hO{|\LL|}$ remaining edges outside $N(s)$.
See \cref{fig:sparse} for an illustration.

\begin{figure}[!ht]
	\centering
	\includegraphics[scale=0.2]{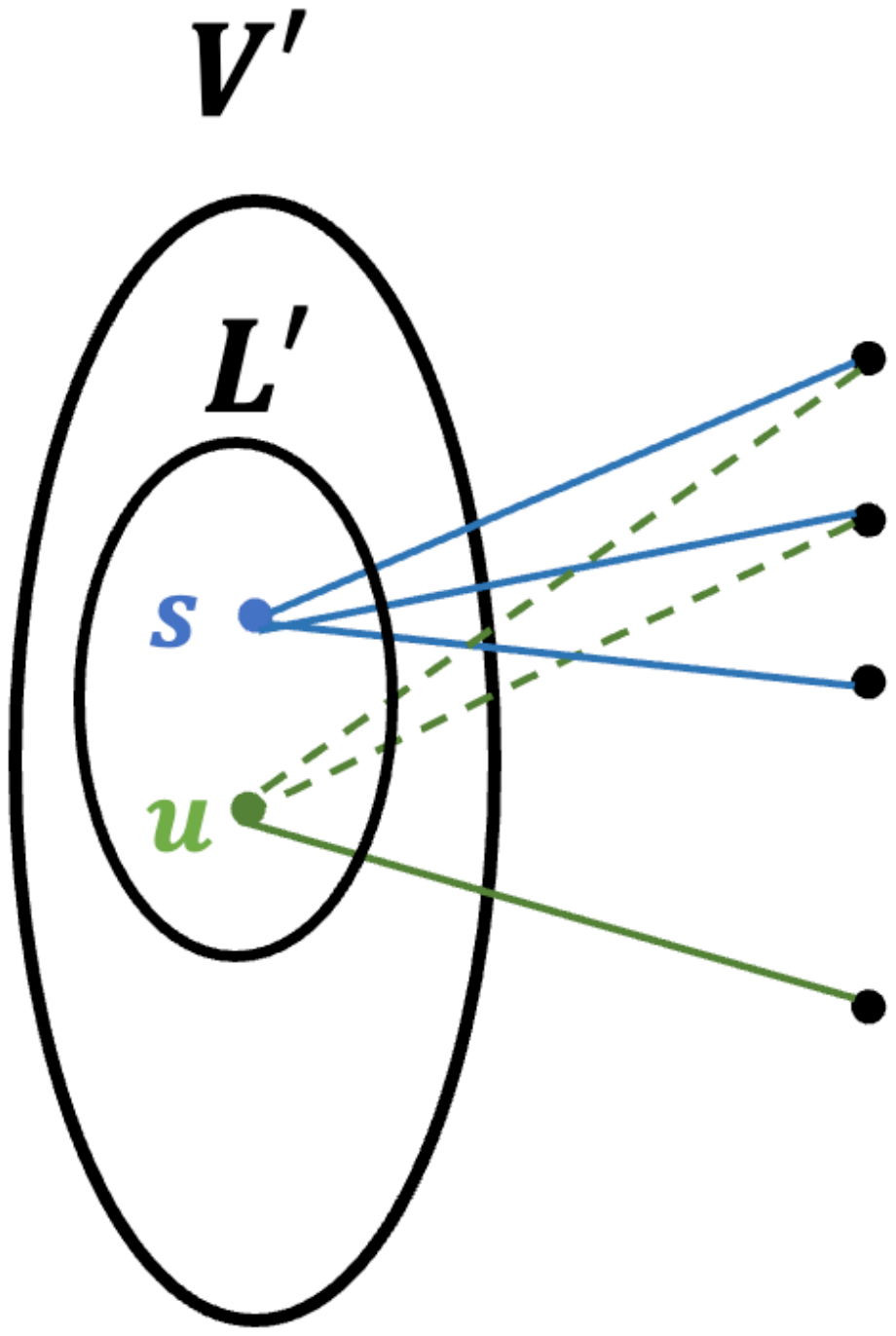}
	\setcaptionwidth{0.95\textwidth}
	\caption{\small All the dashed edges (edges from $V'-\{s\}$ to neighbors of $s$) are deleted. For every $L'$ with $s\in L'$, $N(L')$ does not change.}\label{fig:sparse}
\end{figure}

\begin{remark}
    The proof of \cref{lem:twocases} and the idea of sparsifying the instance can be viewed as an adaption from the proof in \cite{LiNPSY21}, so we do not claim novelty here. Our novelty lies in (i) introducing the graph decomposition for common-neighborhood clustering, and showing how it can be implemented in various models, (ii) circumvent the bottleneck of not preserving the number of vertices, as we will illustrate in the following sections.
\end{remark}

\paragraph{Building the Framework.} To summarize, once we have a common-neighborhood cluster $V'$, the problem reduces to one of (i) finding a minimum isolating cut, or (ii) solving \mnnccn{}. If $V'$ is not a common-neighborhood cluster, the natural idea is to decompose it into common-neighborhood clusters, which we call \emph{common-neighborhood clustering} defined as follows.

\begin{definition}[Common-neighborhood clustering]
    Given $G=(V,E)$ and $V'\subseteq V$, output a set of vertex sets $\cC$ such that
    \begin{enumerate}
        \item each vertex $u\in V'$ is contained it at most $\hO{1}$ clusters in $\cC$,
        \item for any $C\in\cC$, $C$ is a common-neighborhood cluster and $C\subseteq V'$,
        \item $\LL\subseteq C$ for some $C\in\cC$ with probability at least $1/n^{o(1)}$.
    \end{enumerate}
\end{definition}

\noindent
Our framework can be thus summarized as follows.

\begin{description}
    \item[Step 1:] Compute a \emph{common-neighborhood clustering} denoted as $\cC$.
    \item[Step 2:] For every $C\in\cC$, solve \emph{minimum isolating cut} and \mnnccn{}{} to get non-empty $L_C\subseteq C$ minimizing $N(L_C)$.
    \item[Step 3:] Output the minimum $N(L_C)$ among all $C\in\cC$.
\end{description}

\subsection{A Decomposition Algorithm for Common-neighborhood Clustering}\label{subsec:overviewCNC}

In this section, we show a \pram{} algorithm for common-neighborhood clustering which can be implemented in different models. One crucial property of $\LL$ that we state here is that $\LL$ has \nd{} at most $2|\LL|$---we defer the formal proof of this statement to ~\cref{lem:unbalancedclose}.

\paragraph{A proof of existence.} We first give a simple, but inefficient, algorithm that shows the existence of common-neighborhood clustering. Define the neighborhood-difference graph $G'=(V',E')$ as follows: for any pair of vertices $(u,v)\in V'\times V'$, $(u,v)\in E'$ if and only if $|N_G(u)\triangle N_G(v)|\le 2|\LL|$.

We also need the notion of the sparse neighborhood cover that is studied in the literature~\cite{AwerbuchBCP98}. We define it as follows.

\begin{definition}[Sparse neighborhood cover (SNC)]\label{def:snc}
    Given an undirected graph $G=(V,E)$, the \emph{sparse neighborhood cover} of $G$ is a set of vertex sets $\cC$ (called clusters) such that
    \begin{enumerate}
        \item (sparse) each vertex $v\in V$ is contained in at most $\tO{1}$ clusters in $\cC$,
        \item (low diameter) the diameter of $G[C]$ for any $C\subseteq\cC$ is at most $O(\log n)$,\label{prop:low-diam}
        \item (cover) for every $v\in V$, there exists a cluster $C\subseteq\cC$ such that $N_G[v]\subseteq C$.
    \end{enumerate}

\end{definition}

It is also proved in~\cite{AwerbuchBCP98} that SNC can be constructed efficiently. We claim that applying SNC on $G'$ gives us a common-neighborhood clustering. Let us verify each property of common-neighborhood clustering one by one.
\begin{enumerate}
    \item Each vertex $u\in V'$ is contained in at most $\tO{1}$ clusters in $\cC$ according to the \emph{(sparse)} condition.
    \item For any $C\in\cC$ and any $u,v\in C$, there exists a path from $u$ to $v$ in $G'$ with length $O(\log n)$ according to \emph{low diamter}, denoted by $(u,v_1,v_2,...,v)$. Notice that an edge $(u,v_1)\in E'$ means we can delete and add $2|\LL|$ nodes from $N_G(u)$ to get $N_G(v_1)$. We can repeat this argument along the path, which shows that $|N_G(u)\triangle N_G(v)|=O(|\LL|\log n)$. This proves that $C$ has \nd{} at most $O(|\LL|\log n)$, i.e., it is a common-neighborhood cluster.
    \item $\LL$ has \nd{} at most $2|\LL|$ according to \cref{lem:unbalancedclose}, $\LL$ must be a clique in $G'$ according to the definition of $G'$ (actually here we do not require $G[\LL]$ to be connected, this additional property is to assist a faster algorithm which we will explain later). Thus, take an arbitrary $v\in \LL$, and note $\LL\subseteq N_{G'}[v]$, which means $\LL$ is contained in some cluster due to property \emph{(cover)}.
\end{enumerate}

This completes the existential proof. However, construction $G'$ is not efficient and runs in $O(n^3)$ time in a trivial way. \qed

\paragraph{Speeding up the construction of $G'$: approximating the neighborhood differences.} One reason for the slow computation of $G'$ is that in order to decide if an edge $(u,v)$ exists in $G'$ or not, we need to go over the neighbors of both $u$ and $v$, which in the worst case use time $n$, resulting in a total time $O(n^3)$. This is unavoidable if we want to know the exact size of $|N_G(u)\triangle N_G(v)|$. However, we will show that an approximate size of $|N_G(u)\triangle N_G(v)|$ suffices.

\newcommand{\apd}{\text{apd}}
We can build $G'$ in the following way:  (i) if $(u,v)\in E'$, then $|N_G(u)\triangle N_G(v)|\le 3|\LL|$, (ii) if $(u,v)\not\in E'$, then $|N_G(u)\triangle N_G(v)|> 2|\LL|$. Such $G'$ can be constructed by only knowing an approximate value of $|N_G(u)\triangle N_G(v)|$ for any $u,v\in V$. One can verify that applying SNC on this $G'$ does not make too much difference and can still give us a common-neighborhood clustering. Moreover, computing an approximate value of $|N_G(u)\triangle N_G(v)|$ only requires $\tO{1}$ time by standard sampling and sketching techniques. This is proved in~\cref{cor:approxsketching}.

\begin{assumption}
    In the rest of this overview, we use $\apd(u,v)$ to denote a $1.1$-approximation value of $|N_G(u)\triangle N_G(v)|$, which can be found using~\cref{cor:approxsketching}.
\end{assumption}

This means that $G'$ can be constructed in time $\tO{n^2}$. However, this is still slow if our goal is $\tO{m}$ work. More severely, computing $G'$ is not a `localized' procedure since it requires two very far away nodes to decide whether they have an edge in $G'$ or not, posing a challenge for distributed computing. Thus, our algorithm will not compute $G'$, or a SNC of $G'$, but instead involves a completely new idea.

\paragraph{A fast algorithm with large depth.} Now we describe an algorithm that has a large depth but in $\tO{m}$ work in \pram{}.
Given that $\LL$ has \nd{} at most $2|\LL|$ and $G[\LL]$ is connected, if we know a vertex $u\in \LL$, we can construct a common-neighborhood cluster $C$ with $\LL\subseteq C$ in the following way: build a BFS tree on $G$ with root $u$, such that a vertex $v$ is included in the BFS tree only if $\apd(u,v)\le 3|\LL|$. Let us denote the vertices in this BFS tree as $B'\left[3|\LL|\right](u)$. It can be equivalently defined as follows.
\begin{definition}
    Define $B'[\ell](u)$ as the largest possible set such that (i) $G[B'[\ell](u)]$ is connected, (ii) for any $v\in B'[\ell](u)$ we have $\apd(u,v)\le \ell$.
\end{definition}
In this way, we can guarantee that the \nd{} of $C$ is at most $O(|\LL|)$, $\LL\subseteq C$ and the construction time of $C$ is proportional to the number of adjacent edges of $C$.

\paragraph{Forgetting what $\LL$ is.} In general, we do not know a vertex in $\LL$. A natural idea is to sample vertices with probability $1/|\LL|$ and run the above algorithm for each sampled node and add all the resulting clusters into $\cC$. In this way, with good probability one of the clusters in $\cC$ will contain $\LL$. However, it is easy to see that the \emph{(sparse)} property does not hold, i.e., each vertex could appear in as many as $n/|\LL|$ clusters (this also results in a large work). The reason for the large overlap is that clusters can intersect with each other. A naive way to reduce the overlap would be to delete every cluster that intersects with others. Unfortunately, this does not work because the cluster that contains $\LL$ could be deleted by this procedure. Hence, we set our goal as follows: sample some vertices $u$ and build the cluster $B'[\hO{|\LL|}](u)$, such that (i) at least one sampled vertex is in $\LL$, and (ii) the cluster containing $\LL$ does not intersect other clusters. If these conditions are satisfied, we can let $\cC$ be all the clusters that do not intersect with others, and we will be done. The work is also guaranteed to be $\tO{m}$ since if two BFS trees meet together at one node, they can stop growing and they are both excluded from being in $\cC$.

To achieve the aforementioned non-overlap requirement (ii), an ideal situation would be the following. There are two vertex sets $A,B$ with $\LL\subseteq A\subseteq B$ such that: \begin{enumerate}
    \item the size of $B$ is roughly equal to the size of $A$,
    \item $B$ has \nd{} at most $\hO{|\LL|}$, and
    \item Any vertex outside $B$ has large \nd{} from the vertices in $A$ (e.g., at least $10d_A$ where the \nd{} of $A$ is $d_A$).
\end{enumerate}

If we have such $A,B$, then we can sample with probability $1/|B|$ such that, with good probability, one of the sampled nodes will be in $A$ while all other sampled nodes are outside $B$. Now we can build $B'[2d_A](u)$ for each sampled node $u$ and $A$ will be contained in one of the clusters, and the cluster containing $A$ will not overlap with other clusters.

Now we discuss how to prove the existence of such $A,B$. For a vertex $v \in \LL$, $\LL \subseteq B'[3|\LL|](v)$. In what follows we assume $v$ is an arbitrary vertex in $\LL$.
If we can find a parameter $r>2$ such that $|B'[99r|\LL|](v)|\le n^{o(1)}|B'[r|\LL|](v)|$ holds, then we can set $A=B'[r|\LL|](v)$ and $B=B'[99r|\LL|](v)$ and all the three conditions for $A$ and $B$ hold. Our goal is to find the smallest such $r$ starting from $r=2$ and increasing $r$ to $99r$ iteratively.
Notice that each repetition increases the size of $B'[r|\LL|](v)$ by a factor of $n^{o(1)}$ but only increases $r$ by a factor of $99$, by setting the parameter correct we can make sure $r=\hO{1}$ in the end. This proves the existence of $A$ and $B$. As long as $A,B$ exists, we can guess the size of $B$ by powers of $2$, so that we can sample with probability $1/|B|$ as the algorithm suggests.

A side remark is that the algorithm only guarantees $\LL\subseteq C\in\cC$ with probability $1/n^{o(1)}$. This probability can be boosted by repeating the algorithm $n^{o(1)}$ and add all the found clusters to $\cC$.

\paragraph{Reducing the depth by graph shrinking.} Unfortunately, the above idea still requires doing BFS with unbounded depth, which is not suitable for both parallel and distributed models. The main bottleneck is about finding $B'[\ell](u)$ for many different $u$ and some parameter $\ell$. If $B'[\ell](u)$ contains less than $n^{o(1)}$ nodes, then it can be done in $n^{o(1)}$ depth since the BFS can include at most $n^{o(1)}$ nodes. Thus, let us assume $B'[\ell](u)$ contains much more than $n^{o(1)}$ nodes.

The idea is to shrink the graph size by a factor of $n^{o(1)}$ while not losing too much information about $B'[\ell](u)$. We sample vertices with probability $1/n^{o(1)}$ and denote the sampled vertex set as $X$. With good probability we have (i) $X\cap B'[\ell](u)\not=\emptyset$, (ii) size of $X$ decrease the size of $V$ by a factor of $1/n^{o(1)}$. We decompose the graph by growing vertex disjoint BFS trees simultaneously from every vertex $v\in X$, such that a vertex $w$ is included in the BFS tree of $v$ only if $\apd(v,w)\le 3\ell$. After that, we contract each BFS tree and call them \emph{super nodes}. The graph size decreases by a factor of $1/n^{o(1)}$. Moreover, all vertices in $B'[\ell](u)$ are preserved in some super nodes since $X$ contains at least one node $v\in B'[\ell](u)$ and all nodes in $B'[\ell](u)$ are eligible to join the BFS tree of $v$.

We can repeat the above operation as long as the number of super nodes intersecting $B'[\ell](u)$ is at least $n^{o(1)}$: sample super nodes with probability $1/n^{o(1)}$ to make sure one sampled super node intersects $B'[\ell](u)$; from each super node $v$ grows a BFS tree only including another supernode $w$ such that there exist original nodes $v'\in v,w'\in w$ with $\apd(v,w)\le 3\ell$; contract each BFS tree. Notice that the \nd{} of each super node grows by a constant factor in each iteration.

We repeat until the number of super nodes intersecting $B'[\ell](u)$ becomes $n^{o(1)}$, then we can find all of them by a simple BFS of depth $n^{o(1)}$. The only problem is that a super node might contain extra nodes besides nodes in $B'[\ell](u)$. This is fine due to the following reason: in each iteration the \nd{} of each super node grows by a constant factor, while the number of nodes of the graph shrink by a factor of $1/n^{o(1)}$; by setting the parameter correctly, the final \nd{} of each super node is at most $n^{o(1)}\ell$. Thus, all the nodes in the super modes intersecting $B'[\ell](u)$ are contained in $B'[n^{o(1)}\ell](u)$. This $n^{o(1)}$ factor is not too large and a slight change in the algorithm can handle it.

The algorithm and analysis are described in details in \cref{sec:snc}.

\subsection{Distributed Algorithms}\label{subsec:overviewcongest}

Recall the framework for our \congest{} algorithms: we first apply common-neighborhood cluster described in \cref{subsec:overviewCNC} so that we can assume $V'$ is a common-neighborhood cluster as defined in \cref{subsec:buildingtheframework}, then we use the observations in \cref{subsec:buildingtheframework} to reduce the problem to (i) minimum isolating cut, and (ii) \mnnccn{}. Both problems can be reduced to s-t vertex connectivity while preserving the total number of edges $m$ when doing the reduction, as explained in \cref{subsec:buildingtheframework}.

Now, we describe the key obstacles to implementing distributed \congest and outline our solution to overcome the technical obstacles.

\subsubsection{Bottlenecks of Proving \cref{thm:congestmain}: Minimum Isolating Cut}

The following isolating cuts lemma can almost handle our case.
\begin{lemma} [Lemma 4.2 of~\cite{LiNPSY21}] \label{lem:isolating cut old}
Given a graph $G = (V,E)$ and an independent set $T \subseteq V$ of size at least $2$, there is an algorithm that outputs for each $v \in T$ a $(v, T -\{v\})$-min-separator $C_v$. The algorithm makes calls to s-t vertex mincut on graphs with $O(m)$ total number of vertices and $O(m)$ total number of edges and takes $\tilde O(m)$ additional time.
\end{lemma}

This version of the isolating cuts lemma works perfectly fine for sequential settings as they focus on bounding the total number of edges. Their algorithm can also be easily parallelized. However, the number of vertices can be as large as $O(m)$ which is not suitable for distributed computing.

To handle this situation, we prove a more efficient version of the isolating cuts lemma that also guarantees $\tilde O(n)$ total number of vertices and $O(m)$ total number of edges.

\begin{lemma} \label{lem:refined isolating cut lemma new}
Given a graph $G = (V,E)$ and an independent set $T \subseteq V$ of size at least $2$, there is an algorithm that outputs for each $v \in T$ a $(v, T -\{v\})$-min-separator $C_v$. The algorithm makes calls to s-t vertex mincut on graphs with $\tilde O(n)$ total number of vertices and $O(m)$ total number of edges and takes $\tilde O(m)$ additional time.
\end{lemma}

To elaborate
 for minimum isolating cut, by following the proof of \Cref{lem:isolating cut old} (formally described in Lemma 4.2 of~\cite{LiNPSY21}), we remove a set of vertices $X$ from the graph so that each connected components $C$ of the remaining graph contains at most one terminal $t\in T$. For each connected component $C$ with one terminal $t\in T$ inside, we try to find $\min_{s\in L'\subseteq C}N_G(L')$. Notice that this minimization problem only depends on all the edges adjacent to $C$, which are non-overlapping for different $C$, so the total number of edges is bounded by $m$. However, it depends on all vertices in $N_G[C]$, which could intersect at lot for different $C$.

To obtain the improvement as in \Cref{lem:refined isolating cut lemma new}, we briefly explain how to deal with this issue.  For every $C$ where we wish to find $\min_{s\in L'\subseteq C}N_G(L')$, we prove that it suffices to keep $\tO{|C|}$ vertices in $N_G(C)$ and deleting the others, while preserving $\min_{s\in L'\subseteq C}N_G(L')$. This is done by finding a maximum bipartite matching between $C$ and $N_G(C)$. The formal proof can be found in \Cref{sec:cluster sparsification lemma}. After that, the total number of vertices is upper bounded by $\sum_{C}\tilde O(|C|)=\tilde O(n)$. Given \Cref{lem:refined isolating cut lemma new}, we can easily obtain PRAM and \congest{} algorithms where the number of vertices involved in the isolating cuts lemma is $\tilde O(n)$. We refer to \Cref{sec:isolating cut lemma parallel dist} for details.

\subsubsection{Bottlenecks of Proving \cref{thm:congestmain}: Solving \mnnccn{}} It is easy to see that the algorithm described in \cref{subsec:buildingtheframework} for \mnnccn{} does not preserve the number of edges. In fact, it contains $|V'|/|\LL|$ many instances, where each instance contains $|V'|\cdot |\LL|$ edges and $n$ vertices.

 Now we briefly explain how this is solved in \congest{}. We will use the fact that $V'$ is a clique. For each instance with $|V'|/|\LL|$, we map those edges uniformly at random into the clique $G[V']$, and we solve the instance by an oracle call to s-t vertex connectivity where communication happens only on the mapped edges. Since there are in total at most $\tO{|V'|^2}$ edges, and $G[V']$ is a clique, in expectation each edge is only included in $\tO{1}$ many instances.

\paragraph{Random mapping of the communication inside the clique.} Here we elaborate on the random mapping procedure. Our solution is to offload the communication of each instance to the cliques by a simple random load-balancing strategy.

Given each instance $G'_s$, which is the sparsified graph with $\tO{|V'|\cdot |\LL|}$ edges as described in \cref{subsec:buildingtheframework}, we define a random (public) function $f_s: N(V') \rightarrow V'$ that maps a node  $v \in N(V')$ to a random node in $V'$. For each $u \in N(V')$, our goal is to have $f(u) \in V'$ simulate $u$ (in other words, $f(u)$ acts as a proxy for $u$), i.e., $f(u)$ acts as if it is node $u$. In order for  $f(u)$ to simulate $u$, $f(u)$ must learn the following information:
\begin{itemize} [noitemsep]
\item The algorithm description of node $u$ for executing s-t vertex mincut algorithm in $G'_s$. This can be done by having $u$ send its code to $f(u)$.
\item $f(u)$ learns every node $N_G(u)\cap V'$. This can be done by letting every neighbor $v \in N_G(u) \cap C$ send its id to $f(u)$.
\end{itemize}
In this case, $f(u)$ can act as if it is node $u$. So, any communication between $v \in V'$ to $u \in N(V')$ can be moved to the same communication between $v \in V'$ and $f(u) \in V'$. Therefore, the communication between every edge $(v,u)$ where $v\in C, u\in N(V')$ will map to a $(v,f(u))$ where $u$ is a random node.  Using the fact of graph $G'_s$ that every node $v \in C - \{s\}$ has at most $\hO{|\LL|}$ edges to $N(V')$, we conclude that every node $v \in C -\{s\}$ in the instance $G'_s$ communicates using random $\hO{|\LL|}$ edges adjacent to $V'$.

In summary, fix a node $v \in V'$, and the algorithm will communicate from $v$ to $\tilde O(|\LL|)$ random vertices in $V'$ for each instance $G'_s$. Since every $\tilde O(\frac{|V'|}{|\LL|})$ instance will use $\tilde O(|\LL|)$ random neighbors from $v$ to in $V'$, and $v$ has degree $|V'|-1$ inside $V'$, in expectation, the load of each edge in $V'$ incident to $v$ is $\tilde O(1)$. This means every edge $e$ in $V'$ will have $\tilde O(1)$ number of instances $G'_s$ use $e$ to communicate. We can now run all instances with $\tilde O(1)$ congestion. We present the full algorithm in \Cref{sec:mnnccn}.

\begin{remark}
If the cluster $V'$ is not a clique, but close enough to a clique, then we simulate a clique communication with $O(\log n)$ factor overhead in communication, before running the clique algorithm. We formalize the algorithm in \Cref{sec:virtual clique}.
\end{remark}

\subsection{Parallel Algorithm in Matrix Multiplication Work}\label{subsec:overviewparallel}

Recall that in the last section, we explain how to solve minimum isolating cut and \mnnccn{} by reductions to s-t vertex connectivity while preserving both the number of edges and vertices, in \congest{} model. For the case of minimum isolating cut, one can easily check that the idea for preserving the number of vertices described in the last section also works in the \pram{} model, and thus we can solve s-t vertex connectivity using a reduction to bipartite matching which can be solved in $O(n^\omega)$ work.

We now focus on solving the case of \mnnccn{}. We will use the same local graph $G_{V'}$.  Recall that the algorithm that sparsifies $G_{V'}$ into $G'_s$ can result in $\hO{m} = O(n^2)$ vertices in total, which is not suitable for our purpose.

In this section, we explain how to solve \mnnccn{} in matrix multiplication time by utilizing \emph{convex embedding} introduced in \cite{LinialLW88}. The details can be found at \Cref{sec:parallelimplementation mm work}. Throughout this section, we let $n_0 = |N_G[V']|+1$ be the number of vertices in $G_{V'}$. The goal is to solve \mnnccn{} in $\tilde O(n_0^\omega)$ work (and polylog depth).  For simplicity, we assume that $|N(V')| \leq |V'|$. Later we show how to remove this assumption.

Let \(\mathbb{F}\) be a finite field. For \(k \geq 0\), the space \(\mathbb{F}^k\) is a \(k\)-dimensional linear space over \(\mathbb{F}\). Let \(X = \{ x_1, \ldots, x_n \}\) be a finite set of points within \(\mathbb{F}^k\). The \emph{affine hull} of \(X\) is defined as
\[
\text{aff}(X) = \left\{ \sum_{i=1}^k c_i x_i : x_i \in X \text{ and } \sum_{i=1}^k c_i = 1 \right\}.
\]
The rank of \(X\), represented as \(\text{rank}(X)\), is one plus the dimension of \(\text{aff}(X)\). Specifically, if \(\mathbb{F} = \mathbb{R}\), we refer to the \emph{convex hull} of \(X\), denoted \(\text{conv}(X)\).
For any sets $V, W $, any function $f : V \rightarrow W$, and any subset $U \subseteq V$, we denote $f(U) := \{ f(u) \colon u \in U \}$.

\begin{definition}[Convex $X$-embedding~\cite{LinialLW88}]
For any $X \subset V$, a convex  $X$-embedding of
a graph $G = (V,E) $ is a function $f : V  \rightarrow \mathbb{R}^{|X|-1}$
such that for each $v \in V \setminus X$, $f(v) \in \text{conv}(f(N_G(v)))$.
\end{definition}

We defer the details of construction to \Cref{sec:parallelimplementation mm work}. For now, let us assume that we are given $N_{G_{V'}}(t)$-convex embedding $f_t$ in graph $G_{V'}$ such that for all $x \in V'$, $\kappa_{G_{V'}}(s,t) = \rank (f_t(N_{G_{V'}}(x))).$ The task now is to compute the rank of $f_t(N_{G_{V'}}(x))$, which can be done in $\tilde O(n_0^\omega)$ time. In total, it would take $\tilde O(\frac{|V'|}{|\LL|} \cdot  n_0^\omega)$ work to compute $\kappa_{G_{V'}}(s,t)$ for all instances of $s$-source \mn{} via this approach.

Here is the speedup we can exploit: Since $V'$ has neighborhood difference at most $\hO{|L^*|}$, and thus for any $s, s' \in V'$, $|N(s) \triangle N(s')| \leq \hO{|L^*|}$. Therefore, $\rank (f_t(N_{G_{V'}}(s')))$ can be obtained from \emph{low-rank} updates of $\rank (f_t(N_{G_{V'}}(s)))$ given we preprocess the matrix representing $f_t(N_{G_{V'}}(s)$.  Here, low-rank updates correspond to changing $\hO{|L^*|}$ columns of the matrix representing $f_t(N_{G_{V'}}(s))$ since $|N(s) \triangle N(s')| \leq \hO{|L^*|}$.

In terms of low-rank updates, we can use a \emph{dynamic matrix data structure} to support this operation (see \Cref{lem:matrix ds} for the formal definition). Namely, we can preprocess the matrix $M_s$ representing $f_t(N_{G_{V'}}(s))$ in $\tilde O(n_0^\omega)$ time so that given $s' \in V'$, we can decide if $\rank(f_t(N_{G_{V'}} (s'))) \geq k$ in $T_{\MM}(|\LL|,n_0,|\LL|)$ time where $T_{\MM}(n,k,r)$ be the number of field operations needed to multiply an $n \times k$ matrix with a $k \times r$ matrix in $O(\log^2n)$ depth.

To summarize, using the matrix data structure, the total work to compute $\kappa_{G_{V'}}(s,t)$ for $\tilde O(\frac{|V'|}{|\LL|})$ instances of $s$-source \mn{} is at most
\[ \tilde O(n_0^\omega) + \tilde O\left(\frac{|V'|}{|\LL|}\right) \cdot T_{\MM}(|\LL|,n_0,|\LL|)  \leq \tilde O(n_0^\omega) = \tilde O(|V'|^\omega). \]

\subsection{Communication Complexity}\label{subsec:overviewcc}

When implementing our framework in the communication model, common-neighborhood clustering (as described in \cref{subsec:overviewcongest}) works well. The problem we face is solving \mnnccn{}: the best s-t vertex connectivity algorithm~\cite{blikstad2022nearly} takes $n$ as the parameter, but in the worst case we need to run $n$ sparse instances, each with $n$ vertices. One might hope that there is a way to solve this issue with an almost-perfect reduction just like we did in \congest{}. However, we proved that this is not true by showing an $\Omega(n^{1.5})$ lower bound.

\paragraph{An $\Omega(n^{1.5})$ lower bound.} We will reduce the following problem to vertex connectivity. Alice and Bob will get $\sqrt{n}$ instances of \emph{\subtri}, each containing $\sqrt{n}$ instances of \emph{subset} problems with length $\sqrt{n}$. The inputs to Alice and Bob are denoted as $A^{(i)}_1,A^{(i)}_2,...,A^{(i)}_{\sqrt{n}}\subseteq[\sqrt{n}]$ and $B^{(i)}_1,B^{(i)}_2,...,B^{(i)}_{\sqrt{n}}\subseteq[\sqrt{n}]$ from $i=1$ to $i=\sqrt{n}$, where it is guaranteed that $|A^{(i)}_j|=\sqrt{n}/2$ for any $i,j$. They want to determine whether there exists $i$ such that $B^{(i)}_j\subseteq A^{(i)}_j$ for any $j$. This problem has randomized communication lower bound $\Omega(n^{1.5})$ (see \cref{lem:ORANDsd}). Now we show how to reduce this problem to vertex connectivity.

We let $G=(V_1\cup V_2...\cup V_{\sqrt{n}}\cup U,E)$ where $V_i$ is a clique with $\sqrt{n}$ vertices, and $U$ is a clique with $n$ vertices. We connect $V_i$ to $U$ based on $A^{(i)}_1,...,A^{(i)}_{\sqrt{n}}$ and $B^{(i)}_1,...,B^{(i)}_{\sqrt{n}}$, such that $N(V_i)$ has size $n$ iff. $B^{(i)}_j\subseteq A^{(i)}_j$ for any $j$, while making sure that any other cut has size at least $n+1$. The way of connecting $V_i$ to $U$ is as follows: split $U$ into $\sqrt{n}$ pieces $U_1,...U_{\sqrt{n}}$, and connect the $j$-th vertex in $V_i$ (denoted by $v^{(i)}_j$) to $U_j$ according to $B^{(i)}_j$; connect the $j$-th vertex in $V_i$ to $U_k$ for any $k\not=j$ according to $A^{(i)}_k$. See~\cref{fig:n15bound} for an illustration. The details of this construction and proofs are in~\cref{subsec:n15lowerbound}.

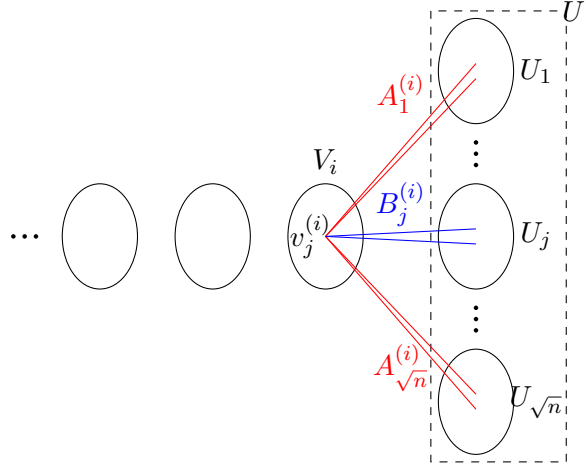
\begin{figure}[!ht]
    \centering
\begin{tikzpicture}

        \foreach \z in {0.15,0,-0.15}
        \fill (-4+\z,0) circle (0.03);

    \draw (-3,0) ellipse (0.5cm and 0.7cm);
    \draw (-1.5,0) ellipse (0.5cm and 0.7cm);

    \draw (0,0) ellipse (0.5cm and 0.7cm);
    \node[above] at (0,0.7) {$V_i$};
    \node at (-0.2,0) {$v^{(i)}_j$};

    \draw[dashed] (1.4,3) rectangle (3.2,-3);

    \foreach \y in {2.2, 0, -2.2}{
        \draw (2,\y) ellipse (0.5cm and 0.7cm);
    }

    \foreach \y in {1.1, -1.1}{
        \foreach \z in {0.15,0,-0.15}
        \fill (2,\y + \z) circle (0.03);
    }

    \foreach \y in {2.2, -2.2}{
        \draw[red] (0,0) -- ([yshift=-0.1cm]2,\y);
        \draw[red] (0,0) -- ([yshift=0.1cm]2,\y);
    }
    \draw[blue] (0,0) -- ([yshift=-0.1cm]2,0);
    \draw[blue] (0,0) -- ([yshift=0.1cm]2,0);

    \node[red, above] at (1,1.5) {$A^{(i)}_1$};
    \node[blue, above] at (1,0) {$B^{(i)}_j$};
    \node[red, above] at (1,-2.2) {$A^{(i)}_{\sqrt{n}}$};

    \node at (3.3,3) {$U$};
    \node at (2.8,2.2) {$U_1$};
    \node at (2.8,0) {$U_j$};
    \node at (2.8,-2.2) {$U_{\sqrt{n}}$};
\end{tikzpicture}
\captionsetup{width=0.8\textwidth}

    \caption{Each ellipsoid (and the set $U$) represents a clique. Every node in $V_i$ for $i=1$ to $\sqrt{n}$ has edges to $U$ as illustrated in the graph.}\label{fig:n15bound}
\end{figure}

\paragraph{A simple nearly tight upper bound.} Now we describe a simple way to solve \mn{} when $V'$ is a common-neighborhood cluster in $\tO{n^{1.5}}$ communication. We break it into two cases based on the size of $\LL$.

When $|\LL|>n^{0.5}$, we can sample $\tO{\sqrt{n}}$ nodes and one of them will hit $\LL$. Recall that if we know a node in $\LL$, we can find $\LL$ in one s-t vertex connectivity call. Thus, in this case, the communication complexity os $\tO{n^{1.5}}$.

If $|\LL|<n^{0.5}$, for a common-neighborhood cluster $V'$, we can recover all edges adjacent to $V'$ by using $O(n^{1.5})$ communications in the following way: choose an arbitrary vertex $u\in V'$, Alice and Bob learn the set $N_G(u)$ by $O(n)$ communication, then for every $v\in V'-\{u\}$, Alice and Bob can learn $N_G(v)$ by $O(n^{0.5})$ communication since $|N(u)\triangle N(v)|\le n^{0.5}$. Knowing all edges adjacent to $V'$ suffices to solve \mnnccn{}.

\begin{remark}
    The algorithm for communication can be easily implemented in the streaming model. Moreover, there is an even simpler communication and streaming protocol in $\tO{n^{1.5}}$ total communication that does not require common-neighborhood clustering. For completeness, we present it in~\cref{sec:streaming}.
\end{remark}

\section{Preliminaries} \label{sec:prelims}
\subsection{Terminologies} \label{sec:terminology}

\paragraph{Basic graph terminologies.}

An (undirected) graph is denoted by $G=(V,E)$ where $V$ is the vertex (or node) set and $E\subseteq \{\{x,y\}\mid x,y\in V,x\not=y\}$ is the edge set\footnote{All graphs in this paper are simple, i.e., without multi-edges and self-loops, as written in the definition.}. Following the convention, we abuse the notations a bit and will also use $(x,y)$ to denote an edge, which is a set $\{x,y\}$. $x,y$ are the \emph{endpoints} of the edge $(x,y)$. $N_G(u)$ is the set of neighbors of $u$ defined by $\{v\in V\mid \{u,v\}\in E\}$. We omit $G$ in the subscript if $G$ is clear from the context. We also define $N[u]=N(u)\cup\{u\}$. The \emph{degree} of a vertex $u$ is defined as $|N(u)|$. The \emph{minimum degree} of $G$ is defined as $\min_{v\in V}|N(v)|$. Define $E(A,B)=\{\{x,y\}\in E\mid x\in A,y\in B\}$.

For a vertex set $V'$, $G[V']$ is called the induced subgraph of $G$ on $V'$ defined by $G[V']=(V',\{(u,v)\in E\mid u,v\in V'\})$.

Throughout the paper, we use $n$ to denote $|V|$ and $m$ to denote $|E|$. We write $\tO{f}=O(f\cdot\log^cn)$ for some constant $c$, and $\hO{f}=f\cdot n^{o(1)}$. It is important that $n$ represents $|V|$ here, which is roughly the input size.

\paragraph{Set terminologies.} We define $[z]=\{1,2,...,z\}$. For two sets $A,B$, we use $A-B$ or $A\backslash B$ to denote $\{a\in A\mid a\not\in B\}$. We define $A\triangle B$ as the symmetric difference $(A-B)\cup(B-A)$. We say $(A_1,...,A_z)$ is a partition of $A$ if $\cup_{i\in[z]}A_i=A$ and $A_i\cap A_j=\emptyset$ for any $i\not=j$.

\paragraph{Vertex cut.} A \emph{vertex cut} (or simply \emph{cut} in this paper) is a partition of $V$ denoted by $(L,S,R)$ where (1) $L \neq \emptyset, R \neq \emptyset, E(L,R) = \emptyset$  or (2) $|S| \geq n-1$. The size of a vertex cut $(L,S,R)$ is defined as $|S|$. The \emph{minimum} vertex cut refers to the vertex cut with the smallest size. For a vertex $s\in V$, The \emph{minimum $t$-sink vertex cut} is defined as a vertex cut $(L,S,R)$ with $t\in R$ which minimizes $|S|$.

We say a vertex cut $(L,S,R)$ as a $(u,v)$-vertex cut if $u\in L,v\in R$. For two vertex sets $A,B\subseteq V$, we say $(L,S,R)$ is an $(A,B)$-vertex cut if $A\subseteq L,B\subseteq R$.
In this case, we say $S$ is a $(u,v)$-separator or $(A,B)$-separator. $S$ is called a \emph{separator} (or a \emph{vertex cut}) if $S$ is a $(u,v)$-separator for some $u,v\in V$. We use $\kappa_G(u,v)$ to denote the size of the smallest $u,v$-separator in $G$, and $\kappa(G)$ to denote the smallest separator in $G$. When there are no $(u,v)$-separator or no separator in $G$, we define $\kappa_G(u,v)$ or $\kappa(G)$ as $n-1$. In this case, we say any $n-1$ nodes is a separator.

For a vertex set $X\subseteq V$, we say $X$ has \emph{\nd{}} $d$ if for any $u,v\in X$ we have $|N_G(u)\triangle N_G(v)|\le d$. The following lemma is crucial to our algorithm.

\begin{lemma}\label{lem:unbalancedclose}
    Suppose $(L,S,R)$ is a minimum vertex cut, then the \nd{} of $L$ is at most $2|L|$.
\end{lemma}
\begin{proof}
    Since $S$ is the minimum vertex cut, we have $|N_G(u)|\ge|S|$ for any $u\in V$, otherwise $N_G(u)$ is a smaller vertex cut. Moreover, for $u,v\in L$, $N_G(u)\cup N_G(v)\subseteq L\cup S$, otherwise $S$ is not a vertex cut. Thus, $|N_G(u)\triangle N_G(v)|=  2|N_G(u)\cup N_G(v)| -|N_G(u)|-|N_G(v)|\le 2(|S|+|L|)-2|S|=2|L|$.
\end{proof}

\subsection{Computational Models} \label{sec:model-problems}
\paragraph{\pram{} model.} In \pram{} model, we have a set of processors and a shared memory space. Each processor can independently read and write on the shared memory space or do other unit operations. The input is given initially on the shared memory space, and the processors are required to jointly compute a specific problem given the input. The complexity is measured by \emph{work} and \emph{depth}, where work is measured by the \textbf{total} amount of unit operations performed by all the processes, and depth is measured by the time consumed before the output is generated.

\paragraph{\congest{} model.} In the distributed \congest{} model, we have an initial graph $G = (V,E)$ of $n$ nodes with diameter $D$ where each node has a unique ID and unlimited computational power. Initially, each node only knows its neighbors. The algorithm runs in synchronous rounds. For each round, each node can exchange $O(\log n)$ bits of information to its neighbors. The goal is to design an algorithm that minimizes the number of rounds needed to determine the output, for example, the vertex connectivity of $G$.

\paragraph{Two-party communication model.} In the two-party communication model (or simply communication model), the edge set $E$ is arbitrarily partitioned into two sets $E=E_1\cup E_2$ where $E_1\cap E_2=\emptyset$. Two players Alice and Bob both know the vertex set $V$, and Alice knows the set $E_1$, and Bob knows the set $E_2$. Their goal is to exchange as few bits as possible to find some intrinsic property of the graph $G=(V,E)$.
\paragraph{Graph streaming model.} In the graph streaming model (or simply streaming model), a graph $G =(V, E)$ is given to the algorithm as an arbitrarily ordered stream of edges. The algorithm has limited space and can read this stream in sequential passes. The goal is to minimize both the space usage and the number of passes. \emph{Semi-streaming} refers to the case when the space is restricted to $\tOh{n}$.

\begin{remark}[public and private randomness]\label{rem:publicrandomness}
    In the two-party communication model and \congest{} model, \emph{public randomness} refers to the case when every party (or node) can access the same random bits, \emph{private randomness}, in contrast, does not allow parties (or nodes) to share the randomness. According to Newman's Theorem~\cite{Newman91}, to simulate public randomness using private randomness, parties only need to share an additional $O(\log n)$ bits of message (which can be broadcasted in \congest{} model). Thus, in this paper, we assume access to public randomness when devising algorithms in communication or \congest{} model.
\end{remark}

\subsection{Problem Definitions}\label{sec:problemdefinitions}
We define the problems considered in this paper as follows.
We say a (randomized) algorithm solves a problem if it outputs the correct answer with high probability (or simply w.h.p.), i.e., with probability at least $1-\frac{1}{n^c}$ for an arbitrarily large constant $c$. All algorithms in this paper are randomized except with explicit clarification.

\paragraph{Vertex connectivity (\UDVC{}).} Given an undirected graph $G=(V,E)$, output a minimum vertex cut.
\paragraph{Single sink unidrected vertex connectivity (\SSUDVC{}).} Given a undirected graph $G=(V,E)$ and a sink vertex $t\in V$, output a minimum $t$-sink vertex cut.
\paragraph{s-t vertex connectivity (s-t \UDVC{}).} Given an undirected graph $G=(V,E)$ and $s,t\in V$, output a minimum $(s,t)$-vertex cut.
\paragraph{S,T-vertex connectivity (S-T \UDVC{}).} Given an undirected graph $G = (V,E)$, and two  vertex sets $A,B \subseteq V$, compute the minimum $(A,B)$-vertex cut.
\paragraph{Subgraph problems in \congest{}.} For subgraph S-T vertex connectivity (or other subgraph problems) in \congest{}, we are given a network along with a subgraph $H$ where every node knows its adjacent edges in $H$, and we want to solve the problem on $H$ while other edges of $G$ are only used for communication purposes.
\paragraph{Bipartite maximum matching (\BMM{}).} Given a bipartite graph $(A,B,E)$ (which is defined as $E\subseteq A\times B$ where $A\cup B$ is the vertex set), output the maximum matching (a matching is defined as a set of edges with mutually disjoint endpoints).

\paragraph{Bipartite minimum vertex cover (\BMC{}).} Given a bipartite graph $(A,B,E)$, output the minimum vertex cover (a vertex cover is defined as a set of vertices such that every edge has at least one endpoint in this set).

\subsection{Sketching}\label{subsec}

We will use the standard linear sketching algorithms from~\cite{CormodeF14}.

\begin{theorem}[Section 3 in~\cite{CormodeF14}]\label{thm:sketching}
For any numbers $n$ and
$s$, there is an algorithm that preprocesses in $\tO{s}$ work and $\tO{1}$ depth and
then, given any vector $v\in \{-1,0,1\}^{n}$, return a sketch $\sk{s}(v)\in\mathbb{Z}^{\tO{s}}$
in $\tO{\|v\|^2_2}$\footnote{Notice that $\|v\|^2_2$ is basically the number of non-zero entries of $v$.} work and $\tO{1}$ depth and guarantees the following w.h.p. (as long as
the number of recovery operations is $\poly(n)$).
\begin{itemize}
\item If $\|v\|^2_2\le s$, then we can recover $v$ from $\sk{s}(v)$
in $\tO{s}$ time and $\tO{1}$ depth. (More specifically, we obtain all non-zero entries
of $v$ together with their indices).
\item Otherwise, if $\|v\|^2_2>s$, then the algorithm returns $\bot$.
\end{itemize}
Moreover, the sketch is \emph{linear}, i.e.~$\sk{s}(u+v)=\sk{s}(v)+\sk{s}(u)$
for any $u,v\in \mathbb{Z}^{n}$.
\end{theorem}

\begin{remark}\label{rem:sketchingimplementation}
    The algorithm in \cite{CormodeF14} can be naturally implemented in the streaming model, in the sense that $v$ is given in a stream, where each time a non-zero entry of $v$ is given, and the algorithm needs to compute the sketching using $\tO{s}$ space in total after all entries of $v$ is revealed. An easy way to see this fact is from the linearity of sketching: we can view $v$ as a summation of vectors where each vector only contains one non-zero entry, and once a non-zero entry is revealed, we calculate the corresponding sketching for this vector and add it to the final sketching. This observation also implies the parallel complexity, where we can compute the sketching for each entry separately, and in parallel add them together.
\end{remark}

The following corollary is from~\cref{thm:sketching} combined with sampling.
\begin{corollary} \label{cor:approxsketching}
For any number $n$, there
is an algorithm that preprocesses in $\tO{n}$ time and then, given
any vector $v\in\{-1,0,1\}^n$, return a sketch $\ska(v)\in\bbR^{\tO{1}}$
in $\tO{\|v\|^2_2}$ time such that
whp. (as long as
the number of recovery operations is $\poly(n)$)
we can recover an approximate size of $\|v\|^2_2$ (a value between $0.9\|v\|^2_2$ and $1.1\|v\|^2_2$) in $\tO{1}$ work.
Moreover, the sketch is \emph{linear}, i.e.~$\ska(u+v)=\ska(v)+\ska(u)$
for any $u,v\in\bbR^{n}$.
\end{corollary}
\begin{proof}
    For any $i\in[\lfloor\log n\rfloor]$, we sample $2^i$ coordinates denoted by $I_i$. Given any vector $v\in\bbR^n$, we write $v_{I_i}$ as the vector restricted to $I_i$. For every $i$ we use~\cref{thm:sketching} with $n=|I_i|,s=\log^2n$ to compute $\sk{s}(v_{I_i})$. The $\sk{s}(v_{I_i})$ for all $i$ comprises our $\ska(v)$.

    Given $\ska(v)$, we look at the first $i$ such that when recovering from $\sk{s}(v_{I_i})$ the algorithm returns a value larger than $100\log n$ and not $\bot$. Such an $i$ must exist since we set $s=\log^2n$: when $i$ is large enough so that the expected values of $I_i$ hitting a non-zero entry of $v$ is roughly $500\log n$, according to the Chernoff bound, it will return a value larger than $100\log n$. Moreover, this value provides a good approximation of $\|v\|^2_2$ according to Chernoff bound.
\end{proof}
When using~\cref{cor:approxsketching,thm:sketching}, we sometimes will write $\sk{s}(S)$ where $S\subseteq V$ is a vertex set instead of a vector in $\{0,1\}^n$. In this case, we consider $S$ as a vector in $\{0,1\}^n$ by treating each coordinate of the vector as a vertex in $V$, and each coordinate is $1$ if and only if the corresponding vertex is in $S$.

\section{Parallel and Distributed Common-Neighborhood Clustering}\label{sec:snc}

We first give some definitions of terminologies used in the algorithm. The algorithm will compute $\ska(N_G(u))$ for any $u\in V$ in the beginning. During the algorithm, we say two nodes $u,v\in V$ are \nc{d} if the approximate size recovered from $\ska(N_G(u))-\ska(N_G(v))$ (see~\cref{cor:approxsketching}) is at most $d$.
We say $Y$ is adjacent to $X$ if $Y\cap X=\emptyset$ and there is an edge connecting a vertex in $Y$ and a vertex in $X$.

\begin{algorithm}[!ht]
\caption{$\cC\leftarrow$\cnc{}$(G,\ell)$}\label{alg:snc}
 \KwData{An undirected graph $G=(V,E)$, an integer $\ell$.}
 \KwResult{A set of clusters $\cC\subseteq 2^V$ satisfying~\cref{lem:snc}}
 $\cC\leftarrow \emptyset$\;
 Compute sketchings $\ska(N_G(u))$ for any $u\in V$\;
 \ForEach{$i=0,1,...,\log^{0.1} n,j=0,1,...,\log n$}
 {
    Let $d=2^{i\sqrt{\log n}}\ell,s\leftarrow 2^j$\;\label{letd}
    $\cP\leftarrow\{\{u\}\mid u\in V\}$\;\label{cpleftarrow}
    Mark each vertex $u\in V$ as \emph{center} independently at random with probability $1/s$\;\label{makreachu}
    \ForEach{$t_1=0,1,...,\log^{0.3}n$\label{foreacht1}}
    {
        $\cP'\leftarrow \cP$\;\label{cp'left}
        Assign each center $u\in V$ a vertex set $S_u=C\in\cP'$ where $u\in C$, delete $C$ from $\cP'$\;\label{assigneachu}
        \ForEach{$t_2=0,1,...,2^{\log^{0.9}n}$}
        {
            \ForEach{Center $u\in V$ and every $C\in\cP'$ that is adjacent to $S_u$}
            {
                    Pick an arbitrary $v\in C$, if $u,v$ are \nc{$9^{t_1}d$}, let $S_u\leftarrow S_u\cup C,\cP'\leftarrow \cP'-C$\;\label{pickan}
            }
        }
        Add $S_u$ to $\cC$ for any center $u\in V$\;\label{addsu}

        Mark each vertex set in $\cP$ as \emph{center} with probability $1/2^{\log^{0.8}n}$\;\label{markeachvertex}
        Assign each center $C\in\cP$ a vertex set $S_C=C$ and delete $C$ from $\cP$\;\label{assigneachC}
        For each center $C\in\cP$, pick an arbitrary node $c\in C$ denoted by $u_C$\;\label{foreachcenterC}
        \ForEach{$t_2=0,1,...,2^{\log^{0.9}n}$}
        {
            \ForEach{Center $C\in V$ and every $X\in\cP$ that is adjacent to $S_C$}
            {
                    Pick an arbitrary $v\in X$, if $u_C,v$ are \nc{$9^{t_1}d$}, $S_C\leftarrow S_C\cup X,\cP\leftarrow\cP-X$\;\label{pickar}
            }
        }
        Set $\cP\leftarrow\{S_C\mid \text{C is a center}\}$\;\label{setcp}
        }
 }
    \Return{$\cC$}
\end{algorithm}
\begin{lemma}[common neighborhood clustering]\label{lem:snc}
    There exists a randomized algorithm (\cref{alg:snc}) that given an undirected graph $G$ and an integer $\ell$, return $\cC$ satisfying the following properties
    \begin{enumerate}
        \item (sparse) for every $v\in V$, $v$ is included in at most $\tO{1}$ clusters in $\cC$,
        \item (common neighborhood) any $C\in\cC$ has \nd{} $\ell\cdot 2^{\log^{0.7}n}$ with high probability,
        \item (cover) for every $L\subseteq V$ such that $G[L]$ is connected and $L$ has \nd{} $\ell$, we have $\Pr[\exists C\in\cC,L\subseteq C]\ge \frac{1}{n^{o(1)}}$.
    \end{enumerate}
    The algorithm can be implemented in $\tO{m}$ work and $\hO{1}$ depth in \pram{} model, or $\hO{\sqrt{n}+D}$ rounds in \congest{} model.
\end{lemma}

\begin{proof}[Proof of~\cref{lem:snc} correctness]
    \textbf{(Sparse)} we need to following claim.
    \begin{claim}\label{cla:cpdisjoint}
        Throughout the algorithm, $\cP$ only contains disjoint vertex sets.
    \end{claim}
    \begin{proof}
        $\cP$ is generated in~\cref{setcp}, which contains $S_C$ for each center $C$, where $S_C$ is updated in~\cref{pickar}. Notice that if $X$ joins $S_C$, then it is deleted from the original $\cP$ and cannot join other $S_C$, thus, as long as the original $\cP$ contains vertex disjoint sets, the new $\cP$ also contains vertex disjoint sets. Initially $\cP$ contains all singleton vertices (see~\cref{cpleftarrow}), so the claim holds.
    \end{proof}

    $\cC$ is only updated in~\cref{addsu}, where $S_u$ is generated in~\cref{pickan}. If $C$ is included in $S_u$ then it is deleted from $\cP'$. So $S_u$ are disjoint vertex sets for centers $u\in V$ as long as $\cP'$ contains vertex disjoint sets, which is true from~\cref{cla:cpdisjoint}. Therefore, every execution of~\cref{addsu} contributes at most $1$ cluster for each node $v\in V$. There are $\tO{1}$ loops of~\cref{addsu}, which proves the (sparse) property.

    \textbf{(Common neighborhood)} We need the following lemma which bounds the \nd{} of every set in $\cP$.
    \begin{lemma}\label{lem:cpcommonneighborhood}
        At line~\cref{cp'left}, it always holds that for any $C\in\cP$, $C$ has \nd{} $9^{t_1}d/3$ w.h.p.
    \end{lemma}
    \begin{proof}
        We will prove it by induction. When $t_2=0$, $\cP=\{\{u\}\mid u\in V\}$ and the lemma trivially holds. $\cP$ will be changed at~\cref{setcp} to the next iteration, thus, we only need to prove that for any center $C$, $S_C$ has \nd{} $9^{t_1+1}d/3$, given that the \nd{} of $\cP$ is bounded by $9^{t_1}d/3$.

        Initially in~\cref{assigneachC} we have the \nd{} of $S_C$ being at most $9^{t_1}d/3$. Notice that in~\cref{pickar}, we only add vertices in $X$ to $S_C$ if there is $v\in X$ and $u_C\in C$ that are \nc{$9^{t_1}d$}. In other words, for every $u,v\in S_C$, their exists $u_1$ (in the same $X$ as $u$), and $v_1$ (in the same $X$ as $v$) such that both $u,u_1$ and $v_1,v$ has neighborhood symmetric difference at most $9^{t_1}d/3$ by induction hypothesis, and $u_1,u_C$ and $v_1,u_C$ has symmetric difference at most $1.1\cdot 9^{t_1}d$ w.h.p. according to the definition of \nc{$9^{t_1}d$}. Thus, we have
        \[|N(u)\triangle N(v)|\le 2\cdot (9^{t_1}d/3)+2\cdot 1.1\cdot 9^{t_1}d\le 9^{t_1+1}d/3\]

        This completes the induction step.
    \end{proof}

    Notice that $\cC$ is only updated in~\cref{addsu} where $S_u$ is updated in~\cref{pickan}, which initially contains a vertex set in $\cP$ and later will only includes vertex sets $C$ in $\cP'$ (also in $\cP$) such that there exists $v\in C$ which is \nc{$9^{t_1}d$} to $u$. Therefore, for any $u,v\in S_u$, there exists $u_1$ (in the same $C$ as $u$), $v_1$ (in the same $C$ as $v$) such that both $u,u_1$, $v_1,v$ has neighborhood symmetric difference at most $9^{t_1}d/3$ by~\cref{lem:cpcommonneighborhood}, and $u_1,u$ and $v_1,u$ has neighborhood symmetric difference at most $1.1\cdot 9^{t_1}d$ w.h.p. according to the definition of \nc{$9^{t_1}d$}. Thus, we have
    \[|N(u)\triangle N(v)|\le 2\cdot (9^{t_1}d/3)+2\cdot 1.1\cdot 9^{t_1}d\le 2^{O(\log^{0.1}n)}\cdot d=2^{\log^{0.7}n}\ell\]
    This completes the proof of property (common neighborhood).

    \textbf{(cover)} let $L\subseteq V$ be a set such that $G[L]$ is connected and has \nd{} $\ell$. We will prove that there exists $C\in\cC$ such that $L\subseteq C$ with probability at least $1/n^{o(1)}$. For an integer $d$, define $B'_{d}(L)$ as the largest $L'$ such that (i) for any $u\in L'$, there exists a $v\in L$ such that $|N_G(u)\triangle N_G(v)|\le d$, (ii) $G[L']$ is connected. In other words, $B'[d](L)$ can be found by doing BFS from $L$ while only including vertices that has neighborhood symmetric difference at most $d$. Let $i$ be the smallest natural number such that $2^{\log^{0.9}n}\cdot |B'[2^{i\sqrt{\log n}}\ell/4](L)|\ge |B'[2^{(i+1)\sqrt{\log n}}\ell/4](L)|$. Notice that if the inequality does not hold, then $|B'[2^{(i+1)\sqrt{\log n}}\ell/4](L)|$ is $2^{\log^{0.9}n}$ multiplicative larger than $|B'[2^{i\sqrt{\log n}}\ell/4](L)|$. Since the total number of nodes is $n$, $i$ cannot be larger than $\log^{0.1}n$. Thus, in~\cref{letd}, there exists a loop such that
    \begin{enumerate}
        \item $i$ satisfies $2^{\log^{0.9}n}\cdot |B'[2^{i\sqrt{\log n}}\ell/4](L)|\ge |B'[2^{(i+1)\sqrt{\log n}}\ell/4](L)|$, notice that $B'[2^{i\sqrt{\log n}}\ell/4](L)$ has \nd{} at most $d$,
        \item $s$ is roughly equal to the size of $B'[2^{(i+1)\sqrt{\log n}}\ell/4](L)$ (with multiplicative error 2), this is because we set $s=2^j$ where $j$ is ranging from $0$ to $\log n$.
    \end{enumerate}

    We will prove that in this loop, we will add a $C$ to $\cC$ such that $L\subseteq C$ with probability at least $1/n^{o(1)}$. Let us write $A=B'[2^{i\sqrt{\log n}}\ell/4](L)$ and $B=B'[2^{(i+1)\sqrt{\log n}}\ell/4](L)$. It is clear that $|B|\le s\le 2|B|$ and $2^{\log^{0.9}n}\cdot |A|\ge |B|$.

    According to the two properties, at~\cref{makreachu}, the following two events happen simultaneously with probability at least $1/n^{o(1)}$ (i) there is exactly one center inside $B'[2^{i\sqrt{\log n}}\ell/4](L)$, (ii) every node in $B'[2^{(i+1)\sqrt{\log n}}\ell/4](L)-B'[2^{i\sqrt{\log n}}\ell/4](L)$ is not a center. In what follows we will assume the two events.

    We first show that if $\cP$ do not split $A$ too much and covers $A$, then we can add a $C$ with $L\subseteq C$ at~\cref{addsu}.
    \begin{lemma}\label{lem:cPcondition}
        If $|\{C\in\cP\mid C\cap A\not=\emptyset\}|\le 2^{\log^{0.9}n}$ and $A\subseteq \cup_{C\in\cP}C$, let $u$ be the only center inside $A$, then $A\subseteq S_u$ at~\cref{addsu}.
    \end{lemma}
    \begin{proof}
        Let $C\in\cP$ be a set with $C\cap A\not=\emptyset$. We first show that $C$ can never be included in $S_v$ for $v\not=u$. If on the contrary there is a center $v\not=u$ such that $C$ joins $S_v$ at~\cref{pickan}, we will prove a contradiction. We first prove that for any $w\in S_v$, $w$ and $v$ are $9^{t_1}d/3+9^{t_1}d$ close. This is because $w\in C'\in\cP$ for some $C'$ which has \nd{} $9^{t_1}d/3$ according to~\cref{lem:cpcommonneighborhood}, and there exists $w'\in C'$ which is \nc{$9^{t_1}d$} to $v$ according to~\cref{pickan}. Since $C\cap A\not=\emptyset$, let $a\in C\cap A$ be an arbitrary vertex. Now we have every vertex in $S_v$ is \nc{$2(9^{t_1}d/3+9^{t_1}d)$} to $a$. Moreover, $G[S_v]$ is connected because at~\cref{pickan} $S_v$ only includes connected subgraphs that is adjacent to it. Remember the definition $A=B'[2^{i\sqrt{\log n}}\ell/4](L)$, which implies that any node in $S_v$ is \nc{$\left(2(9^{t_1}d/3+9^{t_1}d)+2^{i\sqrt{\log n}}\ell/4\right)$} to a vertex in $L$. Notice that $t_1\le \log^{0.3}n$ and $d=2^{i\sqrt{\log n}}\ell/4$, so $\left(2(9^{t_1}d/3+9^{t_1}d)+2^{i\sqrt{\log n}}\ell/4\right)\le 2^{(i+1)\sqrt{\log n}}\ell/4$, which contradicts the fact that $v\not\in B=B'[2^{(i+1)\sqrt{\log n}}\ell/4](L)$.

        Now we show that $C$ can be added to $S_u$ at~\cref{pickan} as long as $C$ is adjacent to the current $S_u$. Let $v\in C$ be an arbitrary vertex, let $a\in C\cap A$, according to~\cref{lem:cpcommonneighborhood}, $a$ and $v$ are \nc{$9^{t_1}d/3$}. According to the definition of $A$, $a$ and $u$ are \nc{$\left(2\cdot 2^{i\sqrt{\log n}}\ell/4+\ell\right)$}. Thus $v$ and $u$ are \nc{$9^{t_1}d$}.

        Since we have $A\subseteq \cup_{C\in\cP}C$ and $G[A]$ is connected, every time~\cref{pickan} is executed, if $A\not\subseteq S_u$ yet, since $u\in S_u$, $S_u$ must be adjacent to at least one $C\in\cP$ with $C\cap A\not=\emptyset$, and it is added to $S_u$. After $2^{\log^{0.9}n}$ rounds, we have $A\subseteq S_u$.
    \end{proof}

    It remains to prove that at some point we will have $|\{C\in\cP\mid C\cap A\not=\emptyset\}|\le 2^{\log^{0.9}n}$ and $A\subseteq \cup_{C\in\cP}C$. Let us inspect how $\cP$ changes.

    Initially $\cP=\{\{u\}\mid u\in V\}$ and it holds that $A\subseteq \cup_{C\in\cP}C$. If $|A|\le 2^{\log^{0.9}n}$ then we are done. Otherwise, we can use the following lemma.

    \begin{lemma}\label{lem:changeofcP}
        If $|\{C\in\cP\mid C\cap A\not=\emptyset\}|> 2^{\log^{0.9}n}$ and $A\subseteq \cup_{C\in\cP}C$, then at~\cref{setcp}, we have $A\subseteq\cup_{C\in\cP}C$ w.h.p.
    \end{lemma}
    \begin{proof}
        For any $C\in\cP$ with $C\cap A\not=\emptyset$, we will prove that there exists a center $C'\in\cP$ such that $C$ is added to $S_{C'}$ at~\cref{pickar}. Combined with the fact that $A\subseteq \cup_{C\in\cP}C$ initially, the lemma is proved.

        Let $(C_0,C_1,...,C_{2^{\log^{0.9}n}})$ be a sequence of different vertex sets in $\cP$ such that $C_0=C$ and for any $i>0$, $C_i$ is adjacent to $C_{i-1}$ and $C_i\cap A\not=\emptyset$. Since $|\{C\in\cP\mid C\cap A\not=\emptyset\}|> 2^{\log^{0.9}n}$, such a sequence must exists.
        At~\cref{markeachvertex}, we mark each element in $\{C\in\cP\mid C\cap A\not=\emptyset\}$ as center independently with probability $1/2^{\log^{0.8}n}$, thus, w.h.p. there exists $i$ such that $C_i$ is marked as center. Let $I$ be the smallest index such that $C_I$ is inside some $S_{C'}$ for some center $C'$. We will prove that $I$ decreases by at least $1$ at each iteration of~\cref{pickar} in the next paragraph. If this is true, $I$ will be decreased to $0$ in the end and the lemma is proven.

        Since $C_{I-1}$ is adjacent to $C_I$, suppose $C_I\subseteq S_{C'}$ for some center $C'$, we only need to prove that when we pick an arbitrary $v\in C_{I-1}$, $u_{C'}$ (which is an arbitrary vertex in $C'$) and $v$ are \nc{$9^{t_1}d$}. Notice that both $C_{I-1}$ and $C'$ intersect $A$ by at least one vertex, let the vertices be $a,b$ separately. $v,a$ are \nc{$9^{t_1}d/3$}, $b,u_{C'}$ are \nc{$9^{t_1}d/3$} according to~\cref{lem:cpcommonneighborhood}; $a,b$ are \nc{d} according to the definition of $A$. Thus, as long as $t_1>0$, $u_{C'}$ and $v$ are \nc{$9^{t_1}d$}. For the case of $t_1=0$, every vertex set of $\cP$ only contain $1$ vertex which means $v,a$ and $b,u_{C'}$ are identical, so $v,u_{C'}$ are \nc{$d$}. Thus, $C_{I-1}$ will be added to some $S_{C'}$, which means $I$ is decreased by at least $1$.
    \end{proof}

    We also need the following lemma which makes sure $|\{C\in\cP\mid C\cap A\not=\emptyset\}|> 2^{\log^{0.9}n}$ will not always stay large.
    \begin{claim}\label{cPbecomesempty}
        At the last loop of~\cref{foreacht1}, $\cP$ becomes empty w.h.p.
    \end{claim}
    \begin{proof}
        In~\cref{markeachvertex}, each element becomes a center with probability $1/2^{\log^{0.8}n}$, which is the only way that element can contribute to one element to $\cP$ in the next loop according to~\cref{setcp}. Thus, we can think of it as each element ``surviving'' at each loop with probability $1/2^{\log^{0.8}n}$, and after $\log^{0.3}n$ loops it vanishes w.h.p.
    \end{proof}

    Now start with a $\cP$ with $A\subseteq \cup_{C\in\cP}C$, either we already have $|\{C\in\cP\mid C\cap A\not=\emptyset\}|\le 2^{\log^{0.9}n}$ which according to~\cref{lem:cPcondition} we are done; or we have $|\{C\in\cP\mid C\cap A\not=\emptyset\}|> 2^{\log^{0.9}n}$, which according to~\cref{lem:changeofcP}, we still have $A\subseteq \cup_{C\in\cP}C$ in the next iteration and we can repeat this argument. According to~\cref{cPbecomesempty}, there must exists a point when this argument ends, which gives us $|\{C\in\cP\mid C\cap A\not=\emptyset\}|\le 2^{\log^{0.9}n}$ and $A\subseteq \cup_{C\in\cP}C$.
\end{proof}

Now we talk about the parallel and distributed implementation of~\cref{alg:snc}. Most of the steps can be trivially implemented.

\begin{lemma}\label{lem:sncrunningtime}
    \cref{alg:snc} can be implemented in \pram{} in $\tO{m}$ work and $\hO{1}$ depth, and in \congest{} in $\hO{\sqrt{n}+D}$ rounds.
\end{lemma}
\begin{proof}[Proof of \pram{}]
    Computing sketchings uses $\tO{m}$ work and $\tO{1}$ depths according to~\cref{thm:sketching}. The loop for $i,j,t_1$ only contributes $\tO{1}$ to work and depth, so let us focus on one loop fixing $i,j,t_1$. Generating $\cP$ and marking centers can be trivially implemented. The non-trivial part is the loop for $t_2$, we implement it by contracting each vertex set in $\cP$, and running a BFS like procedure with root on every center, where the BFS only includes a vertex set $C$ if it satisfies the condition in~\cref{pickan}. BFS trees have depth bound $2^{\log^{0.9}n}$, each check for whether or not to include a vertex set into the BFS tree can be done in $\tO{1}$ work according to~\cref{thm:sketching}; when different BFS trees meet together, the shared node joins a BFS tree arbitrarily, this does not affect the analysis, and makes sure that the total work is $\tO{m}$ and the total depth is $\hO{1}$. After finding $S_u$ for each center $u\in V$, we again mark centers in $\cP$, and try to implement~\cref{pickar}. This step is also done by contracting each vertex set in $\cP$ and start BFS from every center which uses~\cref{pickar} to check whether or not to include a vertex set into a BFS tree. The total work is $\tO{m}$ and the total depth is $\tO{1}$.
\end{proof}

\begin{proof}[Proof of \congest{}]
    For any $u\in V$, we use shared randomness (see~\cref{rem:publicrandomness}) to locally compute $\ska(N_G(u))$ based on the information $N_G(u)$ which is known by $u$ initially. The loop for $i,j,t_1$ only contributes $\tO{1}$ to rounds, so let us focus on one loop fixing $i,j,t_1$.

    $\cP$ is maintained as follows: each vertex set in $\cP$ gets a unique ID which is known by every node in that vertex set. Initially, $\cP=\{\{u\}\mid u\in V\}$, and the ID is the ID for the specific node. Each vertex becomes a center independently, which can be done without communication. The non-trivial part is implementing~\cref{pickan}, which requires a center node $u$ to first find out $S_u$ (which is connected in $G$), and look at all $C\in\cP$ that are adjacent to $S_u$, then verifies whether or not to add them to $S_u$. We achieve this by well-known techniques for \congest{}: if a vertex set contains at most $\sqrt{n}$ vertices, we use broadcasting inside this vertex set to spread and gather information in $\sqrt{n}$ rounds; otherwise we use the whole network to do this in $D$ rounds. As long as all the vertex sets are disjoint, the total dilation is $O(D+\sqrt{n})$ and congest is $O(\sqrt{n})$ since there can be at most $\sqrt{n}$ sets with more than $\sqrt{n}$ nodes. We maintain the size of $S_u$ to determine which strategy to use. Now for each center $u$, it broadcast to all nodes in $S_u$ about $\ska(N_G(u))$, and for every $C\in\cP'$, it chooses a representative node $v_C$ and let all nodes in $C$ know $\ska(N_G(v_C))$. If there is an edge $(v_1,v_2)$ where $v_1\in S_u$ and $v_2\in C\in\cP$ (certifying that $C$ is adjacent to $S_u$), then $v_1$ and $v_2$ can communicate and decide if $C$ should be added to $S_u$ or not. After that, for each $C\in\cP'$, it chooses an arbitrary $S_u$ to join if there are many (this can be decided by broadcasting inside $C$), i.e., they mark their ID to be identical to the ID for $S_u$. This is repeated $2^{\log^{0.9}n}$ times, which in total contributes to $\hO{\sqrt{n}+D}$ rounds. Another non-trivial part is~\cref{pickar}, which is handled identically as~\cref{pickan}.

\end{proof}

\begin{remark}
    The implementation in \congest{} can be viewed as $\hO{1}$ rounds in \emph{minor-aggregation model} (see~\cite{RozhonGHZL22}). For any network, one round of minor-aggregation can be simulated by $O(\sqrt{n}+D)$ rounds in \congest{}, which implies our~\cref{lem:sncrunningtime}. For a more efficient network like a planner graph, it can be simulated faster ($\tO{D}$ rounds). It is a convenient framework for recent works of universal optimality (defined for weighted graphs). However, we are working on unweighted graphs and it is not clear if maximum bipartite matching exists a universally optimal algorithm or not, we do not explore the definition here.
\end{remark}

\section{A Schematic Reduction for Vertex Connectivity}\label{sec:framework}

In this section, we show a general framework for solving \SSUDVC{}, which can be implemented in different models, and prove its correctness. It essentially uses common-neighborhood clustering to reduce the problem to solving \emph{isolating cuts} and solving min-neighbor in a \emph{near-clique common-neighborhood} graph, which we define as follows. For convenience, we write the second problem \mnnccn{}.

\begin{definition}[Isolating Cuts Problem]\label{def:isolatingcutsproblem}
    In the isolating cuts problem, we are given an undirected graph $G=(V,E)$, vertex sets $C, T \subseteq V$ where $T \subseteq C$ is an independent set inside $C$, and the goal is to solve the following minimization problem.
    \begin{align*}
    \min_{L \subseteq C} \quad & |N_G(L)|
    \quad \text{s.t.} \quad  |L \cap T| = 1 \quad \mbox{and }   N_G(L) \cap T = \emptyset.
    \end{align*}
   The algorithm returns a minimizer $L^*$ along with its neighbors' size $|N_G(L^*)|$.
\end{definition}

\begin{definition}[Minimum Neighbor in a Near-Clique  Cluster (\mnnccn{})]\label{def:minneibornearclique}
    In this problem, we are given an undirected graph $G=(V,E)$, a vertex set $C \subseteq V$ where $N_G[C] \neq V$, and an integer $\ell$. The goal is to find a non-empty vertex set $L\subseteq C$ minimizing $|N_G(L)|$ and also return $|N_G(L)|$.
    The inputs have the following guarantees:
    \begin{enumerate} [nosep]
        \item (correct estimate). There exists $L\subseteq C$ which is a minimizer of $\min_{L' \subseteq C : L' \neq \emptyset} |N_G(L')|$ such that \[|L| \leq \ell \leq 2|L|.\]
        \item (near clique). For any $u\in C$,

            \begin{align} \label{eq:near clique prop}
                |C-N_G(u)|\le 2^{\log^{0.8}n}\cdot \ell,
            \end{align}
        \item (common neighborhood). the cluster $C$ has \nd{} $2^{\log^{0.7}n}\cdot \ell$. That is, for all $u,v \in C$,
        \begin{align}  \label{eq:cn diff prop}
        |N_G(u) \triangle N_G(v)| \leq 2^{\log^{0.7}n}\cdot \ell
        \end{align}
    \end{enumerate}
    If the input guarantees are unsatisfied, the algorithm returns an arbitrary $L\subseteq C$ and $|N_G(L)|$.
\end{definition}

\subsection{The Reduction} Let us denote the algorithm solving isolating cuts by \isocut$(G,C,T)$ and the algorithm solving \mnnccn{} by \mnnccnalg{}$(G,C,\ell)$.

\paragraph{The Schematic Algorithm.}  The framework for solving \SSUDVC{} is presented in \cref{alg:framework}.

\begin{algorithm}[!ht]
\caption{$S\leftarrow$\SSUDVC{}$(G,t)$}\label{alg:framework}
 \KwData{An undirected graph $G=(V,E)$, a vertex $t$.}
 \KwResult{A minimum $t$-sink vertex cut $(L,S,R)$.}
 \ForEach{$i=0,1,...,\log n$}
 {
    $\ell\leftarrow 2^i$\;
    $\cC\leftarrow$\cnc{}$(G,2\ell)$\;
    Let $C\leftarrow C-N_G[t]$ for every $C\in\cC$\;\label{letCleft}
    \ForEach{$C\in \cC$}
    {
        Let $T$ include each node in $C$ independently at random with probability $\frac{1}{\ell\cdot 2^{\log^{0.8}n}}$\;\label{letTinclude}
        Let $G'$ be a subgraph of $G$ only containing edges adjacent to $C$\;
        $L_{C,1},s_{C,1}\leftarrow$\isocut{}$(G',C,T)$\;
        $L_{C,2},s_{C,2}\leftarrow$\mnnccnalg{}$(G',C,\ell)$\;
    }
    Return $N_G(L_{C,i})$ where $s_{C,i}$ is minimized among all $C\in\cC$ and $i\in\{1,2\}$\;
 }
\end{algorithm}

The following lemma shows the correctness of \cref{alg:framework} based on the correctness of \isocut{} and \mnnccnalg{}.

\begin{lemma}[Correctness of the framework]\label{lem:framework pram congest}
    Given that \isocut{} and \mnnccnalg{} correctly output according to~\cref{def:isolatingcutsproblem} and~\cref{def:minneibornearclique},~\cref{alg:framework} returns a valid vertex cut, which is a minimum $t$-sink vertex cut with probability at least $1/n^{o(1)}$.
\end{lemma}
\begin{proof}
    Let $(L,S,R)$ be one of the minimum $t$-sink vertex cut such that $G[L]$ is connected (if it is not, take one connected component of $G[L]$ as $L$ which will not increase the cut size). There must exists a loop such that $|L|\le \ell\le 2|L|$. Notice that $N_G(L_{C,i})$ is always a valid $t$-sink vertex cut for any $C$ and $i$ since $L_{C,i}\subseteq V-N_G[t]$, so we only need to prove that in that loop, either $N_G(L_{C,1})$ or $N_G(L_{C,2})$ is a minimum $t$-sink vertex cut for some $C\in\cC$.

    According to~\cref{lem:unbalancedclose} and~\cref{lem:snc}, with probability at least $1/n^{o(1)}$ there exists $C\in\cC$ such that $L\subseteq \cC$ in~\cref{letCleft}. Let us focus on such $C$.

    We will prove that there are only two possible cases for $C$. Let $d=\ell\cdot 2^{\log^{0.7}n}$.
    \begin{claim}\label{lem:twocases}
        If $C$ has \nd{} $d$ and $L\subseteq C$, then either $|N(L)\cap C|\le 2d$, or $|C-N[L]|\le 3d$.
    \end{claim}
    \begin{proof}
        Let us assume $|N(L)\cap C|>2d$ and we will prove $|C-N[L]|\le 3d$.

        For an arbitrary $x\in L$, we have $N(x)\cap (C-N[L])=\emptyset$. Since every node $v\in C$ satisfies $|N(x)\triangle N(v)|\le d$, we have $|N(v)\cap (C-N[L])|\le d$. Thus, the total number of edges from $C\cap N(L)$ to $C-N[L]$ is at most $|C\cap N(L)|\cdot d$.

        Moreover, for an arbitrary $x\in L$, we have $|N(x)|\ge |N(L)|$ ($|N(L)|$ is minimum) and $N(x)\subseteq L\cup N(L)$, which implies $|(L\cup N(L))-N(x)|\le |L|\le \ell$. This gives $|N(L)-N(x)|\le \ell$. Thus, every node $v\in C$ satisfies $|N(L)-N(v)|\le \ell+d$, which implies $|N(v)\cap N(L)\cap C|\ge |N(L)\cap C|-(\ell+d)\ge |N(L)\cap C|/3$. The last inequality is due to $|N(L)\cap C|>2d$ and $d=\ell\cdot 2^{\log^{0.7}n}$. Therefore, the number of edges from $C-N[L]$ to $C\cap N(L)$ is at least $|C-N[L]|\cdot |N(L)\cap C|/3$.

        Combining the two inequalities from the last two paragraphs gives $|C-N[L]|\le 3d$.
    \end{proof}

    In the first case when $|N(L)\cap C|\le 2d$,~\cref{letTinclude} will make $|T\cap L|=1$ and $|T\cap S|=0$ with probability $1/n^{o(1)}$, to make sure $T$ is an independent set we simply delete vertices in $T$ that is adjacent to another vertex in $T$, which according to~\cref{def:isolatingcutsproblem} makes $N(L_{C,1})$ a minimum $t$-sink cut; in the second case when $|C-N[L]|\le 3d$, according to~\cref{def:minneibornearclique} $N(L_{C,2})$ is a minimum $t$-sink cut.
\end{proof}

\subsection{Implementation in \congest and \pram models} As applications of the reduction, we present algorithms that solve \isocut{} and \mnnccn{} in \pram{} and \congest{} models in \Cref{sec:isolating cut lemma parallel dist,sec:mnnccn}, which are summarized as in the following two lemmas.

\begin{lemma}\label{lem:refined isolating cut lemma}
Given a \isocut{} instance $(G,C,T)$ where the input graph has $m$ edges and $n$ vertices, there is an algorithm solving it correctly with high probability in
\begin{enumerate}
    \item \pram model, the algorithm can be implemented to run in work $\tO{W(m,|C|)}$ and depth $\tO{D(m,|C|)}$ where $W(m,n)$ and $D(m,n)$ are the work and depth of s-t vertex connectivity.
    \item  \congest model, the algorithm can be implemented to run in $\tilde O(R(m,n,D))$ rounds where $R(m,n,D)$ is the round complexity of subgraph S-T vertex connectivity and $D$ is the diameter of $G$; furthermore, the algorithm only communicates via the set of edges that is incident to $C$.
\end{enumerate}
\end{lemma}

\begin{lemma}\label{lem:mnnccn}
    Given a \mnnccn{} instance $(G,C,\ell)$ (\Cref{def:minneibornearclique}), it can be solved with high probability in
    \begin{enumerate}
        \item \pram{} model in $\tO{W(m,n)}$ work and $\tO{D(m,n)}$ depth where $W(m,n)$ and $D(m,n)$ are the work and depth of s-t vertex connectivity, satisfying that $W(m,n)$ is superadditive on $m$ (meaning $W(m_1+m_2,n)\ge W(m_1,n)+W(m_2,n)$ for every $m_1,m_2,n$),
        \item \congest{} model in $\tilde O(R(m,n,O(1))$ rounds where $R(m,n,D)$ is the round complexity of subgraph S-T vertex connectivity; furthermore, the algorithm only communicates via the set of edges that is incident to $C$.
    \end{enumerate}
\end{lemma}

To solve \isocut{} and \mnnccn{} problems, in addition to sparsification techniques from \cite{LiNPSY21}, we introduce novel vertex sparsification lemmas that leverages the structure of the common-neighborhood property. We discuss them in details in \Cref{sec:cluster sparsification}.

With these two lemmas, we are ready to prove our main theorem about solving \SSUDVC{}.

\begin{lemma}\label{lem:SSVCpram}
    If, in \pram{} model, the s-t vertex connectivity problem can be solved in $W(m,n)$ work and $D(m,n)$ depth where $W(m,n)$ is superadditive on $m$, then \SSUDVC{} can be solved in $W(m,n)\cdot n^{o(1)}$ work and $D(m,n)\cdot n^{o(1)}$ depth.
\end{lemma}
\begin{proof}
    The correctness of \cref{alg:framework} is proved by \cref{lem:framework pram congest}. Notice that to boost the correct probability to w.h.p., we repeat the algorithm $\hO{1}$ times and choose the smallest vertex cut as the output, which only increases the work and depth by at most a $\hO{1}$ factor. Now we calculate the total work and depth for one run of \cref{alg:framework}.

    There are $O(\log n)$ iterations of the outer loop of $i$. Common-neighborhood clustering uses $\hO{m}$ work and $\hO{1}$ depth according to \cref{lem:sncrunningtime}. For every $C\in\cC$, we solve the problems \isocut{} and \mnnccn{} simultaneously. Write $E_C$ as the edge set of edges in $G$ adjacent to $C$. Each one of them on $C$ cost work $\tO{W(|E_C|,n)}$ and depth $\tO{D(|E_C|,n)}$ according to \cref{lem:refined isolating cut lemma} and \cref{lem:mnnccn}. Notice that each edge can be contained in at most $\tO{1}$ different $C$ according to \cref{lem:snc}. Thus, the total work is $\sum_{C\in\cC}\tO{W(|E_C|),n)}\le \tO{W(\sum_{C\in\cC}|E_C|,n)}=\tO{W(m,n)}$ where the inequality is due to $W$ is superadditive on $m$, and the depth is $\tO{D(m,n)}$ since they are run simultaneously.
\end{proof}

\begin{lemma}\label{lem:SSVCcongest}
    If, in \congest{} model, the subgraph S-T vertex connectivity problem can be solved in $R(m,n,D)$ rounds, then \SSUDVC{} can be solved in $R(m,n,D)\cdot n^{o(1)}$ depth.
\end{lemma}

\begin{proof}
    The correctness of \cref{alg:framework} is proved by \cref{lem:framework pram congest}. Notice that to boost the correct probability to w.h.p., we repeat the algorithm $\hO{1}$ times and choose the smallest vertex cut as the output (the size of each vertex cut is also output according to \cref{def:isolatingcutsproblem,def:minneibornearclique}), which only increases the work and depth by at most a $\hO{1}$ factor. Now we calculate the total rounds for one run of \cref{alg:framework}.

    There are $O(\log n)$ iterations of the outer loop of $i$. For each iteration, common-neighborhood clustering uses $\hO{\sqrt{n}+D}$ rounds according to \cref{lem:sncrunningtime}. For every $C\in\cC$, we solve the problems \mnnccn{} simultaneously. Notice that each edge can be adjacent to at most $\tO{1}$ different $C$ according to \cref{lem:snc}. Thus, each edge is involved in at most $\tO{1}$ different algorithms of \isocut{} and \mnnccn{}, which results in total congestion of $\tO{R(m,n,O(1))}$ according to \cref{lem:refined isolating cut lemma} and \cref{lem:mnnccn}. So the total number of rounds caused by \mnnccn{} is at most $\tO{R(m,n,O(1))}$.

    Now we explain how to implement \isocut{} for every cluster $C\in\cC$ by one call to \cref{lem:refined isolating cut lemma}. We first explain why we can not do similar things to \mnnccn{}: this is because the round complexity stated in \cref{lem:refined isolating cut lemma} depends on the diameter of the subgraph $G'$ which only includes edges incident to $C$, and the diameter could be much larger than the diameter of $G$. We solve this issue in the following way. For each $C\in\cC$, we first apply \cref{lem:vertexsparsification} with $G,C,T$ to get $C'\subseteq N_G(C)$ and $K$ so that it suffices to run \isocut{} on $G'[C\cup C']$ as they preserve all isolating cuts up to a fixed number $K$. We also have $|C'|=\hO{|C|}$ since $|T|=\hO{|C|/\ell}$ w.h.p. Then we construct a virtual graph with $\hO{n}$ vertices in the following way: the virtual graph contains $|\cC|$ many disconnected connected components, where each component corresponds to a cluster $C\in\cC$, which we call $G_C$, where $G_C$ contains the vertex set $C\cup C'$ and all edges adjacent to $C$ in $G$. Now we call \isocut{} on this virtual graph with the input vertex $C$ defined in \cref{def:isolatingcutsproblem} set to be $\cup_{C\in\cC}C$ and the input vertex set $T$ set to be the union of sample $T$ in the algorithm. To summarize, we call $\isocut{}(G_{vir},\cup_{C\in \cC}C,\cup_{C\in\cC}T_C)$ where $T_C$ is the sampled set for $C$.

    We first show how to use the original communication network to simulate the virtual graph: for each vertex $u\in V$, $u$ will simulate all the nodes in the virtual graph which is a duplication of $u$. Consider an edge adjacent to $u$ which is $(u,v)$, we should prove that $(u,v)$ is not duplicated too many times into the virtual graph so that the congestion of simulating the virtual graph is bounded. To see this, notice that an edge can be included in $G_C$ for a cluster $C\in\cC$ only if $(u,v)$ is adjacent to $C$. This can happened at most $\tO{1}$ times according to \cref{lem:snc}. Thus, we can simulate \isocut{} algorithm on the virtual graph using the original graph with a congestion increased by a factor of $\tO{1}$. The total round complexity is $R(\tO{m},\tO{n},\tO{D}=\tO{R(m,n,D)}$.

    Then, we show that $\isocut{}(G_{vir},\cup_{C\in \cC}C,\cup_{C\in\cC}T_C)$ on the virtual graph gives the answer to $\isocut{}(G',C,T_C)$ for all independent instances on $C$. Firstly, every minimum isolating cut $N_{G_{vir}}(L)$ with $|L\cap \cup_{C\in\cC}T_C|=1$ cannot cross different connected components in $G_{vir}$ because otherwise we can restrict $L$ into a single connected component, which cannot increase the value of $N_{G_{vir}}(L)$. Thus, the value returned by $\isocut{}(G_{vir},\cup_{C\in \cC}C,\cup_{C\in\cC}T_C)$ is at least the size of $\isocut{}(G',C,T_C)$ for some $C\in\cC$. Moreover, every isolating cut $N_{G'}(L)$ with $L\subseteq C$ and $|L\cap T_C|=1$ is definitely a valid isolating cut for $G_{vir}$. Thus, the value returned by $\isocut{}(G_{vir},\cup_{C\in \cC}C,\cup_{C\in\cC}T_C)$ is exactly what we want.
\end{proof}

\subsection{Cluster Sparsification} \label{sec:cluster sparsification}
A \emph{bipartite matching} of a bipartite graph $(A,B,E)$ (where $A, B$ forms a partition of the vertex set and $E\subseteq A\times B$) is an edge set $M$ such that any two edges do not share an endpoint. For convenience, we define $M_A=\{u\mid u\in A,\exists v\in B,(u,v)\in M\}$, and $M_B$ similarly. By the definition, $|M_A|=|M_B|=|M|$. Define $M(u)=v$ where $u\in M_A,v\in B,(u,v)\in M$, and $M(U)=\{M(u)\mid u\in U\}$ where $U\subseteq M_A$. Similarly, we can define $M^{-1}(v)=u$ and $M^{-1}(U)=\{M^{-1}(v)\mid v\in U\}$ for $v\in M_B$ and $U\subseteq M_B$. A \emph{vertex cover} of this bipartite graph is a set of vertices $C\subseteq A\cup B$ such that every edge in $E$ is adjacent to at least one vertex in $C$.

\begin{lemma} \label{lem:remove z}
    Let $G = (V,E)$ be an undirected graph and $s, t \in V$. Let $Z = N_G(s) \cap N_G(t)$, and $H := G - Z$. Then, $\kappa_G(s,t) = \kappa_H(s,t) + |Z|$. Furthermore, If $S$ is an $(s,t)$-min-separator in $H$, then $Z \cup S$ is an $(s,t)$-min-separator in $G$.
\end{lemma}
\begin{proof}
    This follows because every vertex $v \in N_G(s) \cap N_G(t)$  belongs to every $(s,t)$-separator in $G$.
\end{proof}

\begin{lemma} \label{lem:repeat mm}
  Given a bipartite graph $G = (A,B,E)$ where $|A| \leq |B|$, there is an $\tilde O(m)$-work $O(\log n)$-depth algorithm that outputs a vertex set  $D \subseteq B$ of size $O(|A|\log n)$ such that there is a maximum matching $M$ where $M_B \subseteq D$.
\end{lemma}
\begin{proof}
This set can be computed by $O(\log n)$ computations of maximal bipartite matching, see Appendix B of~\cite{assadi2022semi}.
\end{proof}

\begin{lemma}   [Cluster Boundary Sparsification Lemma (Singleton Version)] \label{lem:sparsify by mm generic}Let $G=(V,E)$ be an undirected graph and $s,t\in V$ satisfying $N_G[s]\cap N_G[t]=\emptyset$. Define $B=N_G(t)$ and $A=V-N_G[t]$ and let $M$ be a maximum matching of $(A,B,E_G(A,B))$. Define $G'=G[V-(B-M_B)]$ as the graph $G$ after removing neighbors of $t$ that are unmatched by the maximum matching $M$. Then, we have $\kappa_G(s,t)=\kappa_{G'}(s,t)$. In particular, every minimum $(s,t)$-separator in $G'$ is a minimum $(s,t)$-separator in $G$.
\end{lemma}
\begin{proof}
    We first prove that $\kappa_G(s,t)\ge \kappa_{G'}(s,t)$. Let $S$ be a minimum $(s,t)$-separator in $G$. So, every $(s,t)$-path in $G$ passes a vertex in $S$. Since $G'$ is a subgraph of $G$, every $(s,t)$-path in $G'$ is an $(s,t)$-path in $G$, and thus it passes a vertex in $S$. Thus, $\kappa_G(s,t) = |S| \ge \kappa_{G'}(s,t)$.

    It remains to prove $\kappa_G(s,t)\le \kappa_{G'}(s,t)$.  We argue that there is a maximum $(s,t)$ vertex-disjoint paths of size $\kappa_G(s,t)$ in $G'$.
    There exists $\kappa=\kappa_G(s,t)$ internal vertex disjoint paths $p_1,...,p_{\kappa}$ in $G$. Let $b_i$ be the vertex that $p_i$ intersect $B$ the first time, and let $B'=\{b_i\mid i\in[\kappa]\}$. There exists a matching $M'$ of size $\kappa$ such that $M'_B=B'$ (by matching $b_i$ to the previous node of $b_i$ on path $p_i$). This means the set  $A' := M'_A$ contains all the previous nodes of $b_i$ on $p_i$.  To prove $\kappa_G(s,t)\le \kappa_{G'}(s,t)$, it suffices to prove that there exists a matching $M''$ of size $\kappa$ such that $M''_A=A'$ and $M''_B\subseteq M_B$ (in which case there exists $\kappa$ internal vertex disjoint path connecting $s,t$ in $G'$).

     Now we show the existence of such a matching $M''$. Let $C$ be a minimum vertex cover of $(A,B,E_G(A,B))$. By Kőnig's theorem, $|C| = |M|$. Let $C_A=C\cap A$ and $C_B=C\cap B$. By definitions,  $C_A\subseteq M_A,C_B\subseteq M_B$. Each edge in $M$ has exactly one endpoint in $C_A$, or one endpoint in $C_B$ (because every edge in $M$ must have one endpoint in $C$ while $|C|=|M|$) . Let $C'_A=C_A\cap A'$ and $C'_B=C_B\cap B'$. We define $M''$ in the following way: for every $u\in C'_A$, define $M''(u)=M(u)$; for every $u\in A'-C'_A$, define $M''(u)=M'(u)$. We first prove that $M''_B\subseteq M_B$. For any $u\in C'_A$, trivially we have $M''(u)=M(u)\in M_B$. For any $u\in A'-C'_A$, since $C$ is a vertex cover and $(u,M'(u))$ is an edge where $u\not\in C_A$, we have $M'(u)\in C_B\subseteq M_B$. Thus, $M''_B\subseteq M_B$. Now we prove that $M''$ is a valid matching, i.e., for any $u\in C'_A,v\in A'-C'_A$, we need to prove $M(u)\not=M'(v)$. Since $u\in C'_A$, we have $M(u)\not\in C_B$. Moreover, we have $M'(v)\in C_B$, thus, $M(u)\not=M'(v)$.
\end{proof}

\begin{lemma} [Cluster Boundary Sparsification Lemma (Batched Version)] \label{lem:vertexsparsification}
    Given $G=(V,E)$, $C\subseteq V$ and $\ell$ satisfying the common-neighborhood property $|N_G(u)\triangle N_G(v)|=\hO{\ell}$ for every $u,v\in C$, and in addition given $X\subseteq C$ with $|X|=\hO{|C|/\ell}$.
    There exists \pram{} and \congest{} algorithm that outputs $C'\subseteq N_G(C)$ and an integer $K$ such that
    \begin{enumerate}
        \item $|C'|=\tO{|C|}$,
        \item for every $x\in X$, define $G'=G[C\cup C']$, we have $\min_{L' \subseteq C : x\in L'} |N_G(L')|=\min_{L' \subseteq C : x\in L'} |N_{G'}(L')|+K$ for a fixed number $K$.
        \item for every $x\in X$, define $G'=G[C\cup C']$, we have $\min_{L' \subseteq C : x\in L',|L'\cap X|=1} |N_G(L')|=\min_{L' \subseteq C : x\in L',|L'\cap X|=1} |N_{G'}(L')|+K$ for a fixed number $K$.
    \end{enumerate}
    The algorithm runs in $\tO{m}$ work and $\tO{1}$ depth in \pram{} where $m$ is the number of edges in $G$, and $\tO{1}$ rounds in \congest{}.
\end{lemma}
\begin{proof}

    If $X$ is empty, return $C' = \emptyset$, and the lemma is vacuously true; so let us assume $X$ is non-empty.

    Notice that if the maximum degree of nodes in $C$ becomes much larger than $\hO{|C|}$, then we can delete all vertices in $(\cap_{x\in X}N_G(x))\cap N_G(C)$, which decrease the value of $|N_G(L')|$ by $K=|(\cap_{x\in X}N_G(x))\cap N_G(C)|$ for every $L'$ satisfying $x\in L'$ for some $x\in X$. This set $(\cap_{x\in X}N_G(x))\cap N_G(C)$ can be found in $O(m)$ work and $O(1)$ depth in \pram{}, and in $1$ round in \congest{}, after deleting which can reduce the maximum degree of nodes in $C$ to $\hO{|C|}$ because of the common-neighborhood property: in order for a vertex $v\in N_G(x)$ to not be deleted, either $v\in C$, or there exists $x'\in X$ such that $v\in N_G(x')\triangle N_G(x)$, which can happen at most $\hO{\ell}$ times for each $x'$. Thus, the degree is at most $\hO{|C|+|X|\cdot \ell}=\hO{|C|}$. In the following proof we will assume the maximum degree of nodes in $C$ is at most $\hO{|C|}$, so the total number of edges adjacent to $C$ is $\hO{|C|^2}$.

    We first show that for a specific $x$, in order for $\min_{L' \subseteq C : x\in L'} |N_G(L')|=\min_{L' \subseteq C : x\in L'} |N_{G'}(L')|$ to be true, it suffices to let $C'$ contain (i) $N_G(x)\cap N_G(C)$, (ii) the (right) endpoints of the edges in an arbitrary maximum bipartite matching between $C$ and $N_G(C)-N_G(x)$. To see this, we need \cref{lem:sparsify by mm generic}, in which we will set $s=x$, $t$ be a new super node connecting to every vertex in $N_G(C)$ and $G$ be the graph after adding the super node $t$ and deleting all common neighbors of $s$ and $t$. It is easy to notice that by \cref{lem:sparsify by mm generic}, preserving the vertices in an arbitrary maximum bipartite matching between $C$ and $N_G(C)-N_G(x)$, along with $N_G(x)\cap N_G(C)$ suffices to preserve the value of $\min_{L' \subseteq C : x\in L'} |N(L')|$.

    Thus, the problem becomes finding $C'$ such that for every $x\in X$, it preserves at least one maximum bipartite matching between $C$ and $N_G(C)-N_G(x)$, and $C'\supseteq N_G(x)\cap N_G(C)$. Let $H$ be the bipartite subgraph of $G$ only containing edges between $C$ and $N_G(C)$. We will use the folklore vertex size reduction for BMM, which can be found in Appendix B of \cite{assadi2022semi}. A critical lemma can be summarized as follows.

    \begin{lemma}[Lemma 14 of \cite{assadi2022semi}]\label{lem:reducingvertexBMM}
        For a bipartite graph $G=(V,E)$, and a vertex cover $\tilde{V}\subseteq V$ (i.e., for every $(u,v)\in E$, either $u\in \tilde{V}$ or $v\in\tilde{V}$), let $M$ be an arbitrary maximal bipartite matching between $\tilde{V}$ and $V-\tilde{V}$ and $V(M)$ be the vertex set of its endpoints. If the maximum bipartite matching of $G[\tilde{V}]$ has size $F$ and the maximum bipartite matching of $G$ has size $F^*$, then the maximum bipartite matching of $G[\tilde{V}\cup V(M)]$ has size at least $\frac{1}{3}(F^*-F)+F$.
    \end{lemma}

    Thus, for a specific $x\in X$, it suffices to set $C'=\emptyset$ initially, then repeatedly finding an arbitrary maximal bipartite matching between $C\cup C'$ and $N_G(C)-N_G(x)-C'$, after which add the endpoints of the maximal matching to $C'$. Notice that $C$ is a vertex cover for $H$. According to \cref{lem:reducingvertexBMM}, after $O(\log n)$ iterations, we are guaranteed that $H[C\cup C']$ contains the same size of maximum bipartite matching as $H[C\cup (N_G(C)-N_G(x))]$. In the end we add $N_G(x)\cap N_G(C)$ to $C'$.

    However, we need to construct $C'$ for all $x\in X$, which will increase the size of $C'$ to $|X|\cdot |C|$, far from our goal of $\tO{|C|}$. To solve this, we let $\tilde{C}=\cup_{x\in X}(N_G(x)\cap N_G(C))$. Notice that $|\tilde{C}|=\hO{|X|\cdot \ell+|C|}=\hO{|C|}$ according to the common neighborhood property and the maximum degree bound. Then, we repeatedly find an arbitrary maximal bipartite matching between $C$ and $N_G(C)-\tilde{C}$, add the endpoints to $\tilde{C}$ and repeat for $O(\log n)$ times. A critical observation is that, in each iteration, the edges between $C$ and $N_G(C)-\tilde{C}$ are the same as edges between $C\cup\tilde{C}$ and $(N_G(C)-N_G(x))-(\tilde{C}-N_G(x))$ for every $x$ because $N_G(x)\cap N_G(C)\subseteq \tilde{C}$ and $N_G(C)$ is an independent set in $H$. Thus, in the end $H[C\cup (\tilde{C}-N_G(x))]$ contains the same size of maximum bipartite matching as $H[C\cup (N_G(C)-N_G(x))]$, according to \cref{lem:reducingvertexBMM}. In the end we set $C'=\tilde{C}$ which certainly contains $N_G(x)\cap N_G(C)$ for every $x\in X$, so $C'$ preserves $\min_{L' \subseteq C : x\in L'} |N_{G'}(L')|$ for every $x\in X$. Since there are $O(\log n)$ iterations where each maximal bipartite matching has size at most $|C|$, we have $|C'|=\hO{|C|}$. The algorithm can be implemented in $\tO{m}$ work and $\tO{1}$ depth because maximal bipartite matching can be solved in $\tO{m}$ work and $\tO{1}$ depth. Moreover, in \congest{}, maximal bipartite matching can be solved in $\tO{1}$ rounds.

    In order to also make sure $\min_{L' \subseteq C : x\in L',|L'\cap X|=1} |N_G(L')|=\min_{L' \subseteq C : x\in L',|L'\cap X|=1} |N_{G'}(L')|+K$, we the same algorithm while setting $C\leftarrow C-X$. In this way, one maximum matching between $C-X$ and $N_G(C)-N_G(x)$ is preserved for every $x\in X$. According to \cref{lem:sparsify by mm generic}, the value of
    $\min_{L' \subseteq C : x\in L',|L'\cap X|=1} |N_G(L')|$ is preserved for every $x\in X$.

\end{proof}

\section{Isolating Cuts Lemma for Parallel and Distributed Algorithms (Proof of \cref{lem:refined isolating cut lemma})}   \label{sec:isolating cut lemma parallel dist}

In this section we will focus on proving the following lemma.

\begin{lemma} [Parallel and Distributed Isolating Cuts Lemma] \label{lem:refined isolating cut lemma 2}
Given a graph $G = (V,E)$ and an independent set $T \subseteq V$ of size at least $2$, there is an algorithm that outputs for each $v \in T$ a $(v, T -\{v\})$-min-separator $C_v$. The (sequential) algorithm makes calls to s-t vertex mincut on graphs with $\tilde O(n)$ total number of vertices and $O(m)$ total number of edges and takes $\tilde O(m)$ additional time.
\begin{itemize}
    \item In the PRAM model, the algorithm can be implemented to run in work $\tO{W(m,n)}$ and depth $\tO{D(m,n)}$ where $W(m,n)$ and $D(m,n)$ are the work and depth of s-t vertex connectivity.
    \item In the distributed \congest model, the algorithm can be implemented to run in $\tilde O(R(m,n,D))$ rounds where $R(m,n,D)$ is the round complexity of subgraph S-T vertex connectivity and $D$ is the diameter of $G$.
\end{itemize}
\end{lemma}

We first show that \cref{lem:refined isolating cut lemma} directly follows from \cref{lem:refined isolating cut lemma 2}.
\begin{proof}[Proof of \cref{lem:refined isolating cut lemma 2}]
    Given a graph $G$, a vertex set $C$ and a vertex set $T$, we first use \cref{lem:vertexsparsification} to construct $C'\subseteq N_G[C]$ such that $G[C\cup C']$ preserves the minimum isolating cut. Write $G'=G[C\cup C']$. Then for each $v\in N_{G'}(C)$, we add a new node $v'$ and connect $v'$ to $v$ by an edge, add $v'$ to $T$, let the resulting terminal set be $T'$. In the distributed network, $v'$ is simulated by $v$. Now we run the algorithm described in \cref{lem:refined isolating cut lemma 2} on $G'$ with $T'$. It is easy to see that this solves \cref{lem:refined isolating cut lemma}. Moreover, $G'$ only has $\hO{C}$ many vertices according to \cref{lem:vertexsparsification}. Thus, the work and depth for \pram{} are $\tilde{W(m,|C|)}$ and $\tilde{D(m,|C|)}$.
\end{proof}

We prove the parallel version of the isolating cut lemma~\cite{LiNPSY21}. Previously, the isolating cut lemma~\cite{LiNPSY21} can only guarantee $O(m)$ total number of vertices which is suboptimal for parallel implementation.
Before we prove \Cref{lem:refined isolating cut lemma 2}, we first review an algorithm for the isolating cut lemma~\cite{LiNPSY21} and then we state the refinement.

\paragraph{Algorithm~\cite{LiNPSY21}} The inputs consist of the input graph $G = (V,E)$ and an independent set $T \subseteq V$ of size at least $2$.
\begin{enumerate}[noitemsep]
    \item Encode each $t \in T$ as a binary string of length $\ceil{\log_2|T|}$. For each $i \leq  \ceil{\log_2|T|}$, define       $A_i :=  \{ t \in T \colon i^{\text{th}} \text{bit of }  $t $ = 0 \}$ and   $B_i := T - A_i$
    Compute $F_i := (A_i,B_i)$-vertex mincut.
    \item For each $s \in T$, compute a connected component containing $s$, denoted by $U_s$, in $G - \bigcup_iF_i$.
    \item For each $s \in T$, define $G'_s$ as $G[N_G[U_s]]$ followed by (1) adding an additional vertex $t$ and all edges from $t$ to $N_G(U_s)$, and (2) remove all edges inside $N_G(U_s)$. Then, compute a minimum $(s,t)$-separator denoted by $C_s$ in $G'_s$.
\end{enumerate}

We summarize the correctness into the following lemma.
\begin{lemma} [\cite{LiNPSY21}]
For each  $s \in T$, $C_s$ is an $(s, T - \{s\})$-min-separator in $G$.
\end{lemma}

This construction yields $O(m)$ total number of vertices because the boundary of $U_s$ for all $s \in T$ can be large.

\paragraph{Refinement.}  Our refinement is at step 3: before computing the minimum $(s,t)$-separator in $G'_s$, we further sparsify $G'_s$, while preserving the minimum $(s,t)$-separator, so that the total number of vertices is $\tilde O(|U_s|)$.

\begin{lemma} [Cluster Sparsification] \label{lem:sparsify by MM}
    Let $G'_s$ be given as an input, and denote $Z_s := N_{G'_s}(s) \cap N_{G'_s}(t)$ as the common neighbors between $s$ and $t$ in $G'_s$. There is an $\tilde O(vol_{G}(U_s))$-time algorithm that returns a vertex set $D \subseteq N_G(U_s) - Z_s$ satisfying the following property:
    Let $G''_s$ be the graph $G'_s$ after removing $Z_s \cup (N_G(U_s) - D)$.
    \begin{itemize} [noitemsep]
        \item $\kappa_{G'_s}(s,t) = \kappa_{G''_s}(s,t) + |Z_s|$,
        \item If $W$ is an $(s,t)$-min-separator in $G''_s$,  then $W \cup Z_s$ is an $(s,t)$-min-separator in $G'_s$, and
        \item  $|V(G''_s)| = O(|U_s|\log n)$.
    \end{itemize}

    In the  PRAM model, $G''_s$ can be constructed in $\tilde O(vol_{G}(U_s))$-work and $\tilde O(1)$-depth.

\end{lemma}
\begin{proof}
    We are given the graph $G'_s$, and the goal is to compute a graph $G''_s$ with the properties described in the statement.  We describe the construction of $G''_s$ in two steps.
\begin{enumerate} [noitemsep]
    \item [(R1)] Let $Z := N_{G'_s}(s) \cap N_{G'_s}(t)$ from $G'_s$.
    \item [(R2)] Remove all vertices in $N_G(U_s) - (Z \cup D)$  where $D$ is computed from the following.
     \begin{enumerate}
         \item
         Define the bipartite graph $(A,B,E')$ where $A := U_s, B := N_G(U_s) - Z$ and $E' = E_G(A,B)$.
    \item   If $|A| \leq |B|$, apply \Cref{lem:repeat mm} on the bipartite graph $(A,B,E')$ to obtain the vertex set $D \subseteq B$.  Otherwise, $D = N_G(U_s) - Z$.
     \end{enumerate}
\end{enumerate}

\paragraph{Analysis}
In terms of work, it takes $O(m)$ work to construct the bipartite graph instance $(A,B,E')$. In step $(R2)$, we compute the set $D$ in $\tilde O(m)$ work and $O(\log n)$ depth~\Cref{lem:sparsify by mm generic}. All these steps can be implemented in $\tilde O(1)$ depth.
We now argue the size. In step (R2), if $|A| > |B|$, then we are done. Otherwise, the set $D$ has size $O(|A|\log n)$ by \Cref{lem:repeat mm}. Therefore, $|V(G''_s)| \leq |U_s| + 2|D| = O(|U_s|\log n)$.  For correctness, the fact that $\kappa_{G'_s}(s,t) = \kappa_{G''_s}(s,t) + |Z|$ follows from \Cref{lem:remove z,lem:sparsify by mm generic,lem:repeat mm}.
\end{proof}
Therefore, \Cref{lem:refined isolating cut lemma 2} can be shown as follows. We apply the same algorithm in the first two steps, but we slightly modify Step 3 as follows. In Step 3, for each $s\in T$,  we replace the graph $G'_s$ with $G''_s$ produced by \Cref{lem:refined isolating cut lemma 2}, and return $Z_s \cup C'_s$ where $Z_s$ is the common neighbors between $s$ and $t$ as stated in \Cref{lem:sparsify by MM}, and $C'_s$ is a minimum $(s,t)$-separator $C'_s$ on $G''_s$. The total number of edges over all $s\in T$ is $O(m)$ as before because we only remove edges from $G'_s$. Now, the total number of vertices is \[\sum_{s \in T}|V(G''_s)| \leq O(\sum_{s\in T} |U_s|\log n) \leq \tilde O(n).\]
It is easy to see that we can implement in the PRAM model in total $\tilde O(m)$ work and $\tilde O(1)$ depth outside of the calls to s-t vertex mincut. Note that $(A_i,B_i)$-vertex mincut can be reduced to s-t vertex mincut by contracting $A_i$ into $s$ and $B_i$ into $t$.

\paragraph{Implementation in distributed \congest setting.} The algorithm is similar to PRAM except that contracting a vertex set in Step 3 is a global operation. To deal with this issue,
we modify $G'_s$ into a \textit{cut-equivalent} graph by adding new terminal sets along $N_G(U_s)$. For any vertex sets $A,B \subseteq V$, $\kappa_G(A,B)$ is the size of an $(A,B)$-vertex mincut.
\begin{lemma} \label{lem:cut-equivalent terminal graph}
Given $G'_s$, define $H_s$ as follows. First, remove $t$ from $G'_s$. Then, for each vertex $u \in N(U_s)$, add a new vertex $y$ along with an edge $(u,y)$. Denote $Y$ as the set of new vertices. Then, every $(s,t)$-cut in $G'_s$ is an $(s,Y)$-cut in $H_s$ and vice ver sa. In particular, $\kappa_{G'_s}(s,t) = \kappa_{H_s}(s,Y)$.
\end{lemma}
\begin{proof}
If $S$ is an $(s,t)$-separator in $G'_s$, then denote a vertex cut $(L,S,R)$ in $G'_s$ where $s \in L, t \in R$. Since $N_{G'_s}(t) \subseteq S \cup R$,   every vertex in $Y$ in $H_s$ does not have an edge to $L$. Therefore, $(L, S, (R \cup Y) - \{t\})$ is a vertex cut in $H_s$. If $S$ is an $(s,Y)$-separator in $H_s$, then the same separator must be an $(s,t)$-separator in $G'_s$ by contracting all $Y$ into $t$.
\end{proof}

Note that the set $Y$ can be simulated inside the incident vertices. Now, we can state the cluster sparsification lemma for the distributed version as follows. The proof is similar to \Cref{lem:sparsify by MM}.

\begin{lemma} \label{lem:sparsify by MM congest}
    Given $H_s$,  denote $Z_s := N_{H_s}(s) \cap N_{H_s}(Y)$ as the common neighbors between $s$ and $Y$. There is an $\tilde O(1)$-round algorithm that returns a vertex set $D \subseteq N_G(U_s)-Z_s$ satisfying the following property. Define $H'_s$ as the graph $H_s$ after removing $Z_s \cup (N_{G}(U_s) - D)$ and isolated vertices in $Y$, and denote $Y'$ to be the remaining nodes in $Y$. Then, $\kappa_{H_s}(s,Y) = \kappa_{H'_s}(s,Y') + |Z|$. If $W$ is an $(s,Y')$-mincut in $H'_s$,  then $W \cup Z$ is an $(s,Y)$-mincut in $H_s$.  Furthermore, $|V(H'_s)| = O(|U_s|)$.

\end{lemma}

Finally, observe that each graph $H'_s$ alone can have a larger diameter than the original graph. To handle this case, instead of running each graph separately, we will combine all $H'_s$ instances and compute $(s,Y)$-mincut over all $s \in T$ and over all the union of $H'_s$. Running on all instances the combined graph will preserve the diameter. We are ready to state the distributed algorithm in \congest model.

\paragraph{Algorithm (\congest).} The inputs consist of graph $G = (V,E)$, an independent set $T \subseteq V$ of size at least $2$.
\begin{enumerate} [noitemsep]
    \item Encode each $t \in T$ as a binary string of length $\ceil{\log_2|T|}$. For each $i \leq  \ceil{\log_2|T|}$, define       $A_i :=  \{ t \in T \colon i^{\text{th}} \text{bit of }  $t $ = 0 \}$ and   $B_i := T - A_i$.
    Compute $F_i := (A_i,B_i)$-vertex mincut.
    \item For each $s \in T$, compute a connected component containing $s$, denoted by $U_s$, in $G - \bigcup_iF_i$.
    \item For each $s \in T$, construct $H_s$ and then compute $H'_s$ using \Cref{lem:sparsify by MM congest}. Let $Z_s$ be the common neighbors in the lemma.
    \item Compute an $(s,Y')$-min-separator $W$ in $H'_s$ for all $s \in T$ on the union of $H'_s$ overall $s \in T$.
    \item  For each $s \in T$, by \Cref{lem:sparsify by MM congest}, $Z_s \cup W_s$ is an $(s,Y)$-mincut in $H_s$.
\end{enumerate}

\paragraph{Running Time.}  The first two steps can be implemented in $O( R(m,n,D) \cdot \log |T|)$ rounds (\Cref{lem:ST vc problem in matching}). Step 3 can be done in $\tilde O(1)$ rounds by \Cref{lem:sparsify by MM congest}. For each $s \in T$, let $m_s$ and $n_s$ be the number of edges and number of vertices of $H'_s$. We can run Step 4 all in parallel on the union of $H'_s$ as a single graph, which can be done in $O(R(\sum_{s \in T} m_s, \sum_{s \in T} n_s, D)) = \tilde O(R(m,n,D))$ rounds.

\paragraph{Correctness.} Steps 1 and 2 are identical to the sequential algorithm. We now focus on Step 3. Fix a vertex $s \in T$. We prove that the algorithm returns an $(s,t)$-min-separator in $G'_s$. By \Cref{lem:cut-equivalent terminal graph}, an $(s,Y)$-min-separator in $H_s$ is also a $(s,t)$-min-separator in $G'_s$.  By \Cref{lem:sparsify by MM congest}, an $(s,Y')$-min-separator in $H'_s$ is also an $(s,Y)$-min-separator in $H_s$, and we are done.

\section{Parallel and Distributed Algorithms for \mnnccn{} (Proof of \cref{lem:mnnccn})}\label{sec:mnnccn}

\subsection{\pram{} Algorithm}
This section proves part 1 of \cref{lem:mnnccn}.   We are given the \mnnccn{} instance $(G,C,\ell)$ (\Cref{def:minneibornearclique}).    The PRAM algorithm goes as follows. We first sample $\hO{|C|/\ell}$ random vertices in $C$, denoted by $X$. Given $X, C$ and $\ell$, we apply the batched version of the cluster boundary sparsification lemma (\Cref{lem:vertexsparsification}) to replace the boundary $N_G(C)$ to $C'$ where $|C'| = \tO{|C'|}$ without changing the minimizer. For each $x \in X $, we compute $(x,t)$-min separator in the cluster after sparsification using similar reduction rules in~\cite{LiNPSY21}.

\paragraph{Algorithm.} The inputs consist of $(G,C,\ell).$
\begin{enumerate} [noitemsep]
    \item As preprocessing, we first apply \cref{thm:sketching} to compute $sk_{z}(N(u))$ for every $u\in C$ where $z=2^{\log^{0.8}n}\ell = \hO{\ell}$.
    \item Let $X$ be a set of  $\hO{|C|/\ell}$ random vertices sampled in $C$.
    \item Given $G,C,\ell, X$, apply the batched version of the cluster boundary sparsification lemma (\Cref{lem:vertexsparsification}) to obtain $C' \subseteq N_G(C)$ and an integer $K$.
    \item For each $x \in X$, construct the following graph $H_x = (V_x,E_x)$ where \begin{align*}
    V_x &:= C \sqcup C'_x \sqcup \{t\}, \\
    E_x &:= \bigcup_{v \in C}E_{G}(v, V_x - (N_G(x) \cap C)) \cup \{(v,t) \colon v \in C'_x\} \mbox{, and }\\
    C'_x &:= C' - Z_x \text{ where } Z_x := N_G(x) \cap C'.
\end{align*}
\item Let $x^* := \arg\min_{x \in X}\bigg \{ \kappa_{H_x}(x,t) + |Z_x|+K\bigg\}$
\item Let $(L,S,R)$ be a vertex mincut in $H_x$ such that $x^* \in L$ and $t \in R$.
\item {\bf Return} $L$ and $|S|$.
\end{enumerate}

Before we present the analysis, we establish the following sparsification lemma at Step 4.

\begin{lemma}\label{lem:sparsify G in pram}
Let $C' \subseteq N_G(C)$ be the sparsified boundary of $C$ defined in Step 3. For each $x \in X$ where $X$ is defined in Step 2, the graph $H_x$ defined in Step 4 can be constructed in $\hO{|C|\cdot \ell}$ work and polylog depth and $H_x$ satisfies the following properties:
\begin{enumerate} [nosep]
    \item For every vertex $v \in C, \textdeg_{H_x}(v,V_x - (N_G(x) \cap C)) \leq O(2^{\log^{0.7}n}\cdot \ell)$
    \item $\kappa_{H_x}(x,t) +|Z_x| \geq \min_{L' \subseteq C: L' \neq \emptyset}|N_G(L')|$.
    \item If $x \in \LL$, then $\kappa_{H_x}(x,t) + |Z_x| = \min_{L' \subseteq C: L' \neq \emptyset}|N_G(L')|$.
\end{enumerate}
\end{lemma}
\begin{proof}
 We focus on the fast construction of $H_x$. By the preprocessing step, we have an access to the oracle that can list all the elements in the symmetric difference  $N(s)\triangle N(s')$ for any pair $(s,s')$ up to $\hO{\ell}$ elements in $\hO{\ell}$ time. By the property of $C'$, $|N_{H_x}(C)| \leq \tilde O(|C|)$. Hence, the set $Z_x$ can be computed quickly. For each vertex $v \in C - \{x\}$, we list all the neighbors of $v$ that are not the neighbors of $x$ in $\hO{\ell}$ work using the oracle, and we can do it all in parallel.  Given the batch version of the cluster boundary sparsification lemma, the correctness proof is similar to \cite{LiNPSY21}.
\end{proof}

We are ready to analyze the algorithm.
\paragraph{Analysis.} Since we sample $\hO{|C|/\ell}$ many vertices $X$, there is $x \in X$ such that $x \in \LL$ with high probability. Therefore, the correctness follows   from \Cref{lem:sparsify G in pram}.  For the running time, the algorithm is easily parallelizable. We bound the total work needed to complete this operation. The sketching can be done in almost linear work~\Cref{thm:sketching}. By \Cref{lem:sparsify G in pram}, the total work is
\[\sum_{x \in X} \hO {W(|C| \ell, |C|)} \leq \hO{W(|C|^2,|C|)} \leq \hO{W(m,n)}  .\]

\subsection{\congest{} Algorithm}
In this section, we prove part 2 of \cref{lem:mnnccn}.  We are given the \mnnccn{} instance $(G,C,\ell)$ (\Cref{def:minneibornearclique}).

\paragraph{High-level algorithm.}   At high level, we sample $\tilde \Theta (|C|/ \ell)$ many vertices in $C$. With high probability, there is a node $x \in L \subseteq C$. Fix $x$. We construct a graph $H_x$ such that the minimizer of $\min_{L' \subseteq C: L' \neq \emptyset}|N_G(L')|$ can computed by calling maximum bipartite matching in $H_x$.
However, this may cause too much congestion when we run all instances of $H_x$ for all $x \in X$. To reduce congestion, we \emph{map} its communication of $\mathcal{A}$ on $H_x$ to a random graph $H'_x$ so that the congestion due to running $\mathcal{A}$ is \textit{load-balanced} in that with high probability, the maximum overlap (in terms of the number of instances on a pair $(u,v) \in C \times C$) of $H'_x$ instances overall $x \in X$ is at most $\widehat O(1)$.

\paragraph{Step 1. Graph Sparsification.} Given $x \in C$, we construct a graph $H_x$ as in the following lemma.
\begin{lemma} \label{lem:sparsify G in congest}
Given $x \in C$, we construct graph $H_x = (V_x, E_x)$ in $O(1)$ rounds where
\begin{align*}
    V_x &:= C \sqcup C'_x \sqcup Y_x \mbox{, and }\\
    E_x &:= \bigcup_{v \in C}E_{G}(v, V_x - (N_G(x) \cap C)) \sqcup M_x,\\
    M_x &:=\text{the perfect matching between $C'_x$ and $Y_x$} \\
    C'_x &:= N_G(C) - Z_x \text{ where } Z_x := N_G(x) \cap N_G(C).
\end{align*}
The graph $H_x$ satisfies the following properties:
\begin{enumerate} [nosep]
    \item For every vertex $v \in C, \textdeg_{H_x}(v,V_x - (N_G(x) \cap C)) \leq O(2^{\log^{0.7}n}\cdot \ell)$
    \item $\kappa_{H_x}(x,Y_x) +|Z_x| \geq \min_{L' \subseteq C: L' \neq \emptyset}|N_G(L')|$.
    \item If $x \in \LL$, then $\kappa_{H_x}(x,Y_x) + |Z_x| = \min_{L' \subseteq C: L' \neq \emptyset}|N_G(L')|$.
\end{enumerate}
\end{lemma}

The proof is based on the reduction rules in \cite{LiNPSY21}. For completeness, we prove \Cref{lem:sparsify G in congest} in \Cref{sec:sparsfy in congest}.

\paragraph{Step 2. Mapping the communication.} We would like to run $\mathcal{A}$ on $H_x$. There are two issues. First, the graph $H_x$ is a subgraph of $G$ so the diameter can be large. Second,  the same edge in $H_x$ can be used for multiple instances of $x \in X$. To deal with the first issue, we reduce the diameter of $H_x$ by adding a star to $H_x$.
\begin{definition} [$H^*_x$]
We define $H^*_x = (V(H_x) \cup \{s^*\}, E(H_x) \cup E^*)$ where we add $(s^*,v)$ edge to $E^*$ for all $v \in V(H_x)$.
\end{definition}
Observe that the diameter of $H^*_x$ is $O(1)$, and adding a star to $H_x$ does not change the minimizer of $\min_{L' \subseteq C: L' \neq \emptyset}|N_G(L')|$. We assume $s^*$ exists for now;  we describe how to simulate the star node $s^*$ later, which can be done efficiently since $C$ is nearly a clique.

To deal with the second issue, we \textit{simulate an algorithm} $\mathcal{A}$ on $H^*_x$ by \textit{mapping the communication} of $\mathcal{A}$ to a random graph.

\begin{definition} [Mapping the communication of $\mathcal{A}$ on $H^*$] \label{def:mapping comm}
 Let $f_x: C \cup C'_x \cup \{s^*\} \rightarrow C$ be a function. Given $H^*_x = (C \cup C'_x \cup \{s^*\}, E_x \cup E^*)$, define the graph $H'_x := (C,E'_x)$ where $(u,v) \in E_x \cup E^*$  if and only if $(f_x(u),f_x(v))_{(u,v)} \in E'_x$ (with possibly parallel edges, and self-loop).  We say that the communication of $\mathcal{A}$ on $H^*_x$ \emph{is mapped} to $H'_x$ \emph{via} $f_x$  if every node $v \in C$ knows the neighbor sets of its preimage, i.e., $v$ knows $N_{H_x}(w)$ for all $w \in f^{-1}(v)$.
\end{definition}

If we can set up such a mapping of the communication, then the algorithm $\mathcal{A}$ on $H^*_x$ can be \textit{simulated} entirely on $H'_x$.

\begin{lemma} [Simulation] \label{lem:simulation alg}If the communication of $\mathcal{A}$ on $H^*_x$ is mapped to $H'_x$ via $f_x$ then there is an algorithm $\mathcal{A'}$ running on $H'_x$ such that every round the message exchange $(u,v)$ in $H^*_x$ by $\mathcal{A}$ is equivalent to exchanging the same message between $(f_x(u),f_x(v))_{(u,v)}$ in $H'_x$ by $\mathcal{A'}$.
\end{lemma}

Furthermore, if $f_x$ is a random function, then $H'_x$ is a random graph, and thus running multiple instances of $\mathcal{A}$ on $H^*_x$ (using many $x$'s) can be better load-balanced by simulating on multiple instances of $H'_x$.

\begin{lemma} \label{lem:map}
Given $x \in C$ and the description of algorithm $\mathcal{A}$ running on $H^*_x$, we can map the communication of $\mathcal{A}$ to the graph $H'_x$ (\Cref{def:mapping comm}) via a random function $f_x$ where, in $H'_x$, every vertex $v \in C$ has $\widehat O(\ell)$ random neighbors. The mapping can be constructed using $O(1)$ rounds of communication in $H'_x$.
\end{lemma}

We prove \Cref{lem:simulation alg,lem:map} in \Cref{sec:mapping comm congest}.

\paragraph{Step 3. Virtual Clique in $C$.} Notice that the algorithm $\mathcal{A}'$ assumes the communication on every edge inside $H'_x$, which may not exist in the original network $G_x$. To run $\mathcal{A}'$ on $H'_x$, we can simulate clique communication inside $C$.

\begin{definition}
    An $t$-virtual clique network in $C \subseteq V$ of a graph $G = (V,E)$ is a protocol that simulates the clique communication in $C$ in that any message exchange between $u,v \in C \cup \{s^*\}$ can be done by at most $t$ rounds in $G$.
\end{definition}

\begin{lemma} \label{lem:virtual clique}
Given a cluster $C \subseteq V$ in a graph $G = (V,E)$, we can construct a $t$-virtual clique network $\mathcal{V}$ in $\tilde O(1)$ rounds. With high probability, $t \leq O(\log n)$.
\end{lemma}

We prove \Cref{lem:virtual clique} in \Cref{sec:virtual clique}.

We are now ready to state the main \congest{} algorithm.

\paragraph{Algorithm (\congest{}).} The input is a graph $G = (V,E)$ along with a cluster $C \subseteq V$ and an estimate $\ell$.

\begin{enumerate} [noitemsep]
    \item Let $\mathcal{V}$ be the virtual clique network on $C$ constructed from \Cref{lem:virtual clique}.
    \item Sample $100\cdot \frac{|C|}{\ell}  \cdot \log n$ random vertices in $C$ into $X$.
    \item For each $x \in X$,
    \begin{enumerate} [noitemsep]
        \item Construct a graph $H_x$ using \Cref{lem:sparsify G in congest}. Let $\mathcal{A}$ be the algorithm that computes minimum $(x,Y_x)$-separator in $H_x$.
        \item Apply \Cref{lem:map} to map the communication of $\mathcal{A}$ to $H'_x$ via $f_x$. Let $\mathcal{A}'$ be the algorithm running on $H'_x$ obtained from \Cref{lem:simulation alg}.
        \item Run $\mathcal{A'}$ on $H'_x$ using the virtual clique network $\mathcal{V}$
        \item Obtain the output of $\mathcal{A}$ on $H^*_x$ from the output of $\mathcal{A}'$ in $H'_x$ using $f_x$.
    \end{enumerate}
    \item Return the smallest cut and its size found so far.
\end{enumerate}

\paragraph{Analysis.} In step 2, there exists $x \in X$ such that $x \in L$ with high probability since $\ell = \Theta (|L|)$. Therefore,  \Cref{lem:sparsify G in congest} implies that a minimum $(x,Y_x)$-separator of $H_x$ plus $Z_x$, and thus we obtain the minimizer of $\min_{L'\subseteq C: L' \neq \emptyset}|N_G(L')|$ at the end of step 4. In terms of round complexity, observe that the bottleneck is to run $\mathcal{A'}$ running on a $H'_x$ as described in \Cref{lem:map}. Fix a vertex $v \in C$, we show in expectation there are at most $O(\log n)$ instances of $\mathcal{A}'$ on every edge from $v$ to $C$. This follows because there are $\tilde O(\frac{|C|}{\ell})$ independent instances of $\mathcal{A}'$ running on $v$ where each instance corresponds to random $\widehat O(\ell)$ edges inside $C$. Therefore, in expectation, each edge has $\widehat O(1)$ congestion due to running $\mathcal{A}'$.

\subsubsection{Mapping the Communication} \label{sec:mapping comm congest}

In this section, we prove \Cref{lem:simulation alg,lem:map}. We start with the proof of \Cref{lem:map}. The goal is to map the communication of $\mathcal{A}$ on $H^*_x$ to $H'_x$ using a function $f_x$ defined as follows.

\begin{definition}
   Given $x \in C$, we define $f_x: C \cup C'_x \cup \{s^*\} \rightarrow C$ where
\begin{align*}
f_x(v) = \begin{cases}
  v & \text{if } v \in N_G[x] \cap C, \\
   \text{ a random node in } C& \text{otherwise}.
\end{cases}
\end{align*}
\end{definition}

Note that $f_x$ is a public random function. So every node can access $f_x$.   To establish a mapping (\Cref{def:mapping comm}), we describe an algorithm to broadcast neighborhoods.

\paragraph{Algorithm.} For all $u \in C \cup C'_x - (N_G[x] \cap C)$, every node $v \in N_{H^*_x}(u)$ sends its id and $u$ to $f(u)$ so that $f(u)$ receives $(v,u)$ to learn that $v$ is a neighbor of $u$.

\paragraph{Complexity.} It is convenient to fix a node $v$ and describe the set of vertices to which $v$ sends its id. By design, every $v$ sends its id to the image $f(w)$ of every node $w \in N(v) - (N_G[x] \cap C)$ where $f(w)$ is a random node in $C$. The distribution can be described as a random graph where every node $u \in C -\{x\}$ has $|N(v)- N_G[x]| \leq O(\ell \log n)$ random neighbors in $C$. The inequality follows from \Cref{lem:sparsify G in congest}.

Therefore, \Cref{lem:map} follows from $f_x$ and this algorithm.  We next prove \Cref{lem:simulation alg}. After the mapping above, we have established the following property.

\begin{quote}
  \textbf{Property 1}: For all $u \in C$, and for all $w \in f^{-1}_x(u)$, $u$ learned the set of $N_{H_x}(w)$ of $w$.
\end{quote}

Property 1 means that for every $u \in C \cup C'_x \cup \{s^*\}$ in $H^*_x$, $f_x(u)$ can act as a surrogate of $u$ since $f_x(u)$ knows all the neighbors of $u$.  That is, we can simulate any algorithm running on $H^*_x = (C \cup C'_x \cup \{s^*\}, E_x \cup E^*)$ by running on another network $H'_x := (C, E'_x)$ where an edge $(u,v) \in E_x$ if and only if $(f_x(u),f_x(v))_{(u,v)} \in E'_x$ to $E'_x$ where we allow parallel edges and self-loops indexed by the pair $(u,v)$ in $H'_x$.

We are ready to prove \Cref{lem:simulation alg}.
\begin{proof} [Proof of \Cref{lem:simulation alg}]
Given $\mathcal{A}$, we define another algorithm $\mathcal{A'}$ that simulates the execution of $\mathcal{A}$ as follows:
\begin{itemize}
    \item For each round in $\mathcal{A}$,
    \begin{itemize}
        \item if $u$ sends/receives a message $m$ to $v$ in $\mathcal{A}$ then $f_x(u)$ sends/receives the message $m$ to $f_x(v)$ in $\mathcal{A'}$ using the edge $(f_x(u),f_x(v))_{(u,v)}$.
    \end{itemize}
     Note that every node $v$ knows its preimage $f_{x}^{-1}(v)$ since $f_x$ is a public (random) function.
\end{itemize} \qedhere
\end{proof}

\subsubsection{Virtual Clique Network} \label{sec:virtual clique}
We prove \Cref{lem:virtual clique} in this section.
We start with basic facts about $C$ being a near-clique.  For all $u \in C$, let $NN_C(u) := C - N_G[u]$ be the non-neighbors of $u$ in $C$.

We assume that
\begin{align} \label{eq: C much larger than L}
|C| \geq 100\cdot 2^{\log^{0.8}n}\cdot \ell.
\end{align}
Otherwise, we sample $n^{o(1)}$ nodes inside $C$ and run bipartite matching trivially on each of them and we are done.

\begin{claim} \label{claim:near clique prop}
 For all $u \in C$,
    \begin{align} \label{eq:Nu_cap_C_omega_NN_u prop}
        |N_G(u) \cap C| \geq \Omega(|NN_C(u)|).
    \end{align}
    For all $u,v \in C$,
    \begin{align}
    |N_G(u) \cap N_G(v) \cap C| &\geq \Omega (|C|). \label{eq:Nu_Nv_C_eq_omega_C prop}
    \end{align}

\end{claim}
\begin{proof}

We prove \Cref{eq:Nu_cap_C_omega_NN_u prop}.
\begin{align*}
    |N_G(u) \cap C| = |C| - |NN_C(u)| - 1  \overset{(\ref{eq:near clique prop})}{\geq} |C| - \ell \cdot 2^{\log ^{0.8}n} -1 \overset{(\ref{eq: C much larger than L})}{\geq} 99 \cdot 2^{\log^{0.8}n} \cdot \ell -1 \overset{(\ref{eq:near clique prop})}{\geq} \Omega (|NN_C(u)|).
\end{align*}
We prove \Cref{eq:Nu_Nv_C_eq_omega_C prop}.
\begin{align*}
|N_G(u) \cap N_G(v) \cap C|
 \geq |C| - |NN_C(u)| - |NN_C(v)| -1
 \overset{(\ref{eq:near clique prop})}{\geq}
|C| - 2\cdot 2^{\log^{0.8}n}\cdot \ell  -1
 \overset{(\ref{eq: C much larger than L})}{\geq} \Omega(|C|).
\end{align*}
\end{proof}

\paragraph{Preprocessing.} For each $u,v \in C$ such that $u$ and $v$ are non-neighbors in $G$,  $u$ and $v$ agree on a random node $r_{u,v}$ from the common neighbors in $C$, i.e., $(N_G(u) \cap N_G(v)) \cap C$.   This step can be implemented as follows.
\begin{itemize}
    \item   For all $u \in C$, and for all $v \in C - N_G[u]$ where  $u < v$ (the id of $u$ is less than the id of $v$),
    \begin{enumerate}
        \item $u$ sends a message requesting to connect to $v$ to $\Theta (\log n)$ random neighbors $N_G(u) \cap C$.
        \item An intermediate node $w \in N_G(u)$, upon receiving the request from $u$, sends to $v$ if $v \in N_G(w) \cap C$.
        \item If node $v$ receives the request of $u$ from the intermediate nodes, then  $v$ selects one node $w$ arbitrarily and replies to $w$ so $w$ sends the confirmation back to $u$
    \end{enumerate}
\end{itemize}

\paragraph{Correctness.} It is enough to prove that the preprocessing is \textit{successful}; that is for every pair of non-neighbors $u,v$ in $C$, $u$ and $v$ agree on the node $r_{u,v} \in N_G(u) \cap N_G(v) \cap C$ as an intermediate node. By design, this happens if and only if, during the preprocessing, one of $u$'s messages requesting to $v$ was sent to $N_G(u) \cap N_G(v) \cap C$. Fix a node $u \in C$, the probability that a random message its request is sent to $N_G(u) \cap N_G(v) \cap C$ is $\frac{|N_G(u) \cap N_G(v) \cap C|}{|C|} \overset{(\ref{eq:Nu_Nv_C_eq_omega_C prop})}{\geq} \Omega(1)$.
Therefore, after $\Theta(\log n)$ trials and by the union bounds, the setup is successful with high probability.

\paragraph{Clique Simulation.} We simulate a message from $u$ to $v$ in clique for all $u,v \in C$ as follows. If there is an edge $(u,v)$, then we send the same message directly. Otherwise, $u$ sends the message to $r_{u,v}$ for which $r_{u,v}$ forwards to $v$.  This step can be implemented in \congest{} easily.

\paragraph{Complexity of the Simulation.} We prove that each clique round can be implemented in $O(\log n)$ rounds of clique simulation.  For each $u,v$ where $u < v$ and $v \not \in N_G(u)\cap C$,  we define a  2-path (a path of length 2), denoted by $P_{u,v} := (u,r_{u,v},v)$.   Let $\mathcal{P} := \{P_{u,v} \colon u < v, v \not \in N_G(u)\cap C\}$. For any edge $e \in E_G(C,C)$, the \textit{congestion}, denoted by $c(e)$, is the number of paths $P_{u,v}$ that contains $e \in E_G(C,C)$. We show that the maximum congestion $\max_ec(e)$ over all edges is $O(\log n)$ with high probability.

Fix an edge $e = (x,y) \in E_G(C,C)$. If a 2-path $P$ contains $e$, then $e$ is either the first or the second edge of $P$. Let $\Gamma_1(x,y) := \{w \in C - N_G[x] \colon \exists P_{x,w} = (x,r_{x,w},w) \in \mathcal{P},  r_{x,w} = y\}$ be the set of all $2$-paths in $\mathcal{P}$ that contains $(x,y)$ as the first edge. Similarly let $\Gamma_2(x,y)$ be the set of all $2$-paths in $\mathcal{P}$ that contains $(x,y)$ as the second edge. By definitions, $c(e) = |\Gamma_1(x,y)| + |\Gamma_2(x,y)|$.

Next, we show that  $|\Gamma_1(x,y)| \leq O(\log n)$ and  $|\Gamma_2(x,y)| \leq O(\log n)$ with high probability. We focus on bounding the size of  $\Gamma_1(x,y)$. By definition of $\Gamma_1(x,y)$, $(x,y)$ belongs to the first edge of a path $P_{x,w} = (x,r_{x,w},w) \in \mathcal{P}$  if $w$ is a non-neighbor of $x$ where $x$ and $w$ agree on $r_{x,w} = y$ as an intermediate node. Thus, $|\Gamma_1(x,y)|$ is the number of non-neighbors $w$ of $x$ whose intermediate node is selected to be $y$. Let $I[r_{x,w} =y]$ be an indicator random variable. We have
\[ |\Gamma_1(x,y)| = \sum_{w \in NN_C(x)} I[r_{x,w} = y] \]

For all non-neighbors $w \in NN(x)$,  since $r_{x,w}$ is a random node in $N_G(x) \cap N_G(w) \cap C$, the probability that $r_{x,w} = y$ is
\begin{align*}
\mathbb{P} (r_{x,w} = y) &=  \begin{cases}
  0 & \text{if } y \not \in N_G(w) \cap C, \\
  \frac{1}{|N(x)\cap N(w) \cap C|} & \text{otherwise}
\end{cases}\\
 &\leq  \frac{1}{|N(x)\cap N(w) \cap C|}  \\
 &\leq  O(\frac{1}{|N(x) \cap C|}).
\end{align*}
The last inequality follows since
$$  |N(x) \cap N(w) \cap C| \overset{(\ref{eq:Nu_Nv_C_eq_omega_C prop})}{\geq} \Omega (|C|) \geq \Omega (|N(x) \cap C|) . $$

Therefore, in expectation,
\begin{align*}
    \mathbb{E}[|\Gamma_1(x,y)|] &= \sum_{w \in NN_C(x)} \mathbb{P}[r_{x,w} = y]\\
    &\leq |NN_C(x)| \cdot  O(\frac{1}{|N(x) \cap C|}) \\
    &\overset{(\ref{eq:Nu_cap_C_omega_NN_u prop})}{\leq}  O(1).
\end{align*}
Similarly, we can also prove that $\mathbb{E}[|\Gamma_2(x,y)|] \leq O(1).$ By Chernoff bound, we conclude that  $c(x,y) \leq O(\log n)$ with high probability.

 \paragraph{Round Complexity.} The round complexity for the preprocesing is $O(\log n)$ factor larger than one iteration of clique simulation.

\section{Parallel Implementation in Matrix Multiplication Work}\label{sec:parallelimplementation mm work}

In this section, we prove the main lemma for solving \SSUDVC{} in \pram{} model in matrix multiplication work and subpolynomial depth.

\begin{lemma}\label{lem:SSVCprammatrix}
    There is a randomized \pram{} algorithm that solves \SSUDVC{} in $n^{\omega+o(1)}$ work and $n^{o(1)}$ depth.
\end{lemma}

\cref{lem:SSVCprammatrix} can be proved by the following two implementations of \isocut{} and \mnnccn{}.
\begin{lemma} \label{lem:near-clique in MM time}
There is a randomized \pram{} algorithm that solves an \mnnccn{} instance $(G,C,\ell)$ (\Cref{def:minneibornearclique}) in $\tilde O(|C|^\omega)$ work and $\tilde O(1)$ depth.
\end{lemma}

\begin{lemma} \label{lem:isolating cut in MM time}
There is a randomized \pram{} algorithm that solves an \isocut{} instance $(G,C,T)$ (\Cref{def:isolatingcutsproblem}) in $\tilde O(|C|^\omega)$ work and $\tilde O(1)$ depth.
\end{lemma}
\begin{proof}
    We use the algorithm \cite{Lovasz79} that solves maximum bipartite matching in matrix multiplication time and work. We can turn that algorithm into an s-t vertex connectivity algorithm by using \cref{lem:ST vc problem in matching} (it reduces s-t vertex connectivity to minimum vertex cover, which can be found by first finding a maximum bipartite matching, then run reachability). Now s-t vertex connectivity can be solved in $\tO{n^{\omega}}$ work and $\tO{1}$ depth; we can plug it in \cref{lem:refined isolating cut lemma} to finish the proof.
\end{proof}

We prove \Cref{lem:near-clique in MM time} in \Cref{sec:proof ncmmtime}. We are now ready to prove \Cref{lem:SSVCprammatrix}.
\begin{proof}[Proof of \cref{lem:SSVCprammatrix}]
    We use \cref{alg:framework}. The correctness is implied by \cref{lem:framework pram congest}. So, we only need to calculate the complexity. We run the algorithm $\hO{1}$ times to boost the probability of succeeding and take the minimum cut as output.

    There are $\log n$ outer loops. Inside each loop, for each $C$, the work for running \isocut{} and \mnnccn{} is $|C|^{\omega}$ and the depth is $\tO{1}$. We run the algorithm for all $C\in\cC$ simultaneously so the depth does not change. The total work is $\tO{\sum_{C\in\cC}|C|^\omega}=\tO{n^{\omega}}$ because $\sum_{C\in\cC}|C|=\tO{n}$ according to \cref{lem:snc}.
\end{proof}

\subsection{Proof of \Cref{lem:near-clique in MM time}} \label{sec:proof ncmmtime}
This section is devoted to proving \Cref{lem:near-clique in MM time}.
 Recall the notations introduced in \Cref{sec:framework}. We are given the \mnnccn{} instance $(G,C,\ell)$ (\Cref{def:minneibornearclique}).
 
In this section, we describe a sequential algorithm for simplicity, but these algorithms can be easily parallelized.  We denote $T_{\MM}(n,k,r)$ as the number of field operations needed to multiply an $n \times k$ matrix with a $k \times r$ matrix.
\begin{lemma} [\cite{CGLZ20}]
\[ T_{\MM}(n,k,r) = O(T_{\MM}(k,n,r)) = O(T_{\MM}(n,r,k)).\]
\end{lemma}

\paragraph{Matrix Data Structures.} We will use the matrix data structure such that the  $\rank(M)$ can be maintained quickly under low-rank updates.

\begin{definition}
Let $n,a$ be natural numbers where $a < n $.   A matrix $V \in \mathbb{R}^{n \times a}$ is a \emph{selector matrix} if $V$ is a binary matrix, every column has exactly one 1's and every row has at most one 1's.
\end{definition}

Selector matrices can be used to construct a larger matrix. To illustrate, let
   \[
   \mathbf{u}_j = \begin{bmatrix} u_{j1} \\ u_{j2} \\ \vdots \\ u_{jn} \end{bmatrix} \quad \text{and} \quad \mathbf{u}_k = \begin{bmatrix} u_{k1} \\ u_{k2} \\ \vdots \\ u_{kn} \end{bmatrix}
   \]
   where \( \mathbf{u}_j \) will define the values in the \( j \)-th column and \( \mathbf{u}_k \) will define the values in the \( k \)-th column.
and
 \[
   \mathbf{v}_j = \begin{bmatrix} 0 \\ \vdots \\ 1 \\ \vdots \\ 0 \end{bmatrix} \quad \text{and} \quad \mathbf{v}_k = \begin{bmatrix} 0 \\ \vdots \\ 1 \\ \vdots \\ 0 \end{bmatrix}
   \]
where the \( 1 \) in \( \mathbf{v}_j \) is in the \( j \)-th position and the \( 1 \) in \( \mathbf{v}_k \) is in the \( k \)-th position.

Define a selector matrix $V = \begin{bmatrix} \mathbf{v}_j &  \mathbf{v}_k \end{bmatrix}.$ Also, define $U = \begin{bmatrix} \mathbf{u}_j & \mathbf{u}_k \end{bmatrix}.$ We obtain a larger matrix by

\[
   UV^T = \begin{bmatrix} \mathbf{u}_j & \mathbf{u}_k \end{bmatrix} \begin{bmatrix} \mathbf{v}_j^T \\ \mathbf{v}_k^T \end{bmatrix}  = \begin{bmatrix} 0 & \cdots & u_{j1} & \cdots & u_{k1} & \cdots & 0 \\ 0 & \cdots & u_{j2} & \cdots & u_{k2} & \cdots & 0 \\ \vdots & & \vdots & & \vdots & & \vdots \\ 0 & \cdots & u_{jn} & \cdots & u_{kn} & \cdots & 0 \end{bmatrix}.
 \]

We will represent column updates of $M$ of the form $M + UV^T$ where $V$ is a selector matrix.

\begin{lemma} \label{lem:matrix ds} There is a data structure $\mathcal{M}$ that supports the following operations.
\begin{itemize}
    \item $\mathcal{M}.\textsc{Construct}(M,k)$ where $M \in \mathbb{R}^{n \times n}$: Preprocess and construct the data structure using $M$ and integer $k \leq n$ in $O(n^\omega)$ time. With high probability, $\mathcal{M}$ correctly operates $\textsc{IsRankGtK}$ defined next.
    \item $\mathcal{M}.\textsc{IsRankGtK}(U,V)$ where $U,V \in \mathbb{R}^{n \times a}, a \leq n $,  and $V$ is a selector matrix: decide if $\rank(M+ U V^\top) \geq k$ in $O(T_{\MM}(n,a,a))$ time.
\end{itemize}
\end{lemma}
We will use low-rank update techniques as in \cite{BrandNS19, Sankowski07}, but their proofs are described using sequential updates. For our purpose, we prove the batch version and perform the updates in one shot.
\begin{proof} We explain the construction first.

\paragraph{$\mathcal{M}.\textsc{Construct}(M,k)$.} We are given $M \in \mathbb{R}^{n \times n}$ and an integer $k \leq n$ for preprocessing.
The data structure constructs the following matrix:
\begin{align}
f(M) =
\begin{bmatrix}
M & X & 0\\
Y& 0 &  I \\
0 & I & I^{(q)} \\
\end{bmatrix}
\end{align}
where $I^{(q)}$ is an $n$-by-$n$ binary matrix whose first $q$ diagonal entries are 1, $X$ and $Y$ are independent and uniformly random matrices where each entry is randomly chosen from some field extension of size $\Omega(n^2)$. Sankowski (\cite{Sankowski07}, Lemma 4.1, Theorem 4.1) showed that the matrix $f(M)$ is full rank w.h.p. if and only if $\rank(M) \geq n - q$.

By setting $q = n-k$, we can decide if $f(M)$ has rank $\geq k$ by checking if $f(M)$ is full rank, which can be checked by computing its determinant.  We fix $q = n - k$ from now on.

In preparation for the query operation where we will be given two matrices $U,V \in \mathbb{R}^{n \times a} $. We can decide if $\rank (f(M + U V^\top)) \geq k$ by computing the determinant of $f(M+ U V^\top)$. That is,
\begin{lemma} \label{lem:test rank via full rank}
Given  $U,V \in \mathbb{R}^{n \times a}$ where $a \leq n$,
$\rank (f(M + U V^\top)) \geq k$ if and only if $f(M + U V^{\top})$ is full rank with high probability.
\end{lemma}

We also precompute $\textdet(f(M))$ and $f(M)^{-1}$ in $O(n^{\omega})$ time. For simplicity, we assume that $f(M)$ is invertible. This assumption can be removed; see Theorem C.10 of \cite{Sankowski07}.

\paragraph{$\mathcal{M}_k.\textsc{IsRankGtK}(U,V)$.} We are given where $U,V \in \mathbb{R}^{n \times a}, k,a < n $,  and $V$ is a selector matrix. By \Cref{lem:test rank via full rank}, it is enough to compute the determinant $\textdet(f(M+ UV^\top))$. To do so, we  use the matrix determinant lemma.
\begin{lemma} [Matrix Determinant Lemma]
Let $A \in \mathbb{R}^{n \times n}$ be an invertible matrix and $U,V \in \mathbb{R}^{n\times a}$.
\begin{align} \label{eq:matrix det lemma}
    \textdet(A + U V^\top) = \textdet(A) \textdet(I + V^\top A^{-1} U).
\end{align}
\end{lemma}

Let $U_0$ be a matrix and  $V^{\top}_0$ be a selector matrix such that
\begin{align}  \label{eq:U0V0}
f(M) + U_0 V^{\top}_0 = f(M+UV^\top).
\end{align}

Therefore, we can compute $\textdet(f(M+ UV^\top))$ as follows.
\begin{align*}
  \textdet(f(M+ UV^\top)) &\overset{(\ref{eq:U0V0})}{=} \textdeg(f(M) + U_0 V^{\top}_0) \\
  &\overset{(\ref{eq:matrix det lemma})}{=}  \textdet(f(M)) \textdet(I_{a} + V_0^\top (f(M))^{-1} U_0).
\end{align*}

Since we precomputed $\textdet(f(M))$ during the preprocessing, it remains to compute \[\textdet(I + V_0^\top (f(M))^{-1} U_0),\] which can be computed in time \[O(n\cdot a +  T_{\MM}(a,n,a) + a^\omega) = O(T_{\MM}(a,n,a)).\]
Indeed, we break down the computation as follows.
 \begin{enumerate} [noitemsep]
     \item Compute $Q_0 := V^\top_0 f(M)^{-1}$ in $O(n\cdot a)$ time.  Since $V_0$ is a selector matrix, the matrix $V^\top_0 f(M)^{-1}$ can be obtained by reading corresponding rows $f(M)^{-1}$ selected by the 1's entries in $V^\top_0$ in $O(n\cdot a)$ time. The matrix $(f(M))^{-1}$ was stored at the preprocessing step.
     \item Compute $Q_1 := Q_0 U_0$ in $O(T_{\MM}(a,n,a))$ time.
     \item Compute $Q_2 := I_a + Q_1$ in $O(a^2)$ time.
     \item Compute $\textdet(Q_2) = \textdet(I_{a} + V_0^\top (f(M))^{-1} U_0)$ in $O(a^\omega)$ time.
 \end{enumerate}

This concludes the query operation in $O(T_{\MM}(a,n,a))$ time.
\end{proof}

\paragraph{Convex Embedding.}

Let \(\mathbb{F}\) be a finite field. For \(k \geq 0\), the space \(\mathbb{F}^k\) is a \(k\)-dimensional linear space over \(\mathbb{F}\). Let \(X = \{ x_1, \ldots, x_n \}\) be a finite set of points within \(\mathbb{F}^k\). The \emph{affine hull} of \(X\) is defined as
\[
\text{aff}(X) = \left\{ \sum_{i=1}^k c_i x_i : x_i \in X \text{ and } \sum_{i=1}^k c_i = 1 \right\}.
\]
The rank of \(X\), represented as \(\text{rank}(X)\), is one plus the dimension of \(\text{aff}(X)\). Specifically, if \(\mathbb{F} = \mathbb{R}\), we refer to the \emph{convex hull} of \(X\), denoted \(\text{conv}(X)\).
For any sets $V, W $, any function $f : V \rightarrow W$, and any subset $U \subseteq V$, we denote $f(U) := \{ f(u) \colon u \in U \}$.

\begin{definition}[Convex $X$-embedding~\cite{LinialLW88}]
For any $X \subset V$, a convex  $X$-embedding of
a graph $G = (V,E) $ is a function $f : V  \rightarrow \mathbb{R}^{|X|-1}$
such that for each $v \in V \setminus X$, $f(v) \in \text{conv}(f(N_G(v)))$.
\end{definition}

For efficient construction, \cite{LinialLW88,CheriyanR94} define convex-embedding over finite field $\mathbb{F}$. In
particular, they construct the convex $X$-embedding over the field of integers modulo a prime $p$, $\mathbb{Z}_p$ by fixing a random prime number $p \in [n^5, n^6]$, and choosing a random nonzero coefficient function $c: E \rightarrow (\mathbb{Z}_p \setminus \{ 0 \})$ on edges. This construction yields a function $f : V \rightarrow (\mathbb{Z}_p^{|X|-1})$ called \textit{random modular $X$-embedding}.

\begin{lemma} [Theorem 4.13 of \cite{LinialLW88}] \label{lem:random construction}
Let $t$ be a fixed node in a graph $G = (V,E)$.  A random modular $N(t)$-embedding $f$ can be constructed in $\tilde O(n^{\omega})$ time. With high probability, for every node $s \in V - N[r]$, we have $\kappa_G(s,t) = \rank(f(N(s)))$.
\end{lemma}

Using matrix data structure (\Cref{lem:matrix ds}), we can construct a fast connectivity data structure as follows.

\begin{lemma} \label{lem:connectivity oracle}
Let $t$ be a fixed node in a graph $G = (V,E)$.  Suppose we have an oracle that can list all the elements in the symmetric difference  $N(s)\triangle N(s')$ for any pair $(s,s')$ up to $a$ elements in $\tilde O(a)$ time.
\begin{itemize}
    \item A random modular $N(t)$-embedding $f$ can be constructed in $\tilde O(n^{\omega})$ time.
    \item With additional $\tilde O(n^{\omega})$ preprocessing time on another fixed node $s \in V - N[t]$ and parameter $k$, we have the following with high probability: for all $s' \in V - N[t]$, we can decide if $\kappa_G(s',t) \geq k$ in $O(T_{\MM}(n,a,a))$ time where $a = |N(s) \triangle N(s')|$.
\end{itemize}
\end{lemma}
\begin{proof}
Let $f$ be a random modular $N_{G}(t)$-embedding $f$, which can be constructed in $\tilde O(n^{\omega})$ using \Cref{lem:random construction}. We can represent $f(V)$ as a matrix $M_V := [\ldots f(v) \ldots]$ where each column corresponds to a vector $f(v)$. For all $U \subseteq V$, we denote $M_U$ as a submatrix of $M_V$ when we restrict the columns of $M_V$ to those in $f(U)$.

We will use the matrix data structure $\mathcal{M}$ in \Cref{lem:matrix ds}. We first use $M := M_{N(x)}$ where $x$ is a fixed node and $k$ is a parameter for the construction of $\mathcal{M}$ in $O(n^{\omega})$ time.  Next, we are given query $s' \in V - N[t]$, where $|N(s) \triangle N(s')| = a' \leq a$. Let $W = N(s')$, and so we can write $M_W := M + UV^{T}$ where $U \in \mathbb{R}^{k \times a'}$ corresponds to the column differences in the embedding, and $V \in \mathbb{R}^{k \times a'}$ is a selector matrix.
Therefore, we can check if $\rank(M_W) \geq k$ by calling $\mathcal{M}.\textsc{IsRankGtK(U,V)}$, which can be done in $O(T_{\MM}(k,a,a))$ time by  \Cref{lem:matrix ds}.
\end{proof}

We are ready to present the PRAM algorithm and prove \Cref{lem:near-clique in MM time}.

\paragraph{\pram{} Algorithm.} We are given  $G=(V,E)$, $C\subseteq V$ and $\ell$ as the inputs to \mnnccn{} satisfying the guarantees described in \cref{def:minneibornearclique}.
\begin{enumerate} [noitemsep]
    \item   Use \cref{thm:sketching} to compute $sk_{z}(N(u))$ for every $u\in C$ where $z=2^{\log^{0.8}n}\ell$
    \item Let $X$ be a set of  $\hO{|C|/\ell}$ random vertices sampled in $C$.
    \item Apply \Cref{lem:vertexsparsification} using $G,C,\ell$ and $X$ as parameters to obtain $C' \subseteq N_G(C)$. Let $K$ be the fixed number in the lemma statement.
    \item Let $G' = G[C \cup C']$.
    \item Let $G''$ be $G'$ after adding a sink $t$ and the set of edges from $t$ to every node in $C'$.
    \item Binary search on $k$ until we obtain $k  = \min_{x \in X}\kappa_{G''}(x,t)$ on the following problem: decide if $\min_{x \in X}\kappa_{G''}(x,t) \geq k$.
    \begin{enumerate}  \item Given the guess value $k$, apply \Cref{lem:connectivity oracle} using $t$ as the fixed node, and the sketching $sk_z$ as an oracle that can list $N_G(s) \triangle N_G(s')$ for any pair $(s,s')$ up to $\hO{\ell}$ elements in $\hO{\ell}$ work.
        \begin{enumerate} [noitemsep]
            \item Construct a random modular $N_{G''}(t)$-embedding $f$.  Let $s \in C$ be the other fixed node and use parameter $k$ for the preprocessing.
            \item For all $x \in X$, decide if $\kappa_{G''}(x,t) \geq k$ using \Cref{lem:connectivity oracle}.
        \end{enumerate}
    \end{enumerate}
    \item returns the corresponding minimum cut.
\end{enumerate}

\paragraph{Correctness.} Let $\LL$ be the minimizer of $\min_{L' \subseteq C : L' \neq \emptyset} |N_G(L')|$. By \Cref{lem:vertexsparsification}, we know that $G''$ have the same minimizer as in $G$ for all $x \in X$. That is, for all $x \in X$, $\min_{L' \subseteq C : x\in L'} |N_G(L')|=\min_{L' \subseteq C : x\in L'} |N_{G'}(L')|+K$. Since we sample $\hO{|C|/\ell}$ random vertices into $X$, there is $x \in \LL$ with high probability. Therefore, by the correctness of \Cref{lem:connectivity oracle}, we also obtain the correct minimizer with high probability.

\paragraph{Complexity.}  The preprocessing steps and sparsification can be done in $\hO{m}$ work. The main bottleneck is to binary search on $k$ in step 6. In this step, the running time follows from \Cref{lem:connectivity oracle}. That is,
\[ \tilde O(|C|^{\omega})  +  \hO{\frac{|C|}{\ell}} \cdot \tilde O(T_{\MM}(n,\ell,\ell)) = \hO{|C|^{\omega}}. \]
The first term on LHS follows from the preprocessing time since $|V(G'')| = \tilde O(|C|)$ by \Cref{lem:vertexsparsification}. The second term on LHS follows from $|X|$ calls of $\kappa_G(s',s) \geq k$ queries of \Cref{lem:connectivity oracle}.

\section{Communication and Streaming}\label{sec:streaming}

\subsection{A Simple Communication and Streaming algorithm}
In this section, we will provide a simple communication algorithm for \SSUDVC{}, and its direct implementation in streaming model, i.e., the following two theorems.
\begin{lemma}\label{lem:communicationSSUDVC}
    There is a randomized communication protocol solving \SSUDVC{} in $\tOh{n^{1.5}}$.
\end{lemma}
\begin{lemma}\label{lem:streamingSSUDVC}
    If there is a randomized semi-streaming algorithm solving Bipartite Maximum Matching (BMM) in $P(n)$ passes, then there is a randomized streaming algorithm solving \SSUDVC{} in $\tOh{n^{1.5}}$ space and $O(P(n))$ passes.
\end{lemma}

We first give a schematic algorithm as described in \cref{alg:communication}, and we will show its correctness, and how to implement it in the communication and streaming setting. According to~\cref{rem:publicrandomness}, we assume public randomness.

\begin{algorithm}[!ht]
	\caption{Schematic Algorithm for \SSUDVC{}}
	\label{alg:communication}
	\KwData{An undirected graph $G=(V,E)$ and a vertex $t\in V$.}
	\KwResult{A minimum $t$-sink vertex cut $(L,S,R)$.}

            Let $s=\lceil\sqrt{n}\rceil$ in~\cref{thm:sketching}, calculate $sk_v=\sk{s}(N_G(v))$ for any $v\in V$\;

            Let $A=V-N_G(t)$\;

            \tcp{\color{blue}Suppose $(L,S,R)$ is a minimum $t$-sink vertex cut.}
            \tcp{\color{blue}Phase 1: $|L|\ge \sqrt{n}/3$}

            Sample $\left\lceil100\sqrt{n}\log n\right\rceil$ nodes $P\subseteq A$ uniformly at random \tcp{w.h.p. $P\cap L\not=\emptyset$}

            $S^*\leftarrow V$\;
            \ForEach{$v\in P$}
            {
                Find the minimum $(v,t)$ vertex cut in $G$ denoted by $S_v$\; \label{stremline66}
                Let $S^*=S_v$ if $|S_v|<|S^*|$\;\label{stremline77}
            }

            \tcp{\color{blue}Phase 2: $|L|< \sqrt{n}/3$}
            \ForEach{$u\in A$\label{streamline4}}
            {
                Maintain a virtual graph $G'=(V,E')$ and a set $V'\subseteq A$. Initially $E'=\emptyset,V'=\emptyset$\;
                \ForEach{$v\in A$\label{stremline6}}{
                    Let $X=N_G(u)-N_G(v)$ and $Y=N_G(v)-N_G(u)$, which is calculated by the recovery algorithm of~\cref{thm:sketching}\;
                    \If{the recovery algorithm does not return $\bot$}
                    {
                    Add $v$ to $V'$\;
                    Add $(u,x)$ to $E'$ for any $x\in X$\;\label{stremline10}
                    Add $(v,x)$ to $E'$ for any $x\in Y$\;\label{stremline11}
                    }
                }
                    Among all $L$ with $u\in L\subseteq V'$, find one with the smallest $|N_{G'}(L)|$ denoted as $L_u$\;\label{streamline14}
                    $\kappa_u\leftarrow |N_{G'}(L_u)|+|N_{G}(u)|-|N_{G'}(u)|$\tcp{w.h.p. $\kappa_u=|N_G(L_u)|$}
            }
        If there exists $u$ such that $\kappa_u<|S^*|$, let $S^*=N_G(L_u)$ where $\kappa_u$ is the smallest\;

        \Return{$S^*$}

\end{algorithm}

The following lemma shows the correctness of~\cref{alg:communication}.
\begin{lemma}\label{lem:correctnessFramework}
    At the end of~\cref{alg:communication}, $S^*$ is the minimum $t$-sink vertex cut in $G$.
\end{lemma}
\begin{proof}
    Suppose $(L,S,R)$ is a minimum $t$-sink vertex cut, suppose $|L|\ge \sqrt{n}/3$, then w.h.p. $P\cap L\not=\emptyset$, thus, $S^*$ has size at most $|S|$ according to~\cref{stremline66,stremline77}. In order to prove $|S^*|=|S|$, we first need to prove $\kappa_u\ge \kappa_G(u,t)$ for any $u$. Let us first prove some helpful claims.

    For any arbitrary $u\in L$, let us consider the loop for $u$ in~\cref{streamline4}.
    \begin{claim}
        W.h.p., for any $v_1,v_2\in V$ such that $v_1\in V',(v_1,v_2)\in E,(v_1,v_2)\not\in E'$ (which we say $(v_1,v_2)$ is \emph{missing}) iff. one of the following events happen: (i) for any $v\in V'$, we have $(v,v_2)\in E$, (ii) $(u,v_2)\in E,v_1\not = u$.
    \end{claim}
    \begin{proof}
        ``$\Leftarrow$''~\cref{stremline10,stremline11} only add edges to $x$ where $u,v\in V'$ differs by the neighborhood relationship to $x$, so $(i)$ implies $(v_1,v_2)$ is missing; if $v_1\not=u$, then $(v_1,v_2)$ will only be added to $E'$ when $(u,v_2)\not\in E$, so $(ii)$ implies $(v_1,v_2)$ is missing.
    ``$\Rightarrow$''. If both (i),(ii) do not happen, then either (i) $v_1=u$, there exists $v\in V'$ such that $(v,v_2)\not\in E$, then $v_2$ is in the difference between $A_u,A_v$, so $(u,v_2)$ is not missing, or (ii) $v_1\not=u$, then $(u,v_2)\not\in E$, so $v_2$ is in the difference between $A_u,A_{v_1}$, $(u,v_2)$ is not missing.
    \end{proof}

    \begin{claim}\label{cla:NGLu}
        W.h.p., for any $L\subseteq V'$, $N_{G}(L)-N_{G'}(L)=N_G(u)-N_{G'}(u)$.
    \end{claim}
    \begin{proof}
        Notice that for any $u\in L\subseteq V'$, a vertex $v\in N_G(L)-N_{G'}(L)$ iff. $\forall w\in V', (w,v)\in E$. Proof: ``$\Rightarrow$'', if $v\in N_G(L)-N_{G'}(L)$, then $(u,v)$ is missing, which means (i) happens. ``$\Leftarrow$'', it implies $(w,v)$ is missing for any $w\in V'$.

        Based on the claim above, we prove $N_{G}(L)-N_{G'}(L)=N_G(u)-N_{G'}(u)$ as follows. ``$\subseteq$'' if $w\in N_{G}(L)-N_{G'}(L)$, then $\forall w\in V', (w,v)\in E$, so $w\in N_G(u)$, and $w\not\in N_{G'}(u)$ (otherwise $w\in N_{G'}(L)$, which leads to $w\not\in N_{G}(L)-N_{G'}(L)$). ``$\supseteq$'' if $w\not\in N_{G}(L)-N_{G'}(L)$, then there exists $v\in V'-L$ such that $(v,w)\not\in E$, then $(u,w)$ is not missing and $w\not\in N_G(u)-N_{G'}(u)$. So we have $|N_G(L)|=|N_{G'}(L_u)|+deg_G(u)-deg_{G'}(u)$. $C_u$ correctly computes the size of $N_G(L_u)$.

    \end{proof}

    Now we are ready to prove that for any $u\in A$ we have $\kappa_u\ge \kappa_G(u,t)$. $\kappa_u=N_{G'}(L_u)+N_G(u)-N_{G'}(u)=N_G(L_u)$ according to~\cref{cla:NGLu}, thus, $\kappa_u\ge\kappa_G(u,t)$.

    Now suppose $|L|<\sqrt{n}/3$. We have $|S^*|\ge |S|$ according to the fact that $\kappa_u\ge\kappa_G(u,t)$ for any $u$. Suppose $u\in L$. Consider the loop for $u$ in~\cref{streamline4}. According to~\cref{lem:unbalancedclose} and the fact that $V'$ includes all nodes $v$ with $|N_G(u)\triangle N_G(v)|\le\lceil\sqrt{n}\rceil$ w.h.p., we have $L\subseteq V'$. Since $N_{G'}(L_u)$ and $N_{G}(L_u)$ differs by a fixed value according to~\cref{cla:NGLu}, $\kappa_u\le |N_G(L)|$. Thus, $|S^*|=|S|$.

\end{proof}

Now we prove the complexity and our main lemmas of this section.

\begin{proof}[Proof of~\cref{lem:streamingSSUDVC}]

    The correctness is implied by~\cref{lem:correctnessFramework}. Now we show the complexity.

    \paragraph{Before Phase 1.} $\mathbf{V}$ can be sampled without accessing the stream. The sketches can be computed in $\tOh{n\sqrt{n}}$ space and $1$ pass according to~\cref{rem:sketchingimplementation}. $A$ can be computed with $O(n)$ space and $1$ pass.

    \paragraph{Phase 1.} $P$ is sampled without accessing the stream. We simulate the $(S,T)$-vertex connectivity algorithm on the stream for any $u\in P$. We do all of them in parallel, using $\tOh{n|P|}$ space and $P(n)$ passes. Remember that $\tOh{n|P|}=\tOh{n\sqrt{n}}$.

    \paragraph{Phase 2.} All calculations of Phase 2 are done without accessing the stream, and they will use $\tOh{n^{1.5}}$ space since the number of edges added to $G'$ for each loop in~\cref{stremline6} is at most $\sqrt{n}$.

    In summary, we use $\tOh{n^{1.5}}$ space and $O(P(n))$ passes.

\end{proof}

\begin{proof}[Proof of~\cref{lem:communicationSSUDVC}]
    We use the folklore simulation of a streaming algorithm using a communication protocol. Alice and Bob simulate the streaming algorithm in the previous proof on the data stream where the stream first goes through all of Alice's edges, and then Bob's edges. The memory space of the streaming algorithm (which is $\tOh{n^{1.5}}$ is sent to Bob, and Bob continues simulating the streaming algorithm).

    We can simulate the algorithm before Phase 1 and Phase 2 using the technique above, which requires $\tOh{n^{1.5}}$ communication.
\end{proof}

\subsection{Lower Bounds}\label{subsec:n15lowerbound}

In this section, we will prove a communication lower bound for \UDVC{}.
\begin{lemma}\label{lem:lowerboundsundirvc}
    In the two-party communication model, any randomized algorithm solving \UDVC{} needs $\Omega(n^{1.5})$ bits of communication.
\end{lemma}

It directly implies a streaming lower bound by standard reduction from communication to streaming.

\begin{corollary}\label{lem:lowerboundsundirvc}
    In the graph streaming model, any randomized algorithm solving \UDVC{} in $P(n)$ pass needs $\Omega(n^{1.5}/P(n))$ space.
\end{corollary}
\begin{proof}
    As in the last section, we try to simulate any streaming algorithm in the following way: the stream first goes through all of Alice's edges, and then Bob's edges. The memory space of the streaming algorithm is sent to Bob after Alice's edges are all presented, and Bob continues simulating the streaming algorithm. Thus, solving \UDVC{} in $P(n)$ pass needs $o(n^{1.5}/P(n))$ space implies a communication algorithm with $o(n^{1.5})$ communication.

\end{proof}

The following communication problem is where we will reduce \UDVC{} from. For convenience, suppose $\sqrt{n}$ is an even integer.

\paragraph{$\sqrt{n}$-OR $\sqrt{n}$-AND subset problem.} Alice is given $\sqrt{n}$ batches, the $i$-th batch contains $\sqrt{n}$ sets $A^{(i)}_1,A^{(i)}_2,....,A^{(i)}_{\sqrt{n}} \subseteq [\sqrt{n}]$, each with size exactly $\sqrt{n}/2$. Bob is given $\sqrt{n}$ batches, the $i$-th batch contains $\sqrt{n}$ sets $B^{(i)}_1,B^{(i)}_2,....,B^{(i)}_{\sqrt{n}} \subseteq [\sqrt{n}]$. Their goal is to determine whether there exists $i\in[\sqrt{n}]$ such that for any $j\in[\sqrt{n}]$, $B^{(i)}_j\subseteq A^{(i)}_j$.

The following lemma shows the randomized communication complexity of the problem above.

\begin{lemma}\label{lem:ORANDsd}
    There exists $k=\Theta(\sqrt{n})$ such that any randomized protocol solving $\sqrt{n}$-OR $\sqrt{n}$-AND $k$-subset problem requires $\Omega(n^{1.5})$ communication.
\end{lemma}
\begin{proof}
We will reduce the following problem to $\sqrt{n}$-OR $\sqrt{n}$-AND $k$-subset. Alice is given $\sqrt{n}$ sets $X_1,...,X_{\sqrt{n}}\subseteq[n/2]$, Bob is given $Y_1,...,Y_{\sqrt{n}}\subseteq[n/2]$, and they want to determine whether or not there exists $i$ such that $X_i\cap Y_i\not=\emptyset$. This problem has $\Omega(n^{1.5})$ lower bound \cite{JayramKS03,HarshaJ13}. Now we construct a $\sqrt{n}$-OR $\sqrt{n}$-AND $k$-subset instance as follows. Firstly, define $X'_i=[n/2]-X_i$. Then, we split $X'_i$ into $\sqrt{n}$ blocks and fit each block into $[\sqrt{n}/2]$, i.e., define the $j$-th block $X^{(i)}_j=\{k\mid k\in[\sqrt{n}/2],(j-1)(\sqrt{n}/2)+k\in X'_i\}$. We do the same thing to $Y_i$, let $Y^{(i)}_j=\{k\mid k\in[\sqrt{n}/2],(j-1)(\sqrt{n}/2)+k\in Y_i\}$ Define $A^{(i)}_j=X^{(i)}_j\cup\{\sqrt{n}/2+k\mid 1\le k\le \sqrt{n}/2-|X^{(i)}_j|\}$. Similarly, define $B^{(i)}_j=Y^{(i)}_j\cup \{\sqrt{n}/2+k\mid 1\le k\le \sqrt{n}/2-|X^{(i)}_j|\}$. Roughly speaking, we do padding to make sure $|A^{(i)}|=\sqrt{n}/2$, while make sure $Y^{(i)}_j\subseteq X^{(i)}_j$ iff. $B^{(i)}_j\subseteq A^{(i)}_j$. Notice that ``$Y^{(i)}_j\subseteq X^{(i)}_j$ for all $j$'' iff. ``$Y_i\subseteq X'_i$'' iff. ``$X_i\cap Y_i=\emptyset$''. Thus, the $\sqrt{n}$-OR $\sqrt{n}$-AND subset problem on $A^{(i)}_j,B^{(i)}_j$ output ``yes'' iff. the OR of set-disjointness problem on $X_i,Y_i$ outputs ``yes''.

\end{proof}
Now we are ready to prove the lower bound for~\UDVC{}.

\begin{proof}
    Given a $\sqrt{n}$-OR $\sqrt{n}$-AND subset instance $A^{(i)}_j,B^{(i)}_j$ for $i,j\in[\sqrt{n}]$, define graph $G=(V,E)$ as follows. We will assume $\sqrt{n}\ge 10$.

    \begin{align*}
    V=&U\cup \left(\cup_{i\in[\sqrt{n}]}V^{(i)}\right)\\
    &U=\{u^{(i)}_j\mid i,j\in[\sqrt{n}]\}\\
    &V^{(i)}=\{v^{(i)}_j\mid j\in[\sqrt{n}]\}\\
    E=&E_{clique}\cup E_{V}\cup E_A\cup E_B\\
    &E_{clique}=\{(u,v)\mid u,v\in U, u\not=v\}\\
    &E_{V}=\bigcup_{i\in[\sqrt{n}]}\{(u,v)\mid u,v\in V_i,u\not=v\}\\
    &E_A=\bigcup_{i\in[\sqrt{n}]}\left(\bigcup_{j\in[\sqrt{n}]}\{(v^{(i)}_{j},u^{(x)}_{y})\mid x\in[\sqrt{n}],x\not=j,y\in A^{(i)}_x\}\right)\\
    &E_B=\bigcup_{i\in[\sqrt{n}]}\left(\bigcup_{j\in[\sqrt{n}]}\{(v^{(i)}_{j},u^{(j)}_{y})\mid y\in B^{(i)}_j\}\right)
    \end{align*}

    We first prove that if there exists $i\in[\sqrt{n}]$ such that for any $j\in[\sqrt{n}]$, $B^{(i)}_j\subseteq A^{(i)}_j$, then there exists a vertex cut of size $n/2$. We claim that $N(V_i)$ is the vertex cut that we want. Firstly, notice that $V_j\cap N(V_i)=\emptyset$ for any $j\not=i$, so $N(V_i)$ is a valid vertex cut. Further, according to $E_A,E_B$ and the fact that $B^{(i)}_j\subseteq A^{(i)}_j$ for any $i$, we know $N(V_i)=\{u^{(x)}_y\mid x\in[\sqrt{n}],y\in A^{(i)}_x\}$. Thus, $|N(V_i)|=\sum_{x\in[\sqrt{n}]}|A^{(i)}_x|=n/2$ (remember that $|A^{(i)}_x|=\sqrt{n}/2$).

    Now we prove that if for any $i\in[\sqrt{n}]$, there exists $j\in[\sqrt{n}]$, $B^{(i)}_j\not\subseteq A^{(i)}_j$, then any vertex cut has size at least $n/2+1$. Taking any two vertices $u,v\in V$, we will prove that there are at least $n/2+1$ internal vertex disjoint paths connecting $u$ and $v$. If $u,v\in V_i$ for some $i$ or $u,v\in U$, then there is an edge between $u,v$. Thus, we only need to consider the following two cases.

    \paragraph{Case 1.} If $u\in V_i$ for some $i$ and $v\in U$. We will first construct $n/2+1$ paths $p^u_1,...,p^u_{n/2+1}$ such that (i) they start with vertex $u$, (ii) they are vertex disjoint except the starting vertex $u$. Suppose $j\in[\sqrt{n}]$ satisfies $B^{(i)}_j\not\subseteq A^{(i)}_j$. Suppose $u=v^{(i)}_{j'}$ for some $j'$. The paths contain three types
    \begin{enumerate}
        \item Type 1: length $1$ paths from $u$ to every vertex in $\{u^{(x)}_y\mid x\in[\sqrt{n}],x\not=j',y\in A^{(i)}_x\}$. There are $(\sqrt{n}/2)\cdot (\sqrt{n}-1)$ many type 1 paths.
        \item Type 2: pick arbitrary $\sqrt{n}/2$ vertices in $V_i$ not identical to $v_j$ and $v_{j'}$ (remember that $\sqrt{n}>10$ and $|V_i|=\sqrt{n}$), denoted as $P$. Notice that every vertex in $P$ has edges to any node in $\{u^{(j')}_y\mid y\in A^{(i)}_{j'}\}$ according to the definition of $E_A$. We one-to-one match these nodes to $\{u^{(j')}_y\mid y\in A^{(i)}_{j'}\}$. Each type two path is $(u,v,m(v))$ where $v\in P$ and $m(v)$ is the matched node. There are $\sqrt{n}/2$ type 2 paths.
        \item Type 3: Let $z\in B^{(i)}_j\backslash A^{(i)}_j$, such $z$ must exist since $B^{(i)}_j\not\subseteq A^{(i)}_j$. The type 3 path is $(u,v_j,u^{j}_z)$.
    \end{enumerate}

    There are in total $n/2+1$ paths, and one can verify that they are vertex disjoint except $u$. Now since $U$ is a clique, we can find $n/2+1$ internal vertex disjoint path from $u$ to $v$ by connecting the $n/2+1$ end points of the $n/2+1$ paths to $v$.

    \paragraph{Case 2.} If $u\in V_i$ and $v\in V_j$ for some $i\not=j$. According to case 1, both $u$ and $v$ has paths $p^u_1,....,p^u_{n/2+1}$ with end point set $P_1\subseteq U$ and $p^v_1,...,p^v_{n/2+1}$ with end point set $P_2\subseteq U$. For each $w\in P_1\cap P_2$, we connect the two paths $p^u_i,p^v_j$ that both take $w$ as end point, which leads to a path connecting $u$ and $v$. For the rest path, we match them one-by-one arbitrarily and connect $p^u_i,p^v_j$ by connecting their end point by an edge (remember that $U$ is a clique). Thus, we get $n/2+1$ internal vertex disjoint path from $u$ to $v$.

    \paragraph{Conclusion.} The graph $G$ can be constructed with $O(\log n)$ communication: $E_{clique},E_V,E_A$ is known by Alice and $E_B$ is known by Bob. The value answer to \UDVC{} solves the $\sqrt{n}$-OR $\sqrt{n}$-AND subset problem. According to~\cref{lem:ORANDsd}, \UDVC{} has lower bound $\Omega(n^{1.5})$.

\end{proof}

\section{Reductions Between Problems and Proofs of Main Theorems}\label{subsec:reductions}

For convenience, all the functions appearing in this section (for example, $W(m,n)$ in \cref{lem:ST vc problem in matching}), we assume they are \emph{smooth}, in the sense that $W(Km,Kn)=\poly(K)\cdot W(m,n)$.

\paragraph{From $(S,T)$-vertex connectivity to \BMC{}.}  We describe the reduction from the $(S,T)$ vertex connectivity problem to \BMC{} (bipartite minimum vertex cover) using the construction from \cite{schrijver2003combinatorial}(Section 16.7c), see also \cite{BlikstadBEMN22,BrandLNPSSSW20}.

\begin{definition} [Construction] \label{def:reduction bipartite}
    Let $G = (V,E)$ be a graph with two vertex sets $A,B \subseteq V$. We construct the bipartite matching instance as follows. We may assume that $A\cap B = \emptyset$.   We define the bipartite graph $H =(V',V'',E')$ as follows.  For each $v \in V - A$, we add a new vertex $v'$ to $V'$, and for each $v \in V- B$ we add a new vertex $v''$ to $V''$. We define $E'$ as the set of pairs $(u',v'')$ with $u \in V - A, v \in V - B$ with the property that $(u,v) \in E$ or $u = v$. We define the capacity of each $v' \in V'$ to be $\infty$  if $v \in B$, and $1$ otherwise. Similarly, we define the capacity of each $v'' \in V''$ to be $\infty$ if $v \in A$ and $1$ otherwise.
\end{definition}

We can extract a min $(A,B)$-separator in $G$ from the minimum vertex cover of $H$ as follows. For any $U \subseteq V$, we denote $U' = \{ v' \mid v \in U\}$, and $U'' = \{ v'' \mid v \in U\}$.

\begin{lemma} [\cite{schrijver2003combinatorial}(Section 16.7c)] \label{lem:recovering vertex cut}
Let $U \subseteq V - B, W \subseteq V - A$ be such that $D = U \cup W$ is a minimum vertex cover in $H$. Then, the set $C = (U \cap A) \cup (U \cap W) \cup (W \cap B)$ is a min $(A,B)$-separator in $G$.
\end{lemma}

  Observe that the biparite graph $H$ in \Cref{def:reduction bipartite} has $2n - |A| - |B| = O(n)$ vertices. Also, the construction above contains some nodes with $\infty$ capacity. We can reduce to unit capacity WLOG as follows. First, we replace $\infty$ with $n$. For each node with capacity $n$, we simulate it to behave as if it is a duplicate of $n$ unit-capacity nodes. With appropriate implementation, we apply vertex reduction algorithm in Appendix B of \cite{assadi2022semi} to reduce to $O(n)$ nodes while preserving the maximum matching. Roughly speaking, the vertex reduction algorithm greedily finds $O(\log n)$ maximal matchings and preserves only the nodes in all these maximal matchings. Finding a maximal matching in a node-capacity graph can be done efficiently.

  In conclusion, we have the following.

\begin{lemma}  \label{lem:ST vc problem in matching}
If the \BMC{} problem with $n$ vertices, $m$ vertices and diameter $D$ can be solved in
\begin{itemize}[noitemsep]
    \item (\pram{}) $W(m,n)$ work and $D(m,n)$ depth,
    \item (\congest{}) $T(m,n,D)$ rounds,
    \item (Communication) $f(n)$ bits of communication,
    \item (Streaming) $s(n)$ space, $p(n)$ passes,
\end{itemize}
then $(S,T)$-vertex connectivity problem can be solved in
 \begin{itemize}[noitemsep]
    \item (\pram{}) $O(W(m,n))$ work and $O(D(m,n))$ depth,
    \item (\congest{}) $O(T(m,n,D))$ rounds,
    \item (Communication) $O(f(n))$ bits of communication,
    \item (Streaming) $O(s(n))$ space, $O(p(n))$ passes,
 \end{itemize}
\end{lemma}

\paragraph{From (global) vertex connectivity to single sink vertex connectivity.}

\begin{lemma}[Global to Single Source]\label{lem:globaltosinglesource}
    Suppose there exist algorithms solving \SSUDVC{} with the following complexities in different models.
    \begin{enumerate}
        \item $T(m,n)$ work and $D(m,n)$ depth in the \pram{} model.
        \item $R(m,n,D)$ rounds in \congest{} model.
        \item $C(n)$ communication bits in the two-party communication model,
        \item $S(n)$ space and $P(n)$ passes in the streaming model,
    \end{enumerate}

\noindent
    Then there exist algorithms solving \UDVC{} in
    \begin{enumerate}
        \item $\tO{T(m,n)}$ work and $\tO{D(m,n)}$ depth in the \pram{} model.
        \item $\tO{R(m,n,D)}$ rounds in \congest{} model.
        \item $\tO{C(n)}$ communication bits in the two-party communication model,
        \item $\tO{S(n)}$ space and $\tO{P(n)}$ passes in the streaming model,
    \end{enumerate}

\end{lemma}

\begin{proof}

    The general idea is to first compute the minimum degree $\delta$, then sample $z=10n\log n/(n-\delta)$ nodes $s_1,...,s_z$, and denote $B_i=V-N_G[s]$. Now construct a graph $G'=(V',E')$ where
    \[V'=\{s\}\cup V'\cup \left(\cup_{i\in[z]}B_i\right)\]
    here $V'=\{v'\mid v\in V\}$ is a copy of $V$, and $s$ is a new node.
    \[E'=\{(s,v')\mid v'\in V'\}\cup \left(\cup_{i\in[z]}\{(v,u)\mid u,v\in B_i,(u,v)\in E\}\cup \{(u',v)\mid u'\in V',v\in B_i,(v,u)\in E\}\right)\]

    Now we query the \SSUDVC{} on $G'$ and source $s$.

    To prove the correctness, suppose $(L,S,R)$ is a minimum vertex cut of $G$. We first prove $G'$ has a cut separating $s$ to another node with size $|S|$. We have $\delta\ge |S|$, otherwise, the neighborhoods of the minimum degree node should be a smaller cut. Thus, $|L|+|R|=n-|S|\ge n-\delta$. By sampling $z=10n\log n/(n-\delta)$ nodes, one of them (suppose it is $s_i$) will be in $R$ w.h.p. since $|R|\ge |L|$. Thus, $(S\cap B_i)\cup \{u'\mid u\in S\cap N(S))\}$ is a cut in $G'$ which has size $|S|$. Now we prove than $\kappa_{G'}(s,t)\ge |S|$ for any $t\in V'$. If $t\in N[s]$, then $\kappa_{G'}(s,t)=|V'|-1$. Otherwise, $t\in B_i$ for some $i$. If there exists $|S|$ internal vertex disjoint path in $G$ from $t$ to $s$, then the $|S|$ paths corresponds to $|S|$ internal vertex disjoint path in $G'$ from $s$ to $t$.

    To recover the cut for $G$ from the cut for $G'$, notice that the cut $(L,S,R)$ in $G'$ with $s\in R$ has the property that $L\subseteq B_i$ if $L$ is connected. Thus, $S$ corresponds to a cut in $G'$.

    The size of $V'$ is at most $(n+1)+(n-\delta)\cdot 10n\log n/(n-\delta)\le n'$.

    Now we prove the complexities in different models by showing how to implement the aforementioned algorithm.

    \paragraph{\pram{}.} If $|S|<0.1n$, the reduction is trivial since $O(\log n)$ sampled nodes suffice to hit a node in $R$. Thus, we assume $|S|\ge 0.1n$, in which case we have $m=\Omega(n^2)$. The minimum degree can be computed in $O(m)$ time. Sampling can be done in $O(n)$ time. Construct the graph $G'$ cost $\tOh{n^2}$ time. The total work is $\tOh{T(m,n)}$. The depth is constant plus the oracle depth.

    \paragraph{\congest{}.} The implementation for \congest{} is a bit different since we cannot explicitly construct $G'$. Instead, after sample $z$ sinks $s_1,...,s_z$, we run \SSUDVC{} on the subgraphs $H_1,...,H_z$ defined as $H_i$ containing all edges adjacent to $V-N_G[s_i]$ plus edges from $N_G(s_i)$ to $s_i$, with the sinks $s_1,...,s_z$. Notice that $H_i$ has the same diameter as the original graph since we are only ignoring edges between $N(s_i)$, which has mutual distances at most $2$. We run all the algorithm on $H_i$ simultaneously, which causes dilation at most $R(n,D)$. Now we bound the congestion, i.e., given an arbitrary edge $(u,v)$, we will bound the number of $H_i$ that includes $(u,v)$. Specifically, we are going to bound the number of $H_i$ including $(u,v)$ through (i) $(u,v)$ is an edge adjacent to $s_i$, this can happen at most $O(\log n)$ times w.h.p., (ii) $(u,v)$ is an edge adjacent to $V-N_G[s_i]$, in other words $s_i$ is not a neighborhood of $u$ or is not a neighborhood of $v$, according to Chernoff bound, both event can happen for at most $O(\log n)$ times since $u$ has degree at least $\delta$ and we sample $s_i$ uniformly at random where each node get sampled with probability $10\log n/(n-\delta)$. Thus, the total congestion is $\tO{R(m,n,D)}$.

    \paragraph{Communication.} The minimum degree can be computed in $\tOh{n}$ communication by computing $|N_G(v)|$ in $\tOh{1}$ communication. Sampling can be done by Alice and then the results are sent to Bob. The vertex set $B_i$ is known to both Alice and Bob in $\tOh{n-\delta}$ communication by using~\cref{thm:sketching} on $V$ and $N_G[s]$, which in total cost $\tOh{n}$ communication. Now edges in $G'$ are split into two parts for Alice and Bob, and they can simulate the \SSUDVC{} algorithm on $G'$ in $C(n')$ communication.

    \paragraph{Streaming.} The minimum degree can be computed in one pass and $\tOh{n}$ space. Sampling can be done without accessing the graph. The vertex sets $B_i$ can be computed in one pass and $\tOh{n}$ space (for all $i$) by~\cref{thm:sketching}. The we can simulate the stream for $G'$ by using the stream for $G$, in $S(n')$ space and $P(n')$ passes.

\end{proof}

\paragraph{From s-t vertex connectivity to (global) vertex connectivity.}
\begin{lemma}[s-t vertex connectivity to (global) vertex connectivity]\label{lem:sttoglobal}
    Suppose there exist algorithms solving \UDVC{} with the following complexities in different models.

    \begin{enumerate}
        \item $T(n)$ work and $D(n)$ depth in \pram{} model.
        \item $C(n)$ communication bits in the two-party communication model,
        \item $S(n)$ space and $P(n)$ passes in the streaming model,
    \end{enumerate}

\noindent
        Then there exist algorithms solving s-t \UDVC{} in
    \begin{enumerate}
        \item $O(C(n))$ communication bits in the two-party communication model,
        \item $O(S(n))$ space and $O(P(n))$ passes in the streaming model,
        \item $O(T(n))$ work and $O(D(n))$ depth in \pram{} model.
    \end{enumerate}

\end{lemma}
\begin{proof}

    By \cref{lem:ST vc problem in matching}, we only need to reduce \BMC{} (bipartite minimum vertex cover) to \UDVC{}.

    Given a bipartite graph $G=(A,B,E)$ where $E\subseteq A\times B$, we construct a graph $G'=(A\cup B,E\cup E_A\cup E_B)$ where $E_A=\{(u,v)\mid u,v\in A\}$ and $E_B=\{(u,v)\mid u,v\in B$, i.e., we build a clique on both $A$ and $B$. We run \UDVC{} on $G'$, which gives us a minimum vertex cut $S$. We return (i) $S$ in the case that $|S|<\min(|A|,|B|)$, (ii) the smaller one among $A,B$ in the case that $|S|\ge \min(|A|,|B|)$.

    We first show that every vertex cut in $G'$ must be a vertex cover on $G$. Suppose $(L,S,R)$ is a vertex cut in $G'$, since there are no edges between $L$ and $R$, it cannot happen that both $L\cap A$ and $R\cap A$ are non-empty. W.o.l.g., assume $L\cap A$ is non-empty. If $L\cap B$ is non-empty as well, then $R$ cannot contain any node in $A$ or $B$, which is a contradiction to the fact that $R$ cannot be empty. Thus, $R\cap A=\emptyset$, so we have $L\subseteq A$. Similarly, we have $R\subseteq B$.
    For every edge in $G$, it is either adjacent to $S$, or both of its endpoints are inside $A$ (or inside $B$). But edges in $G'$ that are totally inside $A$ or $B$ are not edges in $G$, thus, $S$ is a vertex cover for $G$.

    Then we show that if there is a vertex cover $C$ in $G$ which satisfies that $|C|<\min(|A|,|B|)$, then $C$ must be a vertex cut for $G'$. W.o.l.g., assume $|A|\le |B|$. We have that $|C|<|A|$. Let $L=A-C$ and $R=B-C$. We claim that $(L,V-C,R)$ is a vertex cut. It is clear that $L,V-C,R$ form a partition of $A\cup B$. Moreover, if there exists an edge between $L$ and $R$, then this edge is not covered by $C$ in $G$, which is a contradiction. Thus, there are no edges between $L$ and $R$.

    Given the last two paragraphs, suppose $S^*$ is the returned set by our reduction algorithm, then either (i) the minimum vertex cover of $G$ is less than $\min(|A|,|B|)$, in which case there exists a veretx cut in $G'$ with the same size of the minimum vertex cover, so $S^*$ cannot be larger, and cannot be smaller since it is a vertex cover itself, or (ii) the minimum vertex cover of $G$ is $\min(|A|,|B|)$, then we returned the correct thing.

\end{proof}

Now we are ready to prove the main theorems.

\begin{proof}[Proof of \cref{thm:prammain}]
    Given an s-t vertex connectivity algorithm, we can solve \SSUDVC{} according to \cref{lem:SSVCpram}, then we can solve \UDVC{} using \cref{lem:globaltosinglesource}.
\end{proof}

\begin{proof}[Proof of \cref{thm:congestmain}]
    Given an S-T vertex connectivity algorithm, we first convert it into a subgraph S-T vertex connectivity algorithm by the same reduction described in \cref{sec:BMM}.
    Then we can solve \SSUDVC{} according to \cref{lem:SSVCcongest}, then we can solve \UDVC{} using \cref{lem:globaltosinglesource}.

\end{proof}

\begin{proof}[Proof of \cref{thm:pramexactrunningtimeVC}]
    We plug in \cref{thm:prammain} with the work and depth for s-t vertex connectivity. The latter one can be derived from the maximum bipartite matching algorithm of \cref{thm:pram-bmm}. To convert maximum bipartite matching algorithm to s-t vertex connectivity algorithm, we first convert it into bipartite minium vertex cover algorithm by running a reachability on the residual graph to find the minium vertex cover, then use \cref{lem:ST vc problem in matching}.
\end{proof}

\begin{proof}[Proof of \cref{thm:congestexactrunningtimeVC}]
    We plug in \cref{thm:congestmain} with the round complexity for S-T vertex connectivity. The round complexity of S-T vertex connectivity can be derived from \cref{thm:congest-bmm} combined with \cref{lem:ST vc problem in matching}.
\end{proof}

\begin{proof}[Proof of \cref{thm:pramexactrunningtimeVCmaxtrix}]
    We plug in \cref{lem:globaltosinglesource} with \cref{lem:SSVCprammatrix}.
\end{proof}

\begin{proof}[Proof of \cref{thm:ccmain}]
    We combine \cref{lem:communicationSSUDVC} with \cref{lem:globaltosinglesource}, and \cref{lem:lowerboundsundirvc}.
\end{proof}

\begin{proof}[Proof of \cref{thm:streamingupperbounds}]
    We combine \cref{lem:streamingSSUDVC} with \cref{lem:globaltosinglesource}.
\end{proof}

\section*{Acknowledgements}
 The project is partially supported by the European Research
Council (ERC) under the European Union’s Horizon 2020 research and innovation programme (grant agreement no. 759557–ALGOCom), and the ERC Starting Grant (CODY 101039914), Dr. Max Rössler, the Walter Haefner Foundation, and the ETH Zürich Foundation.
We thank Danupon Nanongkai for the helpful discussions and the anonymous reviewers for their helpful suggestions on the writing.

\bibliographystyle{alpha}
\bibliography{ref,ref-sorrachai}

\newcommand{\etalchar}[1]{$^{#1}$}
\begin{thebibliography}{vdBGJdV25}

\bibitem[ABC{\etalchar{+}}24]{AshvinkumarBCGH24}
Vikrant Ashvinkumar, Aaron Bernstein, Nairen Cao, Christoph Grunau, Bernhard Haeupler, Yonggang Jiang, Danupon Nanongkai, and Hsin{-}Hao Su.
\newblock Parallel, distributed, and quantum exact single-source shortest paths with negative edge weights.
\newblock In {\em {ESA}}, volume 308 of {\em LIPIcs}, pages 13:1--13:15. Schloss Dagstuhl - Leibniz-Zentrum f{\"{u}}r Informatik, 2024.

\bibitem[ABCP98]{AwerbuchBCP98}
Baruch Awerbuch, Bonnie Berger, Lenore Cowen, and David Peleg.
\newblock Near-linear time construction of sparse neighborhood covers.
\newblock {\em {SIAM} J. Comput.}, 28(1):263--277, 1998.

\bibitem[ABN08]{AbrahamBN08}
Ittai Abraham, Yair Bartal, and Ofer Neiman.
\newblock Nearly tight low stretch spanning trees.
\newblock In {\em {FOCS}}, pages 781--790. {IEEE} Computer Society, 2008.

\bibitem[AG11]{AhnG11}
Kook~Jin Ahn and Sudipto Guha.
\newblock Linear programming in the semi-streaming model with application to the maximum matching problem.
\newblock In {\em {ICALP} {(2)}}, volume 6756 of {\em Lecture Notes in Computer Science}, pages 526--538. Springer, 2011.

\bibitem[AJJ{\etalchar{+}}22]{assadi2022semi}
Sepehr Assadi, Arun Jambulapati, Yujia Jin, Aaron Sidford, and Kevin Tian.
\newblock Semi-streaming bipartite matching in fewer passes and optimal space.
\newblock In {\em Proceedings of the 2022 Annual ACM-SIAM Symposium on Discrete Algorithms (SODA)}, pages 627--669. SIAM, 2022.

\bibitem[AKO18]{AhmadiKO18}
Mohamad Ahmadi, Fabian Kuhn, and Rotem Oshman.
\newblock Distributed approximate maximum matching in the {CONGEST} model.
\newblock In {\em {DISC}}, volume 121 of {\em LIPIcs}, pages 6:1--6:17. Schloss Dagstuhl - Leibniz-Zentrum f{\"{u}}r Informatik, 2018.

\bibitem[ALT21]{AssadiLT21}
Sepehr Assadi, S.~Cliff Liu, and Robert~E. Tarjan.
\newblock An auction algorithm for bipartite matching in streaming and massively parallel computation models.
\newblock In {\em {SOSA}}, pages 165--171. {SIAM}, 2021.

\bibitem[AR20]{AssadiR20}
Sepehr Assadi and Ran Raz.
\newblock Near-quadratic lower bounds for two-pass graph streaming algorithms.
\newblock In {\em {FOCS}}, pages 342--353. {IEEE}, 2020.

\bibitem[AS23]{AssadiS23}
Sepehr Assadi and Vihan Shah.
\newblock Tight bounds for vertex connectivity in dynamic streams.
\newblock In {\em {SOSA}}, pages 213--227. {SIAM}, 2023.

\bibitem[BBE{\etalchar{+}}22]{blikstad2022nearly}
Joakim Blikstad, Jan van~den Brand, Yuval Efron, Sagnik Mukhopadhyay, and Danupon Nanongkai.
\newblock Nearly optimal communication and query complexity of bipartite matching.
\newblock {\em arXiv preprint arXiv:2208.02526}, 2022.

\bibitem[BDD{\etalchar{+}}82]{BeckerDDHKKMNRW82}
Michael Becker, W.~Degenhardt, J{\"{u}}rgen Doenhardt, Stefan Hertel, Gerd Kaninke, W.~Kerber, Kurt Mehlhorn, Stefan N{\"{a}}her, Hans Rohnert, and Thomas Winter.
\newblock A probabilistic algorithm for vertex connectivity of graphs.
\newblock {\em Inf. Process. Lett.}, 15(3):135--136, 1982.

\bibitem[BLN{\etalchar{+}}20]{BrandLNPSSSW20}
Jan van~den Brand, Yin~Tat Lee, Danupon Nanongkai, Richard Peng, Thatchaphol Saranurak, Aaron Sidford, Zhao Song, and Di~Wang.
\newblock Bipartite matching in nearly-linear time on moderately dense graphs.
\newblock In {\em {FOCS}}, pages 919--930. {IEEE}, 2020.

\bibitem[BNW22]{BernsteinNW22}
Aaron Bernstein, Danupon Nanongkai, and Christian Wulff{-}Nilsen.
\newblock Negative-weight single-source shortest paths in near-linear time.
\newblock In {\em {FOCS}}, pages 600--611. {IEEE}, 2022.

\bibitem[BvdBE{\etalchar{+}}22]{BlikstadBEMN22}
Joakim Blikstad, Jan van~den Brand, Yuval Efron, Sagnik Mukhopadhyay, and Danupon Nanongkai.
\newblock Nearly optimal communication and query complexity of bipartite matching.
\newblock In {\em {FOCS}}, pages 1174--1185. {IEEE}, 2022.

\bibitem[CF14]{CormodeF14}
Graham Cormode and Donatella Firmani.
\newblock A unifying framework for {$\ell_0$}-sampling algorithms.
\newblock {\em Distributed Parallel Databases}, 32(3):315--335, 2014.

\bibitem[CF23]{CaoF23}
Nairen Cao and Jeremy~T. Fineman.
\newblock Parallel exact shortest paths in almost linear work and square root depth.
\newblock In {\em {SODA}}, pages 4354--4372. {SIAM}, 2023.

\bibitem[CGK14]{Censor-HillelGK14}
Keren Censor{-}Hillel, Mohsen Ghaffari, and Fabian Kuhn.
\newblock Distributed connectivity decomposition.
\newblock In {\em {PODC}}, pages 156--165. {ACM}, 2014.

\bibitem[CGLZ20]{CGLZ20}
Matthias Christandl, Fran{\c{c}}ois~Le Gall, Vladimir Lysikov, and Jeroen Zuiddam.
\newblock Barriers for rectangular matrix multiplication.
\newblock {\em CoRR}, abs/2003.03019, 2020.

\bibitem[CKL{\etalchar{+}}22]{ChenKLPGS22}
Li~Chen, Rasmus Kyng, Yang~P. Liu, Richard Peng, Maximilian~Probst Gutenberg, and Sushant Sachdeva.
\newblock Maximum flow and minimum-cost flow in almost-linear time.
\newblock In {\em {FOCS}}, pages 612--623. {IEEE}, 2022.

\bibitem[CKM{\etalchar{+}}14]{CohenKMPPRX14}
Michael~B. Cohen, Rasmus Kyng, Gary~L. Miller, Jakub~W. Pachocki, Richard Peng, Anup~B. Rao, and Shen~Chen Xu.
\newblock Solving {SDD} linear systems in nearly \emph{m}log\({}^{\mbox{1/2}}\)\emph{n} time.
\newblock In {\em {STOC}}, pages 343--352. {ACM}, 2014.

\bibitem[CKP{\etalchar{+}}21]{ChenKPS0Y21}
Lijie Chen, Gillat Kol, Dmitry Paramonov, Raghuvansh~R. Saxena, Zhao Song, and Huacheng Yu.
\newblock Almost optimal super-constant-pass streaming lower bounds for reachability.
\newblock In {\em {STOC}}, pages 570--583. {ACM}, 2021.

\bibitem[CKT93]{CheriyanKT93}
Joseph Cheriyan, Ming{-}Yang Kao, and Ramakrishna Thurimella.
\newblock Scan-first search and sparse certificates: An improved parallel algorithms for k-vertex connectivity.
\newblock {\em {SIAM} J. Comput.}, 22(1):157--174, 1993.

\bibitem[Coh95]{cohen1995approximate}
Edith Cohen.
\newblock Approximate max-flow on small depth networks.
\newblock {\em SIAM Journal on Computing}, 24(3):579--597, 1995.

\bibitem[CR94]{CheriyanR94}
Joseph Cheriyan and John~H. Reif.
\newblock Directed \emph{s-t} numberings, rubber bands, and testing digraph \emph{k}-vertex connectivity.
\newblock {\em Combinatorica}, 14(4):435--451, 1994.
\newblock Announced at SODA'92.

\bibitem[CS19]{ChangS19}
Yi{-}Jun Chang and Thatchaphol Saranurak.
\newblock Improved distributed expander decomposition and nearly optimal triangle enumeration.
\newblock In {\em {PODC}}, pages 66--73. {ACM}, 2019.

\bibitem[CS20]{ChangS20}
Yi{-}Jun Chang and Thatchaphol Saranurak.
\newblock Deterministic distributed expander decomposition and routing with applications in distributed derandomization.
\newblock In {\em {FOCS}}, pages 377--388. {IEEE}, 2020.

\bibitem[CT91]{CheriyanT91}
Joseph Cheriyan and Ramakrishna Thurimella.
\newblock Algorithms for parallel k-vertex connectivity and sparse certificates (extended abstract).
\newblock In {\em {STOC}}, pages 391--401. {ACM}, 1991.

\bibitem[DEMN21]{DoryEMN21}
Michal Dory, Yuval Efron, Sagnik Mukhopadhyay, and Danupon Nanongkai.
\newblock Distributed weighted min-cut in nearly-optimal time.
\newblock In {\em {STOC}}, pages 1144--1153. {ACM}, 2021.

\bibitem[dV23]{Vos23}
Tijn de~Vos.
\newblock Brief announcement: Minimum cost maximum flow in the {CONGEST} model.
\newblock In {\em {PODC}}, pages 71--74. {ACM}, 2023.

\bibitem[EH84]{EsfahanianH84}
Abdol{-}Hossein Esfahanian and S.~Louis Hakimi.
\newblock On computing the connectivities of graphs and digraphs.
\newblock {\em Networks}, 14(2):355--366, 1984.

\bibitem[EKMS12]{eggert2012bipartite}
Sebastian Eggert, Lasse Kliemann, Peter Munstermann, and Anand Srivastav.
\newblock Bipartite matching in the semi-streaming model.
\newblock {\em Algorithmica}, 63(1):490--508, 2012.

\bibitem[ET75]{EvenT75}
Shimon Even and Robert~Endre Tarjan.
\newblock Network flow and testing graph connectivity.
\newblock {\em {SIAM} J. Comput.}, 4(4):507--518, 1975.

\bibitem[Eve75]{Even75}
Shimon Even.
\newblock An algorithm for determining whether the connectivity of a graph is at least k.
\newblock {\em {SIAM} J. Comput.}, 4(3):393--396, 1975.

\bibitem[FGL{\etalchar{+}}21]{ForsterGLPSY21}
Sebastian Forster, Gramoz Goranci, Yang~P. Liu, Richard Peng, Xiaorui Sun, and Mingquan Ye.
\newblock Minor sparsifiers and the distributed laplacian paradigm.
\newblock In {\em {FOCS}}, pages 989--999. {IEEE}, 2021.

\bibitem[FRT04]{FakcharoenpholRT04}
Jittat Fakcharoenphol, Satish Rao, and Kunal Talwar.
\newblock A tight bound on approximating arbitrary metrics by tree metrics.
\newblock {\em J. Comput. Syst. Sci.}, 69(3):485--497, 2004.

\bibitem[Gab06]{Gabow06}
Harold~N. Gabow.
\newblock Using expander graphs to find vertex connectivity.
\newblock {\em J. {ACM}}, 53(5):800--844, 2006.
\newblock Announced at FOCS'00.

\bibitem[Gal80]{Galil80}
Zvi Galil.
\newblock Finding the vertex connectivity of graphs.
\newblock {\em {SIAM} J. Comput.}, 9(1):197--199, 1980.

\bibitem[GH20]{GianinazziH20}
Lukas Gianinazzi and Torsten Hoefler.
\newblock Parallel planar subgraph isomorphism and vertex connectivity.
\newblock In {\em {SPAA}}, pages 269--280. {ACM}, 2020.

\bibitem[GP88]{galil1988improved}
Zvi Galil and Victor Pan.
\newblock Improved processor bounds for combinatorial problems in rnc.
\newblock {\em Combinatorica}, 8(2):189--200, 1988.

\bibitem[GPST92]{GoldbergPST92}
Andrew~V. Goldberg, Serge~A. Plotkin, David~B. Shmoys, and {\'{E}}va Tardos.
\newblock Using interior-point methods for fast parallel algorithms for bipartite matching and related problems.
\newblock {\em {SIAM} J. Comput.}, 21(1):140--150, 1992.

\bibitem[GPV93]{GoldbergPV93}
Andrew~V. Goldberg, Serge~A. Plotkin, and Pravin~M. Vaidya.
\newblock Sublinear-time parallel algorithms for matching and related problems.
\newblock {\em J. Algorithms}, 14(2):180--213, 1993.

\bibitem[GZ22]{0001Z22}
Mohsen Ghaffari and Goran Zuzic.
\newblock Universally-optimal distributed exact min-cut.
\newblock In {\em {PODC}}, pages 281--291. {ACM}, 2022.

\bibitem[Hen97]{Henzinger97}
Monika~Rauch Henzinger.
\newblock A static 2-approximation algorithm for vertex connectivity and incremental approximation algorithms for edge and vertex connectivity.
\newblock {\em J. Algorithms}, 24(1):194--220, 1997.

\bibitem[HHL{\etalchar{+}}24]{HaeuplerH0RS24}
Bernhard Haeupler, D.~Ellis Hershkowitz, Jason Li, Antti Roeyskoe, and Thatchaphol Saranurak.
\newblock Low-step multi-commodity flow emulators.
\newblock In {\em {STOC}}, pages 71--82. {ACM}, 2024.

\bibitem[HJ13]{HarshaJ13}
Prahladh Harsha and Rahul Jain.
\newblock A strong direct product theorem for the tribes function via the smooth-rectangle bound.
\newblock In {\em {FSTTCS}}, volume~24 of {\em LIPIcs}, pages 141--152. Schloss Dagstuhl - Leibniz-Zentrum f{\"{u}}r Informatik, 2013.

\bibitem[HRG00]{HenzingerRG00}
Monika~Rauch Henzinger, Satish Rao, and Harold~N. Gabow.
\newblock Computing vertex connectivity: New bounds from old techniques.
\newblock {\em J. Algorithms}, 34(2):222--250, 2000.
\newblock Announced at FOCS'96.

\bibitem[HRG22]{HaeuplerR022}
Bernhard Haeupler, Harald R{\"{a}}cke, and Mohsen Ghaffari.
\newblock Hop-constrained expander decompositions, oblivious routing, and distributed universal optimality.
\newblock In {\em {STOC}}, pages 1325--1338. {ACM}, 2022.

\bibitem[JKS03]{JayramKS03}
T.~S. Jayram, Ravi Kumar, and D.~Sivakumar.
\newblock Two applications of information complexity.
\newblock In {\em {STOC}}, pages 673--682. {ACM}, 2003.

\bibitem[JM23]{JiangM23}
Yonggang Jiang and Sagnik Mukhopadhyay.
\newblock Finding a small vertex cut on distributed networks.
\newblock In {\em {STOC}}, pages 1791--1801. {ACM}, 2023.

\bibitem[Kar00]{Karger00}
David~R. Karger.
\newblock Minimum cuts in near-linear time.
\newblock {\em J. {ACM}}, 47(1):46--76, 2000.
\newblock announced at STOC'96.

\bibitem[Kle69]{Kleitman1969methods}
D~Kleitman.
\newblock Methods for investigating connectivity of large graphs.
\newblock {\em IEEE Transactions on Circuit Theory}, 16(2):232--233, 1969.

\bibitem[KUW86]{KarpUW86}
Richard~M. Karp, Eli Upfal, and Avi Wigderson.
\newblock Constructing a perfect matching is in random {NC}.
\newblock {\em Comb.}, 6(1):35--48, 1986.

\bibitem[LJS19]{LiuJS19}
Yang~P. Liu, Arun Jambulapati, and Aaron Sidford.
\newblock Parallel reachability in almost linear work and square root depth.
\newblock In {\em {FOCS}}, pages 1664--1686. {IEEE} Computer Society, 2019.

\bibitem[LLW88]{LinialLW88}
Nathan Linial, L{\'{a}}szl{\'{o}} Lov{\'{a}}sz, and Avi Wigderson.
\newblock Rubber bands, convex embeddings and graph connectivity.
\newblock {\em Combinatorica}, 8(1):91--102, 1988.
\newblock Announced at FOCS'86.

\bibitem[LMN21]{Lopez-MartinezM21}
Andr{\'{e}}s L{\'{o}}pez{-}Mart{\'{\i}}nez, Sagnik Mukhopadhyay, and Danupon Nanongkai.
\newblock Work-optimal parallel minimum cuts for non-sparse graphs.
\newblock In {\em {SPAA}}, pages 351--361. {ACM}, 2021.

\bibitem[LNP{\etalchar{+}}21]{LiNPSY21}
Jason Li, Danupon Nanongkai, Debmalya Panigrahi, Thatchaphol Saranurak, and Sorrachai Yingchareonthawornchai.
\newblock Vertex connectivity in poly-logarithmic max-flows.
\newblock In {\em {STOC}}, pages 317--329. {ACM}, 2021.

\bibitem[Lov79]{Lovasz79}
L{\'{a}}szl{\'{o}} Lov{\'{a}}sz.
\newblock On determinants, matchings, and random algorithms.
\newblock In {\em {FCT}}, pages 565--574. Akademie-Verlag, Berlin, 1979.

\bibitem[LP20]{LiP20}
Jason Li and Debmalya Panigrahi.
\newblock Deterministic min-cut in poly-logarithmic max-flows.
\newblock In {\em {FOCS}}, pages 85--92. {IEEE}, 2020.

\bibitem[LPSP15]{lotker2015improved}
Zvi Lotker, Boaz Patt-Shamir, and Seth Pettie.
\newblock Improved distributed approximate matching.
\newblock {\em Journal of the ACM (JACM)}, 62(5):1--17, 2015.

\bibitem[LS14]{LeeS14}
Yin~Tat Lee and Aaron Sidford.
\newblock Path finding methods for linear programming: Solving linear programs in {\~{o}}(sqrt(rank)) iterations and faster algorithms for maximum flow.
\newblock In {\em 55th {IEEE} Annual Symposium on Foundations of Computer Science, {FOCS} 2014, Philadelphia, PA, USA, October 18-21, 2014}, pages 424--433, 2014.

\bibitem[Mat87]{Matula87}
David~W. Matula.
\newblock Determining edge connectivity in o(nm).
\newblock In {\em {FOCS}}, pages 249--251. {IEEE} Computer Society, 1987.

\bibitem[MN20]{MukhopadhyayN20}
Sagnik Mukhopadhyay and Danupon Nanongkai.
\newblock Weighted min-cut: sequential, cut-query, and streaming algorithms.
\newblock In {\em {STOC}}, pages 496--509. {ACM}, 2020.

\bibitem[MN21]{MukhopadhyayN21-Note}
Sagnik Mukhopadhyay and Danupon Nanongkai.
\newblock A note on isolating cut lemma for submodular function minimization.
\newblock {\em CoRR}, abs/2103.15724, 2021.

\bibitem[MPX13]{MillerPX13}
Gary~L. Miller, Richard Peng, and Shen~Chen Xu.
\newblock Parallel graph decompositions using random shifts.
\newblock In {\em {SPAA}}, pages 196--203. {ACM}, 2013.

\bibitem[MVV87]{MulmuleyVV87}
Ketan Mulmuley, Umesh~V. Vazirani, and Vijay~V. Vazirani.
\newblock Matching is as easy as matrix inversion.
\newblock {\em Comb.}, 7(1):105--113, 1987.

\bibitem[Nan24]{Nanongkai24}
Danupon Nanongkai.
\newblock Cross-paradigm graph algorithms (invited talk).
\newblock In {\em {ICALP}}, volume 297 of {\em LIPIcs}, pages 3:1--3:1. Schloss Dagstuhl - Leibniz-Zentrum f{\"{u}}r Informatik, 2024.

\bibitem[New91]{Newman91}
Ilan Newman.
\newblock Private vs. common random bits in communication complexity.
\newblock {\em Inf. Process. Lett.}, 39(2):67--71, 1991.

\bibitem[NI92]{NagamochiI92}
Hiroshi Nagamochi and Toshihide Ibaraki.
\newblock A linear-time algorithm for finding a sparse k-connected spanning subgraph of a k-connected graph.
\newblock {\em Algorithmica}, 7(5{\&}6):583--596, 1992.

\bibitem[Par20]{parter2020distributed}
Merav Parter.
\newblock Distributed planar reachability in nearly optimal time.
\newblock In {\em 34th International Symposium on Distributed Computing (DISC 2020)}. Schloss Dagstuhl-Leibniz-Zentrum f{\"u}r Informatik, 2020.

\bibitem[Pod73]{Podderyugin1973algorithm}
VD~Podderyugin.
\newblock An algorithm for finding the edge connectivity of graphs.
\newblock {\em Vopr. Kibern}, 2:136, 1973.

\bibitem[PP22]{parter2022near}
Merav Parter and Asaf Petruschka.
\newblock Near-optimal distributed computation of small vertex cuts.
\newblock In {\em 36th International Symposium on Distributed Computing (DISC 2022)}, 2022.

\bibitem[PT11]{PritchardT11}
David Pritchard and Ramakrishna Thurimella.
\newblock Fast computation of small cuts via cycle space sampling.
\newblock {\em {ACM} Trans. Algorithms}, 7(4):46:1--46:30, 2011.

\bibitem[RGH{\etalchar{+}}22]{RozhonGHZL22}
V{\'{a}}clav Rozhon, Christoph Grunau, Bernhard Haeupler, Goran Zuzic, and Jason Li.
\newblock Undirected (1+\emph{{\(\epsilon\)}})-shortest paths via minor-aggregates: near-optimal deterministic parallel and distributed algorithms.
\newblock In {\em {STOC}}, pages 478--487. {ACM}, 2022.

\bibitem[RST14]{RackeST14}
Harald R{\"{a}}cke, Chintan Shah, and Hanjo T{\"{a}}ubig.
\newblock Computing cut-based hierarchical decompositions in almost linear time.
\newblock In {\em {SODA}}, pages 227--238. {SIAM}, 2014.

\bibitem[San07]{Sankowski07}
Piotr Sankowski.
\newblock Faster dynamic matchings and vertex connectivity.
\newblock In {\em {SODA}}, pages 118--126. {SIAM}, 2007.

\bibitem[Sch03]{schrijver2003combinatorial}
A.~Schrijver.
\newblock {\em Combinatorial Optimization: Polyhedra and Efficiency}.
\newblock Number v. 1 in Algorithms and Combinatorics. Springer, 2003.

\bibitem[She13]{Sherman13}
Jonah Sherman.
\newblock Nearly maximum flows in nearly linear time.
\newblock In {\em {FOCS}}, pages 263--269. {IEEE} Computer Society, 2013.

\bibitem[She17]{sherman2017area}
Jonah Sherman.
\newblock Area-convexity, l regularization, and undirected multicommodity flow.
\newblock In {\em Proceedings of the 49th Annual ACM SIGACT Symposium on Theory of Computing}, pages 452--460, 2017.

\bibitem[ST04]{SpielmanT04}
Daniel~A. Spielman and Shang{-}Hua Teng.
\newblock Nearly-linear time algorithms for graph partitioning, graph sparsification, and solving linear systems.
\newblock In {\em {STOC}}, pages 81--90. {ACM}, 2004.

\bibitem[SY22]{SaranurakY22}
Thatchaphol Saranurak and Sorrachai Yingchareonthawornchai.
\newblock Deterministic small vertex connectivity in almost linear time.
\newblock In {\em {FOCS}}, pages 789--800. {IEEE}, 2022.

\bibitem[Thu97]{Thurimella97}
Ramakrishna Thurimella.
\newblock Sub-linear distributed algorithms for sparse certificates and biconnected components.
\newblock {\em J. Algorithms}, 23(1):160--179, 1997.
\newblock Announced at PODC'95.

\bibitem[vdBGJdV25]{brand2025parallelminimumcostflow}
Jan van~den Brand, Hossein Gholizadeh, Yonggang Jiang, and Tijn de~Vos.
\newblock Parallel minimum cost flow in near-linear work and square root depth for dense instances, 2025.

\bibitem[vdBNS19]{BrandNS19}
Jan van~den Brand, Danupon Nanongkai, and Thatchaphol Saranurak.
\newblock Dynamic matrix inverse: Improved algorithms and matching conditional lower bounds.
\newblock In {\em {FOCS}}, pages 456--480. {IEEE} Computer Society, 2019.

\end{thebibliography}

\appendix
\section{Bipartite Matching in CONGEST and PRAM}\label{sec:BMM}

\newcommand{\lam}{\lambda}
\newcommand{\ma}{\boldsymbol{A}}
\newcommand{\mb}{\boldsymbol{B}}
\newcommand{\mzero}{\boldsymbol{0}}
\newcommand{\bx}{\bar{x}}
\newcommand{\by}{\bar{y}}
\newcommand{\R}{\mathbb{R}}
\newcommand{\dmcm}{d_{\mathrm{MCM}}}
\newcommand{\tmb}{\widetilde{\boldsymbol{B}}}
\newcommand{\LSSherman}{\mathsf{LowSpaceACMirrorProx}}
\newcommand{\norm}[1]{\left\lVert#1\right\rVert}
\newcommand{\x}{^\mathsf{x}}
\newcommand{\y}{^\mathsf{y}}
\newcommand{\diag}[1]{\textbf{\textup{diag}}\left(#1\right)}
\newcommand{\half}{\frac{1}{2}}
\newcommand{\Overflow}{\textup{Overflow}}
\newcommand{\tx}{\tilde{x}}
\subsection{Parallel BMM}

The following theorem is the state-of-the-art for parallel bipartite maximum matching.

\begin{theorem}[\cite{brand2025parallelminimumcostflow}]\label{thm:pram-bmm-fast}
    In \pram{} model, bipartite maximum matching can be solved in $\tilde{O}(m+n^{1.5})$ work and $\tilde{O}(n^{1/2})$ depth, with high probability.
\end{theorem}

However, we cannot plug \Cref{thm:pram-bmm-fast} to \Cref{thm:prammain} since the work function is not a super additive function on $m$. Instead, if we want the work to be super-additive in $m$, the following theorem from \cite{LeeS14} is the state-of-the-art.

\begin{theorem}[\cite{LeeS14}]\label{thm:pram-bmm}
    In \pram{} model, bipartite maximum matching can be solved in $\tilde{O}(mn^{0.5})$ work and $\tilde{O}(n^{1/2})$ depth, with high probability.
\end{theorem}

\subsection{A \congest{} Algorithm}
In this section, we will prove the following two theorems. Notice that \Cref{thm:pram-bmm-old} is subsumed by \Cref{thm:pram-bmm}, but we put it here as an alternative algorithm.

\begin{theorem}\label{thm:congest-bmm}
    In \congest{} model, maximum bipartite matching can be solved in $\tilde{O}(n^{3/4}D^{1/2} + n^{1/2}D + n^{7/10+o(1)}D^{7/10})$ rounds, with high probability.
In particular, when $D = O(n^{3/7 - \epsilon})$ for some constant $\epsilon >0$, this is sublinear in~$n$.
\end{theorem}

\begin{theorem}\label{thm:pram-bmm-old}
    In \pram{} model, bipartite maximum matching can be solved in $\tilde{O}(mn^{0.5+\epsilon})$ work and $n^{1+o(1)-\epsilon}$ depth for every $\epsilon\in [0,0.25]$, with high probability.

\end{theorem}

The main observation is that a semi-streaming $(1-\epsilon)$-approximation algorithm of \cite{assadi2022semi} almost straightforwardly can be implemented in the \congest{} and \pram{} models in roughly $\frac{1}{\epsilon}$ rounds/depth. Then, after computing a $(1-\epsilon)$ maximum bipartite matching, we combine with known reachability algorithms \cite{CaoF23,LiuJS19} in \congest{} and \pram{} to find the remaining $\epsilon n$ many augmenting paths one at a time. Hence, this section does not contain much technical novelty, but rather show how known algorithms can almost straightforwardly be generalized to other models and combined in order to show not yet document results for \congest{} (\cref{thm:congest-bmm}) and \pram{} (\cref{thm:pram-bmm}).

\paragraph{Related Work: Bipartite Matching in \congest{} and \pram{}.}
While bipartite matching has been extensively studied in many models, the exact version seems to resist sublinear (i.e.~ sublinear when $D = n^{1-\delta}$ for constant $\delta > 0$) algorithms in the \congest{} model. The work of \cite{AhmadiKO18} observed that a simple extension of the Hopcroft-Karp algorithm can solve bipartite matching in $O(M^*\log M^*)$ round in \congest{}, where $M^*$ is the size of the maximum matching (i.e., in general their algorithm runs in $O(n\log n)$ rounds).
The work of \cite{ForsterGLPSY21} takes a different approach based on Laplacian solves,
and show how unit-capacitated maximum flow (and hence the special case of bipartite matching) can be solved in $O(m^{3/7+o(1)}(n^{1/2}D^{1/4}+D))$ rounds in \congest{}.
This, to our knowledge, was the only known sublinear algorithm in \congest{}, although only in sparse graphs $m = \tilde{O}(n)$ and when $D = n^{2/7-\delta}$ for some constant $\delta > 0$.
 We note that the running time of our \cref{thm:congest-bmm} is an improvement for any regime of $m$ and $D$.
 We also note that \cite{Vos23} generalized the min-cost max-flow approach of \cite{LeeS14} in \cref{thm:pram-bmm} to also work in the \congest{} model, handling each of the $n^{1/2}$ rounds of solving a laplacian system in $O(\sqrt{n}+D)$ rounds in \congest{} (thus their total number of rounds is not sublinear in $n$).

There has also been a lot of work on bipartite maximum matching algorithms in the \pram{} model, or the more general min-cost max-flow problems (see, e.g., \cite{LeeS14, brand2025parallelminimumcostflow}).
It is known that the problem is in randomized NC
(i.e., polylog depth and polynomial work) due to a reduction to computing matrix inverse \cite{Lovasz79,KarpUW86,MulmuleyVV87,galil1988improved}.
Another sequence of work focuses on sublinear deterministic algorithms \cite{GoldbergPST92, GoldbergPV93}.

There are also many \pram{}, \congest{}, and streaming algorithms for $(1-\epsilon)$-approximating maximum bipartite matching, see e.g., \cite{cohen1995approximate,AssadiLT21,assadi2022semi,eggert2012bipartite,AhmadiKO18,lotker2015improved}, however most of these have depth/round/pass complexity at least $\frac{1}{\epsilon^2}$. A crucial property of \cite{assadi2022semi} is that it's running time only depends on a single factor of $\frac{1}{\epsilon}$, which is crucial to allow the exact algorithm to have sublinear depth/round/pass when combined with fast reachability algorithms in the different models.

\paragraph{Semi-Streaming to Other Models.}
The semi-streaming algorithm of \cite{assadi2022semi} will, when implemented in \congest{} and \pram{}, imply the below lemmas for approximating bipartite matching.
Since, to our knowledge, this connection has not been observed before, we sketch the proof in \cref{sec:generalizing_ajjst}.

\begin{lemma}\label{lem:congest-bmm-stream}
    In \congest{} model, there is an algorithm that computes a $(1-\epsilon)$-approximate fractional matching to the maximum bipartite matching problem in $\tilde{O}(\frac{D}{\epsilon})$ round.
\end{lemma}

\begin{lemma}\label{lem:pram-bmm-stream}
    In \pram{} model, there is an algorithm that computes a $(1-\epsilon)$-approximate fractional matching for the maximum bipartite matching problem in $\tilde{O}(\frac{1}{\epsilon})$ depth and $\tilde{O}(\frac{m}{\epsilon})$ work.
\end{lemma}

\begin{remark}
For the \congest{} algorithm, the $\frac{1}{\epsilon}$-dependency in \cref{lem:congest-bmm-stream} is tight for low-diameter graphs where $D = O(\log(n))$, as shown by a corresponding lower-bound in \cite{AhmadiKO18}. Since there are several algorithms with $\tilde{O}(\frac{1}{\epsilon^2})$ complexity without dependency on $D$, e.g., \cite{AhmadiKO18,AssadiLT21,lotker2015improved}, a natural question that remains open is if there is a $\tilde{O}(\frac{1}{\epsilon})$-round \congest{} algorithm for $(1-\epsilon)$-approximate bipartite matching.
\end{remark}

Before proving these lemmas, we show how they imply \cref{thm:congest-bmm} and \cref{thm:pram-bmm}, by combining with known techiniques. Our plan is to first round the $(1-\epsilon)$-approximate fractional matching to integral, and then find the remaining augmenting paths one-by-one using state of the art reachability algorithms in the two models.

\paragraph{The \congest{} Model.}
In the \congest{} model, \cite{AhmadiKO18} showed how to almost losslessly round a fractional matching to an integral one. In particular, they showed the following theorem.
\begin{theorem}[\cite{AhmadiKO18}]
\label{thm:congest-bmm-rounding}
Let \( G = (V, E) \) be a bipartite graph with maximum degree $\Delta$, \( y\in \mathbb{R}^{E} \) be a fractional matching of \( G \), and \( \epsilon > 0 \) be a parameter. There is a deterministic \( O\left(\frac{\log^2(\Delta / \epsilon) + \log^* n}{\epsilon}\right) \)-time \congest{} algorithm that computes a matching \( M \) of \( G \) such that $|M| \ge (1-\epsilon)||y||_{1}$.
\end{theorem}

\Cref{thm:congest-bmm-rounding} together with
\cref{lem:congest-bmm-stream} implies that we can find,
in $\tilde{O}(\frac{D}{\epsilon})$ rounds, an integral matching $M$ that is at most $\epsilon n$ far away from maximum. We will find the remaining $\le \epsilon n$ augmenting paths one by one, by using the current fastest directed reachability algorithm in \congest{}.
Indeed, note that given an integral matching $M$ on a bipartite graph $G = (L\cup R, E)$, one can put directions on the edges as follows: $e\in E$ is oriented from $R$ to $L$ if $e\in M$, and otherwise oriented from $L$ to $R$. Now, increasing the size of the matching $M$ by one is equivalent to finding any directed path from an unmatched vertex in $L$ to an unmatched vertex in $R$---such a path is called an augmenting path.

We use the single source shortest path algorithm of \cite{CaoF23}, setting all edge-weights to $1$.

\begin{theorem}[Distributed Reachability~\cite{CaoF23}]
There is a distributed algorithm that, given an \( n \)-node directed graph with positive integer edge weights of undirected hop-diameter \( D \), solves the single-source-shortest-paths problem with \( \tilde{O}(n^{1/2} + D + n^{2/5 + o(1)} D^{2/5}) \) rounds of communication in the \congest{} model, with high probability.
\end{theorem}

We need to run the reachability algorithm at most $\epsilon n$ times.
Hence, for solving exact bipartite matching, the total round complexity becomes:
\begin{align*}
\tilde{O}\left(
\frac{D}{\epsilon} + \epsilon n \cdot (n^{1/2} + D + n^{2/5 + o(1)} D^{2/5})
\right)
\end{align*}
By setting $\epsilon \coloneqq \frac{D^{1/2}}{n^{3/4} + n^{1/2}D^{1/2} + n^{7/10}D^{1/5}}$ to balance out the terms, the final time complexity is
\begin{align*}
\tilde{O}\left(
n^{3/4}D^{1/2} + n^{1/2}D + n^{7/10+o(1)}D^{7/10}
\right).
\end{align*}
This proves \cref{thm:congest-bmm}.

\paragraph{The \pram{} Model.}
Our strategy in the \pram{} model is similar. First we round the fractional matching to an integral one using the following \pram{} rounding lemma.
\begin{theorem}[\cite{GoldbergPST92}]
\label{thm:pram-bmm-rounding}
Let \( G = (V, E) \) be a bipartite graph, \( y\in \mathbb{R}^{E} \) be a fractional matching of \( G \). There is a deterministic \( O(\log^2(n)) \)-depth, $O(m\log^2 n)$-work \pram{} algorithm that computes a matching \( M \) of \( G \) such that $|M| \ge \lceil\, ||y||_{1} \,\rceil$.
\end{theorem}

\Cref{lem:pram-bmm-stream,thm:pram-bmm-rounding} together thus implies a
$\tilde{O}(\frac{1}{\epsilon})$ depth and $\tilde{O}(\frac{m}{\epsilon})$ work find an integral matching $M$ that is at most $\epsilon n$ far away from maximum. Again, we will find the remaining $\le \epsilon n$ augmenting paths one by one, by using a known reachability algorithm.

\begin{theorem}[Parallel Reachability \cite{LiuJS19}]
There is a parallel algorithm that, given an \( n \)-node, \( m \)-edge directed graph, solves the single-source reachability problem with work \( \tilde{O}(m) \) and depth \( n^{1/2 + o(1)} \) with high probability in \( n \).
\end{theorem}
Although in \cite{LiuJS19} they only solve the decision version and not how to recover the path. However this has been observed (see comment below \cite[Proposition~4]{assadi2022semi}) that the algoritm can actually recover the path too.

We need to run the reachability algorithm at most $\epsilon n$ times.
The total complexity becomes:
\begin{align*}
&&
\tilde{O}\left(
\frac{1}{\epsilon} + \epsilon n \cdot n^{1/2+o(1)})
\right)
\text{ depth,}
&&
\text{and}
&&
\tilde{O}\left(
\frac{m}{\epsilon} + \epsilon n \cdot m
\right)
\text{ work}.
&&
\intertext{
By setting $\epsilon \coloneqq n^{-1/2-\delta}$ for any $\delta\in [0,0.25]$, we get:
}
&&
\tilde{O}\left(
n^{1+o(1)-\delta}
\right)
\text{ depth,}
&&
\text{and}
&&
\tilde{O}\left(
m\cdot n^{1/2 + \delta}
\right)
\text{ work}.
&&
\end{align*}
This proves \cref{thm:congest-bmm}, since $n \le 2m$ after removing possibly isolated vertices.

\subsection{Generalizing a Semi-Streaming Algorithm}
\label{sec:generalizing_ajjst}
We now show how to generalize the semi-streaming algorithm of \cite{assadi2022semi} to work in the \congest{} and \pram{} settings in order to prove \cref{lem:congest-bmm-stream,lem:pram-bmm-stream}.
We note that the algorithm of \cite{assadi2022semi} is based on \cite{sherman2017area},
but where \cite{assadi2022semi} shows that this algorithm can be implemented in low space. We note that for our purposes in \congest{} and \pram{}, the algorithm of \cite{sherman2017area} would directly work as these models do not care about the space complexity. However, we choose to work with \cite{assadi2022semi} as it it more straightforward to see how it can be implemeted in our models.
This section will skip over most of the proofs and intuition, as they can be found in \cite{assadi2022semi}, and instead focus on simply arguing why the exact algorithm used by \cite{assadi2022semi} also works in \congest{} and \pram{}.

The algorithm of \cite{assadi2022semi} works by formulating the maximum bipartite matching problem as a $\ell_1$-regression problem and then solving it using a $\frac{1}{\epsilon}$-round optimization method for box-simplex games.

Let $\Delta^{m} \coloneqq \{x\in \R^{m} : \norm{x}_1 = 1\}$ be the $m$-dimensional simplex.
In particular, consider the following box-simplex optimization problem, for some $n\times m$ matrix $\ma$,
and vectors $c\in \mathbb{\R}^{m}$ and $b\in \mathbb{\R}^{n}$.
\begin{equation}\label{eq:boxsimplex}
	\min_{x \in \Delta^m} \max_{y \in [-1, 1]^n} y^\top\ma^\top x + c^\top x -b^\top y,
\end{equation}
which is equivalent to the following $\ell_1$-regression problem:
\begin{equation}\label{eq:boxsimplexprimal}
\min_{x \in \Delta^m} c^\top x + \norm{\ma^\top x - b}_1.
\end{equation}

We restate \cite[Algorithm~3]{assadi2022semi} here as \cref{alg:lowspacesherman}, which, according to \cite[Proposition~1]{assadi2022semi} outputs $\bar{x}$ that is a $\varepsilon$-approximate minimizer to the box-simplex game in \eqref{eq:boxsimplex}.

\begin{algorithm}[ht!]
	\DontPrintSemicolon
	\textbf{Input:} $\ma \in \R^{m \times n}_{\ge 0}$, $c \in \R^m$, $b \in \R^n$, $0 \le \varepsilon \le \norm{\ma}_\infty$\;
	\textbf{Output:} $\{v'_t\}_{0 \le t < T} \subset \R^n$, $\{u'_t\}_{0 \le t < T} \subset \R^n$, $\{\lam'_t\}_{0 \le t < T} \subset \R$, $\{y'_t\}_{0 \le t < T} \subset \R^n$ so that for
		\[\by \coloneqq \frac{1}{T}\sum_{0 \le t < T} y'_t,\; \bx \coloneqq \frac{1}{T}\sum_{0 \le t < T} \frac{\exp(\ma v'_t + |\ma|u'_t +  \lam'_t c)}{\norm{\exp(\ma v'_t + |\ma|u'_t + \lam'_t c)}_1},\]
		the pair $(\bx, \by)$ is an $\varepsilon$-approximate saddle point to \eqref{eq:boxsimplex}
  \;
	$T \gets O(\tfrac{\norm{\ma}_\infty \log m}{\varepsilon})$, $K \gets O(\log \tfrac{n\norm{\ma}_\infty}{\varepsilon})$\;
	$t \gets 0$, $\lam_0 \gets 0$, $v_0 \gets \mzero_n$, $u_0 \gets \mzero_n$, $y_0 \gets \mzero_n$\;
	\While{$t < T$}{
		$(v^{(0)}, u^{(0)}, \lam^{(0)}, y^{(0)}) \gets (v_t, u_t, \lam_t, y_t)$\;
		$\gamma\y \gets \tfrac 1 3 (b - \ma^\top \tfrac{\exp(\ma v_t + |\ma| u_t + \lam_t c)}{\norm{\exp(\ma v_t +  |\ma| u_t + \lam_t c)}_1}) - 2\diag{y_t} |\ma|^\top \tfrac{\exp(\ma v_t + |\ma| u_t + \lam_t c)}{\norm{\exp(\ma v_t + |\ma| u_t + \lam_t c)}_1}$\;
		\For{$0 \le k < K$}{
			$v^{(k + 1)} \gets -\tfrac{1}{10\norm{\ma}_\infty} (\tfrac 1 3 y_t  - 10\norm{\ma}_\infty v^{(0)})$\;
			$u^{(k+1)} \gets -\tfrac{1}{10\norm{\ma}_\infty} (- (y^{(0)})^2 + (y^{(k)})^2-10\norm{\ma}_\infty u^{(0)})$\;
			$\lam^{(k + 1)} \gets -\tfrac{1}{10\norm{\ma}_\infty} (\tfrac 1 3 - 10\norm{\ma}_\infty \lam^{(0)})$\;
			$d^{(k + 1)} \gets 2|\ma|^\top \tfrac{\exp(\ma v^{(k + 1)}+ |\ma| u^{(k + 1)} + \lam^{(k + 1)} c)}{\norm{\exp(\ma v^{(k + 1)}+ |\ma| u^{(k + 1)} +\lam^{(k + 1)} c)}_1}$\;
			$y^{(k + 1)} \gets \textup{med}(-1, 1, -\tfrac{\gamma\y}{d^{(k + 1)}})$ entrywise\; }
		$(v'_t,u'_t, \lam'_t, y'_t) \gets (v^{(K)}, u^{(K)}, \lam^{(K)}, y^{(K)})$\;
		$(v^{(0)}, u^{(0)}, \lam^{(0)}, y^{(0)}) \gets (v_t, u_t, \lam_t, y_t)$\;
		$\gamma\y \gets \tfrac 1 3 (b - \ma^\top \tfrac{\exp(\ma v'_t + |\ma| u'_t + \lam'_t c)}{\norm{\exp(\ma v'_t + |\ma| u'_t + \lam'_t c)}_1}) - 2\diag{y_t} |\ma|^\top \tfrac{\exp(\ma v_t + |\ma| u_t + \lam_t c)}{\norm{\exp(\ma v_t + |\ma| u_t + \lam_t c)}_1}$\;
		\For{$0 \le k < K$}{
			$v^{(k + 1)} \gets -\tfrac{1}{10\norm{\ma}_\infty} (\tfrac 1 3 y'_t - 10\norm{\ma}_\infty v^{(0)})$\;
			$u^{(k+1)} \gets -\tfrac{1}{10\norm{\ma}_\infty} (- (y^{(0)})^2 + (y^{(k)})^2-10\norm{\ma}_\infty u^{(0)})$\;
			$\lam^{(k + 1)} \gets -\tfrac{1}{10\norm{\ma}_\infty} (\tfrac 1 3 - 10\norm{\ma}_\infty \lam^{(0)})$\;
			$d^{(k + 1)} \gets 2|\ma|^\top \tfrac{\exp(\ma v^{(k + 1)} + |\ma| u^{(k + 1)}+ \lam^{(k + 1)} c)}{\norm{\exp(\ma v^{(k + 1)} + |\ma| u^{(k + 1)}+ \lam^{(k + 1)} c)}_1}$\;
			$y^{(k + 1)} \gets \textup{med}(-1, 1, -\tfrac{\gamma\y}{d^{(k + 1)}})$ entrywise \;
		}
		$(v_{t + 1}, u_{t+1}, \lam_{t + 1}, y_{t + 1}) \gets (v^{(K)}, u^{(K)}, \lam^{(K)}, y^{(K)})$\;
		$t \gets t + 1$\;
	}
	\caption{$\LSSherman(\ma, b, c, \varepsilon)\quad$ \cite{assadi2022semi}}\label{alg:lowspacesherman}
\end{algorithm}

The plan is as follows:
\begin{itemize}
\item Show how to set up bipartite matching problem as a box-simplex game in \eqref{eq:boxsimplex}, by choosing appropriate $\ma, b, c$ as in \cite[Section~4]{assadi2022semi}.
\item Show how \cref{alg:lowspacesherman} (and a simple post-processing step \cref{alg:removeOverflow}) can be implemented in both \congest{} and \pram{} efficiently, for the parameters given by the bipartite matching problem.
\end{itemize}

\paragraph{Setting up Bipartite Matching as a Box-Simplex Game.}
Let $G = (V,E)$ be a bipartite graph we wish to solve bipartite matching in. Let $M^{*}$ be the size of the maximum matching.
We follow \cite[Section~4]{assadi2022semi} to set it up as a box-simplex game, see their article for more details and discussion.

Let $M\in \mathbb{Z}$ be a $2$-approximation of the size of the maximum matching, i.e., it satisfies $M\le M^{*}\le 2M$.
There are multiple ways to obain this in our models: for example we can guess $M = 2^{i}$ for every $i = 0\ldots \lfloor\log n\rfloor$, and run the algorithm for all guesses and output the best achieved matching. Alternatively,
we can find a maximal matching (hence a $2$-approximation) in $\tilde{O}(1)$ rounds in \congest{},
\footnote{In the \congest{} model, we would need an additional $O(D)$ rounds to calculate the size of this maximal matching and broadcast it to every vertex, however this is dominated by the running time $\tilde{O}(\frac{D}{\epsilon})$ we are aiming to prove.}
or in $\tilde{O}(1)$ depth and $O(m)$ work in \pram{}, by using Luby's folklore (randomized) maximal independent set algorithm (applied to the line graph of $G$).

Let $\mb$ be the $m\times n$ adjacency matrix of $G$.
Let $\tmb$ be the $(m+1)\times (n+2)$ matrix obtained by adding with one extra row (a fake edge), and two extra columns (two fake vertices), all filled with zeroes, to $\mb$. Let $\dmcm \in \{0,1\}^{n+2}$ be $1$ everywhere except $0$ in the indices corresponding to the two fake vertices.

Now consider the following $\ell_1$-regression problem:
\begin{equation}\label{eq:matchl1reg}
\begin{gathered}
\min_{x \in \Delta^{m+1}} \norm{\ma^\top x - b}_1, \text{ where }  \ma \coloneqq M\tmb,\; b \coloneqq\half \dmcm.
\end{gathered}
\end{equation}

Use \cref{alg:lowspacesherman} to solve it up to $\Theta(\epsilon M)$ additive accuracy to obtain a solution $x$. Set
$\tilde{x} \coloneqq 2Mx_E$, where $x_E$ is $x$ but we dropped the coordinate of the fake edge.
Now $\tilde{x}$ will almost be an approximate maximum fractional matching, however it is not necessarily feasible.
See \cite[Section~4.1]{assadi2022semi} for an argument why $\norm{\tilde{x}}_1 - \Overflow(\tilde{x}) \ge (1-\epsilon) M^{*}$, where $\Overflow(x) \coloneqq \sum_{v\in V} \max\left(0, \left(\sum_{u : (u,v)\in E} x_{uv}\right)-1\right)
= \norm{\left(\mb^\top x - d\right)_+}_1$ is the sum of overflow at each vertex.

So, like in \cite{assadi2022semi}, we need a simple post-processing step where we remove the overflow of $\tilde{x}$ to make it into a feasible matching without removing too much.
We restate a simplified version of \cite[Algorithm~4]{assadi2022semi} below as \cref{alg:removeOverflow}.

\begin{algorithm}[!ht]
\DontPrintSemicolon
\textbf{Input:}
		{Bipartite graph $G = (V, E)$ with incidence matrix $\mb$, $x \in \R^E_{\ge 0}$}\;
  \textbf{Output:}
		{$\tx \in \R^E_{\geq 0}$ with $\tx \leq x$ entrywise, $\mb^\top \tx \leq \boldsymbol{1}$, and $\norm{\tx}_1 \geq \norm{x}_1 - \norm{\left(\mb^\top x - d\right)_+}_1$} (see \cite[Lemma~5]{assadi2022semi})\;
		$d^x \gets \mb^\top x$\;
		$f \gets (d^x - \boldsymbol{1})_{+}$\;
		$\tx_e \gets x_e\left(1 - \max\left(\tfrac{f_a}{d_a^x} , \tfrac{f_b}{d_b^x}\right)\right)$ for all $e = (a, b) \in E$ with $x_e > 0$\;
		\Return{$\tx$}
    \caption{$\mathsf{RemoveOverflow}(x, G)$ \cite{assadi2022semi}}\label{alg:removeOverflow}
\end{algorithm}

In particular, we get a fractional matching $x^{\star}\coloneqq \mathsf{RemoveOverflow}(\tx,G)$ of size $\norm{x^{\star}} > (1-\epsilon)M$.

\paragraph{Implementations in \congest{} and \pram{}.}
What remains is to argue that \cref{alg:lowspacesherman,alg:removeOverflow} can be implemented efficiently in \congest{} and \pram{}.

\Cref{alg:removeOverflow} can easily be implemented in $O(1)$ rounds in \congest{}, as all it does is compute the overflow at each vertex, and then decrease the value of the fractional matching on an edge $(u,v)$ based on the overflow at its endpoints. Similarly, it can be computed in $\tilde{O}(1)$ depth and $O(m)$ work in \pram{}.

\cref{alg:lowspacesherman} is a bit more complicated. It runs in $T = O(\frac{\norm{A}_{\infty} \log m}{\varepsilon})$ phases, which is $\tilde{O}(\frac{1}{\epsilon})$ as
$\ma\coloneq \tmb$ so $\norm{\tmb}_{\infty} = O(M)$ by construction, and we solve it up to error $\varepsilon \coloneq \epsilon M$ (note the different epsilons).
Similarly, note that $K = O(\log \frac{n\norm{A}_{\infty}}{\varepsilon}) = O(\log n)$. Hence, each line in the algorithm is run at most $\tilde{O}(\frac{1}{
\varepsilon})$ times.

We will show that each statement inside the loops in \cref{alg:lowspacesherman} can be implemented in $\tilde{O}(D)$ rounds in \congest{}, or $\tilde{O}(1)$ depth and $\tilde{O}(m)$ work in \pram{}.

The algorithm keeps tracks of many variables on supported on vertices.
It also computes intermediary values on the edges as expressions like:
		\[\frac{\exp(\ma v'_t + |\ma|u'_t +  \lam'_t c)}{\norm{\exp(\ma v'_t + |\ma|u'_t + \lam'_t c)}_1}\]
Computing matrix products $\ma v$ or $\ma^\top v$ for some vector $v$ can be done efficiently\footnote{One potential worry is numerical stability, as it might require many bits to represent the numbers. However, the algorithm is numerically stable, so values much smaller than the largest values can safely be rounded; see the comment in \cite{assadi2022semi} about numerical stability.} in parallel and locally, as $\ma$ is essentially the incidence matrix (the additional fake vertices and edges have corresponding zero-rows/columns, and hence can also be handled very efficiently). Similarly computing and multiplying by a diagonal matrix $\diag(y_t)$ is efficient to do locally and in parallel.
All the entrywise computations (like $\exp(v)$, or computing the entrywise median) is also efficiently locally/parallel computable.

The normilization, i.e., computing $\norm{v}_1$ (as in
$\norm{\exp(\ma v'_t + |\ma|u'_t + \lam'_t c)}_1$) is a bit more tricky for the \congest{} algorithm. It essentially boils down to computing the sum of values stored at each vertex, and requires $\tilde{O}(D)$ rounds to aggregate and sum up the value, and then broadcast it to all vertices.
In \pram{}, this normalization step is still easy to do in $\tilde{O}(m)$ work and $\tilde{O}(1)$ depth.
\begin{remark}
We note that this is the main bottleneck in the implementation of the \congest{} algorithm, and why it depends on $D$. An interesting open question is if there is an alternative $(1-\epsilon)$-approximation algorithm in \congest{} that only needs $\tilde{O}(\frac{1}{\epsilon})$ rounds, i.e., without dependency on $D$.
\end{remark}

We conclude that, even though \cref{alg:lowspacesherman} might look initially intimidating, and it's analysis is non-trivial (see \cite{assadi2022semi}), it is straightforward to see how it can be implemented efficiently in \congest{} and \pram{} models when $\ma$ is (almost) the incidence matrix of a graph. In particular, it can be implemented in $\tilde{O}(D/\epsilon)$ rounds in \congest{}, or $\tilde{O}(\frac{1}{\epsilon})$ depth and $\tilde{O}(\frac{m}{\epsilon})$ work in \pram{} to compute (together with \cref{alg:removeOverflow}) an $(1-\epsilon)$-approximate fractional matching. This proves \cref{lem:congest-bmm-stream,lem:pram-bmm-stream}.

\subsection{Solving Bipartite Matching on a Subgraph in CONGEST}\label{subsec:subgraph}

In our reductions, we need to solve bipartite matching on a bipartite subgraph $H$ of the original (not-necessarily bipartite) graph $G$. In most models of computation, we can just focus on solving bipartite matching on $H$ directly, however, in the \congest{} model the diameter of $H$ might be much larger than the diameter of $G$, and we are in fact allowed to use edges of $G\setminus H$ for the purpose of communication. While most algorithms can easily be adapted to solve bipartite matching on a subgraph, we can in fact show a perfect reduction irregardless of the underlying algorithm.

\begin{lemma}
Suppose there is a $R(n,m,D)$ round \congest{} algorithm for solving maximum bipartite matching on $n$-vertex, $m$-edge, $D$-diameter bipartite graphs $G_{\text{bip}}$.
Then there is a $R(O(n),O(m),O(D))$ round \congest{} algorithm that on $n$-vertex, $m$-edge, $D$-diameter (not-necessarily bipartite) graph $G$, solves maximum bipartite matching on a given bipartite subgraph $H$ (of potentially much larger diameter).
\end{lemma}
\begin{proof}[Proof Sketch.]
For each vertex $v\in V$, make $5$ copies of it: $v^{(H)}, v^{(a)}, v^{(a')}, v^{(b)}, v^{(b')}$. Let $(A,B)$ be the vertex bi-partition of $V$ induced by the bipartite subgraph $H$. Add the following edges:
\begin{align*}
(u^{(H)},v^{(H)}) &\text{ for all } (u,v)\in E(H)
\\
(u^{(a)},v^{(b)})
\text{ and }
(u^{(b)},v^{(a)})
&\text{ for all } (u,v)\in E(G)
\\
(v^{(a)},v^{(a')})
\text{ and }
(v^{(b)},v^{(b')})
&\text{ for all } v\in V
\\
(v^{(H)},v^{(a)})
&\text{ for all } v\in A
\\
(v^{(H)},v^{(b)})
&\text{ for all } v\in B
\end{align*}
Let this graph be called $G'$.
That is
$G'[V^{(H)}]$ is $H$,
and  $G'[V^{(a)}\cup V^{(b)}]$ is the binary lift of $G$. It is easy to verify that the graph is bipartite, by construction. Moreover, the diameter of any connected component\footnote{If $G$ is connected, $G'$ will have two connected components if and only if $G$ is bipartite, and else it will have one connected component.} of $G'$ is at most $O(D)$. This is since, if $G$ is not bipartite, let $g$ be its shortest odd-cycle and note that $|g| \le 2D+1$, which means that any vertex $v$ on this cycle will have $v^{(a)}$ and $v^{(b)}$ at most distance $2D+1$ from each other in $G'$.

Any maximum matching $M_H$ of $H$ can be extended to a matching $M_{G'}$ of $G'$, of size $|M_{G'}| = |M_H|+2n$, by adding all the $(v^{(a)},v^{(a')})$ and
$(v^{(b)},v^{(b')})$ edges.
For the other direction, suppose that $M_{G'}$ is a maximum matching of $G'$.
We know that all vertices
$v^{(a)}$ and $v^{(b)}$ must be matched (else  an edge $(v^{(a)},v^{(a')})$ can be added to the matching). We can without loss of generality change the matching edge of  $v^{(a)}$ to be matched to $v^{(a')}$ (and similar for
$(v^{(b)},v^{(b')})$).
Then
let $M_{H}
= M_{G'} \setminus (\cup_{v\in V} \{(v^{(a)},v^{(a')}),
(v^{(b)},v^{(b')})\})$, and note that $M_H$ is a matching on $H$ of size $|M_{G'}|-2n$.

Thus, solving the matching problem on $G'$ and on $H$ is essentially equivalent.
The \congest{} algorithm running on $G$ will simulate the bipartite matching \congest{} algorithm running on each connected component of $G'$,
where each node $v$ simulates the work of all its copies.
\end{proof}

\section{Graph Sparsification} \label{sec:sparsfy in congest}

We fix source $x \in C$ throughout this section.

 For all node $u \in C$, we denote $\mu_C(u) :=  \min_{L' \subseteq C : u \in L'} |N_G(L')|.$

\begin{definition} [$G_x$] \label{def:G_x}
Given a source $x \in C$ where $C \subseteq V$ in a graph $G = (V,E)$, we define the \emph{local graph} $G_x$ as the induced graph $G[N_G[C]]$ after adding a sink $t$ and the set of edges $(v,t)$ for all $v \in N_G(C)$.
\end{definition}

\begin{lemma} \label{lem:correct G_x}
     If $(L,S,R)$ is a minimum $(x,t)$-separator in $G_x$ where $x \in L, t \in R$, then $L$ is a minimizer of $\mu_C(x) = \min_{L' \subseteq C':x \in L'} |N_G(L')|$ and vice versa. In particular,  $\kappa_{G_x}(x,t) = \mu_C(x)$.
\end{lemma}
\begin{proof}
    First, we prove that $\kappa_{G_x}(x,t) \geq  \mu_C(x).$ Let $(L,S,R)$ be an $(x,t)$-min-separator in $G_x$ where $x \in L, t \in R$. We claim that $L \subseteq C$. Suppose $L$ contains a vertex $v \in N(C)$. This means $N_{G_x}(L)$ contains $t$, and thus $t \not \in R$, a contradiction. Observe that $N(L) = S$ because since $(L,S,R)$ is a minimum $(x,t)$-separator. Since $x \in L \subseteq C$ and $N(L) = S$, we have \[ \kappa_{G_x}(x,t) = |S| = |N_{G_x}(L)| =  |N_{G}(L)|  \geq  \min_{L' \subseteq C':x \in L'} |N_G(L')| = \mu_C(x).\]

    Next, we prove $ \mu_C(x) \geq \kappa_{G_x}(x,t).$  Let $L^*$ be a minimizer of $\min_{L' \subseteq C':x \in L'} |N_G(L')|$. Observe that  $N_G(L^*) \subseteq N_G[C]$. We now prove that $N_G(L^*)$ is an $(x,t)$-separator in $G_x$. Since $N_G(L^*)$ is a separator in $G$ and $N_G[L^*] \subseteq N_G[C]$, it is enough to prove that the set of edges from $t$ to every vertex in $N_G(C)$ will not destroy this cut. Indeed, since $L^* \subseteq C$, the new edges will not incident to $L^*$. Therefore, $N_G(L^*) = N_{G_x}(L^*)$ is an $(x,t)$-separator in $G_x$. So,
    \[  \mu_C(x) = |N_{G}(L^*)| = |N_{G_x}(L^*)| \geq \kappa_{G_x}(x,t). \]
    This concludes the proof.
\end{proof}

We cannot yet implement $G_x$ in a distributed network because connecting $s$ to every node in $N_G(C)$ is a global operation. Instead, we simulate $s$ by adding a perfect matching from each node in $N_G(C)$ to a new node (so we can remove $s$).

\begin{definition}[$G'_x$]
    Given $G_x$, define $G'_x$ as follows. First, remove $t$ from $G'_s$. Then, for each vertex $u \in N_G(C)$, add a new vertex $y$ along with an edge $(u,y)$. Denote $Y$ as the set of new vertices.
\end{definition}
\begin{lemma} \label{lem:cut-equivalent terminal graph}
 Every $(s,t)$-cut in $G_x$ is an $(s,Y)$-cut in $G'_x$ and vice ver sa. In particular, $\kappa_{G_x}(s,t) = \kappa_{G'_x}(s,Y)$.
\end{lemma}
\begin{proof}
If $S$ is an $(s,t)$-separator in $G'_s$, then denote a vertex cut $(L,S,R)$ in $G'_s$ where $s \in L, t \in R$. Since $N_{G'_s}(t) \subseteq S \cup R$,   every vertex in $Y$ in $G''_s$ does not have an edge to $L$. Therefore, $(L, S, (R \cup Y) - \{t\})$ is a vertex cut in $G''_s$. If $S$ is an $(s,Y)$-separator in $G''_s$, then the same separator must be an $(s,t)$-separator in $G'_s$ by contracting all $Y$ into $t$.
\end{proof}
\begin{corollary} \label{cor:xY separator congest}
     If $(L,S,R)$ is a minimum $(x,Y)$-separator in $G'_x$ where $x \in L, Y \subseteq R$, then $L$ is a minimizer of $\min_{L' \subseteq C':x \in L'} |N_G(L')|$. In particular, $\kappa_{G'_x}(x,Y) = \min_{L' \subseteq C':x \in L'} |N_G(L')|$.
\end{corollary}

\begin{definition} [Sparsification] \label{def:tilde H x}
    Give $G_x$, define $\tilde H_x$ after the following reductions.
    \begin{enumerate} [noitemsep]
        \item Remove $Z := N_G(x) \cap N_G(t)$ from $G_x$.
        \item Remove every edge $(u,v) \in E_G(N_G(x),N_G(x)) \cup E_G(N_G(t),N_G(t))$.
    \end{enumerate}
\end{definition}

\begin{lemma}
    $\kappa_{G_x}(x,t) = \kappa_{\tilde H_x}(x,t) + |Z|$.  Furthermore, If $S$ is an $(s,t)$-min-separator in $H$, then $Z \cup S$ is an $(s,t)$-min-separator in $G$.
\end{lemma}
\begin{proof}
  The first step follows because every vertex $v \in N_G(x) \cap N_G(t)$  belongs to every $(x,t)$-separator in $G$. The second step also follows from the same reduction rules in \cite{LiNPSY21}.
\end{proof}

Finally, we define the graph $H_x$ as follows.
\begin{definition} [$H_x$] \label{def: H x}
    Give $\tilde H_x$, define $H_x$ as follows. First, remove $t$ from $\tilde H_x$. Then, for each vertex $u \in N_G(C)$, add a new vertex $y$ along with an edge $(u,y)$. Denote $Y$ as the set of new vertices.
\end{definition}

\begin{lemma} \label{lem:cut-equivalent terminal graph H_x}
 Every $(s,t)$-cut in $G_x$ is an $(s,Y)$-cut in $G'_x$ and vice ver sa. In particular, $\kappa_{G_x}(s,t) = \kappa_{G'_x}(s,Y)$.
\end{lemma}
\begin{proof}
    The proof is similar to \Cref{lem:cut-equivalent terminal graph}.
\end{proof}
We are ready to prove \Cref{lem:sparsify G in congest}.
\begin{proof} [Proof of \Cref{lem:sparsify G in congest}]
Given $x \in C$, we construct $H_x$ as defined in \Cref{def: H x}. We can verify that this construction in \Cref{def:tilde H x} can be done in $O(1)$ rounds. We can compute $Z$ in two rounds. First, $s$ and $t$ broadcast its id along its neighbors, and in the second round, if a node receives both $s$ and $t$ then it belongs to $Z$. Every node $v \in N_G(C) - Z$ can simulate another node in $Y$ in the matching by itself. Removing every edge $(u,v) \in E_G(N_G(x),N_G(x)) \cup E_G(N_G(C),N_G(C))$ can be done in $O(1)$ rounds.

We verify each of the following properties. First we show that for every vertex $v \in C, \textdeg_{H_x}(v,V_x - (N_G(x) \cap C)) \leq O(2^{\log^{0.7}n}\cdot \ell)$. This follows from \Cref{eq:cn diff prop} and the fact that the common neighbors $Z$ were removed. Next, $\kappa_{H_x}(x,Y_x) +|Z_x| \geq \min_{L' \subseteq C: L' \neq \emptyset}|N_G(L')|$.
the fact that if $x \in L$, then $\kappa_{H_x}(x,Y_x) + |Z_x| = \min_{L' \subseteq C: L' \neq \emptyset}|N_G(L')|$, can be derived from from \Cref{lem:remove z,cor:xY separator congest,lem:cut-equivalent terminal graph H_x}.
\end{proof}

\end{document}